\documentclass{article}

\usepackage{xspace}
\usepackage{amsmath}
\usepackage{amssymb}
\usepackage{amsthm}
\usepackage{mfirstuc}
\usepackage{tikzsymbols}
\usepackage{MnSymbol}
\usepackage{paralist}
\usepackage{tikz}
\usepackage{url}
\usetikzlibrary{quantikz}
\usetikzlibrary{positioning,calc,arrows}
\usetikzlibrary{decorations.text}
\usetikzlibrary{decorations.pathmorphing}
\tikzset{onslide/.code args={<#1>#2}{%
		\only<#1>{\pgfkeysalso{#2}} 
}}
\tikzset{
	invisible/.style={opacity=0},
	visible on/.style={alt={#1{}{invisible}}},
	alt/.code args={<#1>#2#3}{%
		\alt<#1>{\pgfkeysalso{#2}}{\pgfkeysalso{#3}} 
	},
}
\tikzstyle{inputNode}=[draw,circle,minimum size=10pt,inner sep=0pt]
\tikzstyle{stateTransition}=[-stealth, thick]

\usepackage{graphicx}

\newcommand{\iso}{\sim}
\newcommand{\acong}{\triplesim}
\newcommand{\sgr}{\preceq}
\newcommand{\sgrc}{\sgr_{\aclass}}

\newcommand{\agraph}{G}
\newcommand{\agraphB}{H}
\newcommand{\agraphC}{I}

\newcommand{\aclass}{{\cal C}}

\newcommand{\funs}{\Sigma}
\newcommand{\vars}{{\cal V}}
\newcommand{\ar}[1]{\mathit{ar}(#1)}
\newcommand{\terms}{{\cal T}}

\newcommand{\fv}[1]{\vars(#1)}
\newcommand{\id}{{\mathit id}}

\newcommand{\innernodes}[1]{{\cal N}^*_{#1}}

\newcommand{\nodes}[1]{{\cal N}_{#1}}
\newcommand{\allnodes}{{\cal N}}
\newcommand{\anode}{\alpha}
\newcommand{\anodeB}{\beta}

\newcommand{\inputs}[1]{{\cal R}_{#1}}

\newcommand{\edges}[1]{{\cal E}_{#1}}
\newcommand{\mkedge}[2]{\left(#1 \rightarrow #2\right)}
\newcommand{\anyedge}[2]{(#1 \rightleftarrows #2)}

\newcommand{\mkgraph}[4]{\tuple{#1,#2,#3,#4}}

\newcommand{\labs}[1]{l_{#1}}

\newcommand{\node}{node\xspace}

\newcommand{\edge}{edge\xspace}

\newcommand{\substitutable}{substitutable\xspace}
\newcommand{\graph}{graph\xspace}
\newcommand{\subgraph}{sub\graph}
\newcommand{\nsubstitution}{$\allnodes$-mapping\xspace}
\newcommand{\substitution}{substitution\xspace}
\newcommand{\merge}{merge\xspace}
\newcommand{\groot}{root\xspace}
\newcommand{\rootsimilar}{root-compatible\xspace}

\newcommand{\interv}[2]{[#1,#2]}
\newcommand{\dom}[1]{\mathit{dom}(#1)}
\newcommand{\img}[1]{\mathit{img}(#1)}
\newcommand{\tuple}[1]{\langle #1 \rangle}
\newcommand{\set}[1]{\{ #1 \}}
\newcommand{\setof}[2]{\left\{#1\,\middle|\:#2\right\}}
\newcommand{\isdef}{\stackrel{\scriptscriptstyle{\mathsf{def}}}{=}}
\newcommand{\restrict}[2]{#1|_{#2}}
\newcommand{\inv}[1]{#1^{-1}}

\newcommand{\replace}[3]{#1[#3/#2]}

\newtheorem{theorem}{Theorem}
\newtheorem{proposition}[theorem]{Proposition}
\newtheorem{lemma}[theorem]{Lemma}
\newtheorem{definition}[theorem]{Definition}
\newtheorem{corollary}[theorem]{Corollary}
\newtheorem{example}{Example}

\title{A Superposition-Based Calculus for Quantum Diagrammatic Reasoning and Beyond}
\author{Rachid Echahed and Mnacho Echenim\\ and Mehdi Mhalla and
  Nicolas Peltier\\  \\{Universit\'e Grenoble Alpes, Grenoble INP}\\ {CNRS, LIG, F-38000 Grenoble, France}
}
\date{March 2021}

\begin{document}
	
	\maketitle

        \begin{abstract}
          \noindent
          We introduce a class of rooted graphs which allows one to
          encode various kinds of classical or quantum circuits. We
          then follow a set-theoretic approach to define rewrite
          systems over the considered graphs and propose a new
          complete Superposition calculus which handles sets of
          formulas consisting of equations or
          disequations over these graphs.
        \end{abstract}
\section{Introduction}


Graph-based languages are 
a popular specification and verification tool, both in mathematics and computer science. They have become particularly useful in quantum computing, for the formal verification of quantum algorithms and protocols, and several such languages have been developed over the past years for such an analysis, including the ZW calculus \cite{ZW}, the ZH calculus \cite{ZH}, the ZX calculus \cite{ZX}, the SZX \cite{SZX} calculus, or more recently  the PBS-calculus  \cite{PBS}. The reason these languages are so useful is that standard approaches to verifying the correctness of quantum algorithms and protocols involve computations on complex matrices, and such approaches are non-intuitive and error-prone. On the other hand, graph-based languages permit to represent the same algorithms and protocols in a more intuitive way, by abstracting the numerical values and matrices, and replacing computations on these matrices by rewrite rules. Using graph-based languages to represent quantum algorithms thus permits to reduce many verification tasks to testing the equivalence of two graphs modulo a set of equations, which represent properties of the considered computations. It is also possible to test the equivalence between graphs in a user-friendly way using software such as, e.g., quantomatic\footnote{\url{https://quantomatic.github.io/}}. These tools permit to create graphs, define graph rewrite rules and apply them automatically, and even to  define simplification tactics. 
However, these tactics have to be specified by the user and are not necessarily complete.

A natural question is whether it is possible to find strategies that can be applied to graph rewrite rules in a purely automated way. 
Such rewrite strategies {must fulfill several requirements}:
\begin{inparaenum}[(i)]
	\item{they
		must be {efficient} (i.e., as restrictive as possible with regards to the number of rule applications);}
	\item{they must be {complete}, in the sense that any two equivalent graphs must always be reducible to the same representative; and}
	\item{they must be {generic}, because graph languages and the equations that can be applied are very diverse and constantly evolving.}
\end{inparaenum}
The {Superposition\footnote{There is no relation with quantum superposition.} calculus} \cite{BG94} is a proof procedure that meets these three requirements. This is a {highly efficient calculus that permits to reason automatically on first-order theories involving equalities.}
It may be viewed as an extension of the Knuth-Bendix algorithm \cite{Knuth-Bendix} to a wider class of equational formul{\ae}, involving disjunctions (e.g., conditional rules) and quantifiers.
It may also be viewed as an extension of the Resolution calculus  \cite{Rob65,LEI96}, which is the most efficient proof procedure for predicate logic,
to sets of clauses containing the equality predicate.
The calculus is presented as a set of inference rules, deducing new assertions from axioms or previously generated assertions, together with a generic redundancy criterion that permits to prune the search space and discard many  inferences.
This calculus is {complete}: if a  formula $\phi$ is deducible from a set of axioms $E$ then this is automatically detected by  the calculus. It is also generic, and can be uniformly applied  to any set of axioms 	that can be expressed in first-order logic.
{Most efficient theorem provers, such as Vampire \cite{RV01}, E \cite{Eprover} or Spass \cite{SPASS}, are based on the Superposition calculus.}

In this paper, we present an adaptation of the Superposition calculus to so-called \emph{graph rewrite systems}, i.e., to sets of equations involving graphs rather than first-order terms. 
The Superposition calculus cannot be extended to the entire class of graphs in a straightforward manner.
This is due to the fact that one of the fundamental properties that is required to prove completeness in the first-order case, namely confluence, does not hold for graphs. In fact, confluence on the class of graphs is not even decidable \cite{DBLP:conf/birthday/Plump05}. We thus exploit the fact that the graphs considered in the context of quantum computing are not arbitrary, and define a class that is general enough to encode the graph-based languages that are used for quantum diagrammatic reasoning.

The paper is decomposed as follows. After a recap on notions and notations, we introduce the class of graphs on which the extended Superposition calculus will be applied in Section \ref{sec:graphs}, along with a replacement operation that controls the way graphs can be rewritten. Formulas involving graphs are introduced in Section \ref{sec:lits}, as well as the way such formulas are interpreted. The extension of the Superposition calculus is defined in Section \ref{sec:pf-proc}, and this calculus is proved to be refutationally complete in Section \ref{sec:complete}. Section \ref{sec:circuits} illustrates the generality of our approach by considering the so-called subclass of \emph{circuits}, which is an interesting candidate to encode the graphs used in the ZX calculus.
\section{Basic Definitions and Notations}

\label{sect:basics}

For any partial function $f$, we denote by $\dom{f}$ its domain, i.e., the set of elements $x$ such that $f(x)$ is defined. For any function $f$ and for any $D \subseteq \dom{f}$, we denote by $\restrict{f}{D}$ the restriction of $f$ to $D$. The function of  domain $\{ t_1,\dots,t_n\}$ mapping $t_i$ to $s_i$ 
(for all $i \in \interv{1}{n}$) is denoted by $\{ t_i \mapsto s_i \mid i \in \interv{1}{n} \}$.

Any partial function $f$ operating on some set $S$ may be extended into a function 
operating on tuples or sets of elements in $S$, using the relations: 
$f(\tuple{t_1,\dots,t_n}) \isdef \tuple{f(t_1),\dots,f(t_n)}$
and
$f(\set{t_1,\dots,t_n}) \isdef \set{f(t_1),\dots,f(t_n)}$. 
If $t_i \not \in \dom{f}$ for some $i = 1,\dots,n$ then 
$f(\tuple{t_1,\dots,t_n})$
and 
$f(\set{t_1,\dots,t_n})$ are undefined.
These relations may be applied recursively, e.g., if $t_1,\dots,t_n$ are themselves sets or tuples of elements.

The notation $g \circ f$ denotes as usual the composition of $g$ and $f$, 
with $\dom{g \circ f} = \{ x \in \dom{f} \mid f(x) \in \dom{g}\}$
and $(g \circ f)(x) = g(f(x))$ for every $x \in \dom{g \circ f}$.
According to the previous convention, if functions are viewed as sets of pairs, then
for every function $f$ and for every injective function $g$ such that 
$\img{f} \cup \dom{f} \subseteq \dom{g}$,
$g(f)$ is a function with $\dom{g(f)} = g(\dom{f})$, $\img{g(f}) = g(\img{f})$, and
$(g(f))(g(x)) = g(f(x))$, for any $x \in \dom{f}$.
Note that, by definition, $g(f) = g \circ f \circ g^{-1}$.

With a slight abuse of notations, we sometimes use set notations on tuples, i.e., we write $u \in  \tuple{t_1,\dots,t_n}$ to state that 
$u \in \set{t_1,\dots,t_n}$, $\tuple{t_1,\dots,t_n} \cup E$ to denote the set
$\set{t_1,\dots,t_n} \cup E$ or $\tuple{t_1,\dots,t_n} = \tuple{t_1',\dots,t_m'} \cup \tuple{t_1'',\dots,t_k''}$ for 
$\{ t_1,\dots,t_n \} = \{ t_1',\dots,t_m',t_1'',\dots,t_k''\}$. 
 


\begin{definition}
\label{def:term}
Let $\funs$ be a set of {\em function symbols}  and let $\vars$ be a set of {\em variables}.
Each symbol $f \in \Sigma$ is associated with a unique arity
$\ar{f}$.
The set of {\em terms} $\terms$ is the least set such that 
$\vars \subseteq \terms$ and $t_1,\dots,t_n \in \terms \Rightarrow f(t_1,\dots,t_n)\in \terms$, for every $f\in \funs$ with $n = \ar{f}$.
We denote by $\fv{t}$ the set of variables occurring in $t$.
A term $t$ is {\em ground} if $\fv{t}= \emptyset$.
\end{definition} 

A {\em \substitution} $\sigma$ is a total mapping from $\vars$ to $\terms$.
As usual, for any term $t$, $\sigma(t)$ denotes the
term obtained from $t$ by replacing every variable $x$ by
$\sigma(x)$. 
A substitution $\sigma$ is {\em ground} if 
$\sigma(x)$ is ground for every $x\in \dom{\sigma}$. 
%
%
A {\em unifier} of a set of pairs $E$ is a \substitution $\sigma$ 
such that $\sigma(t) = \sigma(s)$ holds for all pairs $(t,s) \in E$.
$E$ is {\em unifiable} if it admits a unifier. 
A \substitution $\sigma$ is {\em more general than a \substitution $\theta$}
iff there exists a substitution $\sigma'$ such that $\theta = \sigma' \circ \sigma$.
It is well known that every unifiable set of pairs admits a most general unifier (mgu).

\section{\capitalisewords{graphs}}\label{sec:graphs}

\subsection{\capitalisewords{definitions}}

We define the class of graphs on which the adaptation of the Superposition calculus will be applied. Intuitively, this class consists of labeled graphs with a distinguished sequence of nodes which can be viewed as an interface that constrains the operations of replacing a subgraph by another graph that can be performed.


\newcommand{\sorts}{{\tt S}}
\newcommand{\asort}{{\tt s}}
\newcommand{\sort}{sort\xspace}
\newcommand{\leqnodes}{\trianglelefteq}
\newcommand{\lessnodes}{\vartriangleleft}
\newcommand{\eqnodes}{\triangleq}
\newcommand{\sortof}{\mathit{sort}}

Let $\allnodes$ be a fixed countable set of {\em {\node}s}, 
disjoint from the set of terms,
and let $\sorts$ be a set of {\em {\sort}s}.
We consider a function $\sortof$ mapping every node $\anode \in \allnodes$
to a \sort in $\sorts$ and a pre-order $\leqnodes$ on $\allnodes$.
We write $\anode \eqnodes \anode'$ if $\anode \leqnodes \anode'$
and $\anode' \leqnodes \anode$, and $\anode \lessnodes \anode'$ 
for $\anode \leqnodes \anode'$
and $\anode' \not \leqnodes \anode$.

\begin{definition}
\label{def:graph}
A {\em \graph} $\agraph$ is a
tuple $\mkgraph{\nodes{\agraph}}{\inputs{\agraph}}{\edges{\agraph}}{\labs{\agraph}}$
where:
\begin{itemize}
\item{$\nodes{\agraph}$ is a finite subset of $\allnodes$.}
\item{$\inputs{\agraph}$ is a finite sequence of {\node}s in $\nodes{\agraph}$, with no repetition, called the {\em {\groot}s of $G$}. We denote by $\innernodes{\agraph}$ the set of nodes occurring in $\nodes{\agraph}$ but not in $\inputs{\agraph}$.}

\item{$\edges{\agraph}$ is a finite set of ordered pairs (written $\mkedge{\anode}{\anodeB}$) of {\node}s in $\nodes{\agraph}$, called {\em {\edge}s}. We write $\anyedge{\anode}{\anodeB}$ to denote any edge $\mkedge{\anode}{\anodeB}$ or $\mkedge{\anodeB}{\anode}$}
\item{$\labs{\agraph}$ is a function mapping each \node $\anode \in \innernodes{\agraph}$  to a term (called the {\em label} of $\anode$).}
\end{itemize}
Two {\graph}s $\agraph$ and $\agraphB$ are {\em disjoint}
if $\nodes{\agraph} \cap \nodes{\agraphB} = \emptyset$.
We denote by $\fv{\agraph} \isdef \bigcup_{\anode\in\innernodes{\agraph}} \fv{\labs{\agraph}(\anode)}$ the set of variables occurring in the label of some node in $\agraph$.
A \graph $\agraph$ is {\em ground} if $\fv{\agraph} = \emptyset$.
\end{definition}

Note that the function $\labs{\agraph}$ is not defined over the
nodes in $\inputs{\agraph}$.

\begin{definition}
For any substitution $\sigma$ and \graph $\agraph$, $\sigma(\agraph)$ denotes the \graph $\mkgraph{\nodes{\agraph}}{\inputs{\agraph}}{\edges{\agraph}}{\labs{}'}$, where $\labs{}'(\anode) = \sigma(\labs{\agraph}(\anode))$, for all $\anode\in \innernodes{\agraph}$.	
\end{definition}


%
\begin{definition}
\label{def:subgraph}
A \graph $\agraphB$ is a {\em \subgraph}
of $\agraph$ (written $\agraphB \sgr \agraph$) if the following conditions hold:
\begin{enumerate}
\item{$\nodes{\agraphB} \subseteq \nodes{\agraph}$; \label{subgraph:nodes}}

\item{If $\mkedge{\anode}{\anodeB}\in \edges{\agraphB}$ then
$\mkedge{\anode}{\anodeB}\in \edges{\agraph}$. \label{subgraph:edges} 
}
\item{If $\mkedge{\anode}{\anodeB}\in \edges{\agraph}$ and  $\anode,\anodeB \in \nodes{\agraphB}$ 
then $\mkedge{\anode}{\anodeB} \in \edges{\agraphB}$. \label{subgraph:edges_bis} 
}

\item{If $\anyedge{\anode}{\anodeB} \in \edges{\agraph}$, $\anodeB \not \in \nodes{\agraphB}$
and $\anode\in \nodes{\agraphB}$, then $\anode$ occurs in $\inputs{\agraphB}$.  \label{subgraph:inputs}
}
\item{If $\anode$ occurs in $\inputs{\agraph}$ and $\anode \in \nodes{\agraphB}$, then 
$\anode$ occurs in $\inputs{\agraphB}$.  \label{subgraph:inputs_bis}
}
\item{$\labs{\agraphB}$ is the restriction of $\labs{\agraph}$ to the {\node}s in $\innernodes{\agraphB}$.\label{subgraph:labs}}

\end{enumerate}
\end{definition}

\begin{proposition}
The relation $\sgr$ is transitive and reflexive.
\end{proposition}
\begin{proof}
It is immediate to check that $\sgr$ is reflexive.
Assume that $\agraphC \sgr \agraphB \sgr \agraph$.
Then:
\begin{enumerate}
\item{$\nodes{\agraphC}  \subseteq \nodes{\agraphB} \subseteq \nodes{\agraph}$ hence 
$\nodes{\agraphC}  \subseteq  \nodes{\agraph}$.}
\item{We have $\mkedge{\anode}{\anodeB}\in \edges{\agraphC} \Rightarrow \mkedge{\anode}{\anodeB}\in \edges{\agraphB} \Rightarrow \mkedge{\anode}{\anodeB}\in \edges{\agraph}$.}
\item{If $\mkedge{\anode}{\anodeB}\in \edges{\agraph}$ and $\anode,\anodeB\in  \innernodes{\agraphC}$, then by Condition \ref{subgraph:inputs_bis} of Definition \ref{def:subgraph}
$\anode,\anodeB\in  \innernodes{\agraphB}$, thus $\mkedge{\anode}{\anodeB}\in \edges{\agraphB}$ and
$\mkedge{\anode}{\anodeB}\in \edges{\agraphC}$.}
\item{Assume that $\anyedge{\anode}{\anodeB} \in \edges{\agraph}$, $\anode\in \nodes{\agraphC}$ and $\anodeB \not \in \nodes{\agraphC}$. If 
$\anodeB \in \nodes{\agraphB}$, then $\anyedge{\anode}{\anodeB} \in \edges{\agraphB}$ and we deduce that $\anode$ must occur in $\inputs{\agraphC}$ because $\agraphC \sgr \agraphB$.
If
$\anodeB \not \in \nodes{\agraphB}$ then
$\anode$ occurs in $\inputs{\agraphB}$ since $\agraphB \sgr \agraph$.
Thus $\anode \in \inputs{\agraphC}$, since $\agraphC \sgr \agraphB$.
}

\item{$\labs{\agraphC} = \restrict{\labs{\agraphB}}{\innernodes{\agraphC}} = \restrict{(\restrict{\labs{\agraph}}{\innernodes{\agraphB}})}{\innernodes{\agraphC}} = \restrict{\labs{\agraph}}{\innernodes{\agraphC}}$ (because 
$\innernodes{\agraphC} \subseteq \innernodes{\agraphB}$).}
\end{enumerate}
\end{proof}


%


\begin{proposition}\label{prop:sub_inputs_include}
	If $\agraphB\sgr \agraph$ then $\inputs{\agraph} \subseteq \nodes{\agraph} \setminus \innernodes{\agraphB}$.
\end{proposition}

\begin{proof}
	Assume for a contradiction that $\anode \in \inputs{\agraph} \cap \innernodes{\agraphB}$. Then by definition, $\anode$ is a node in $\nodes{\agraphB}$ that does not occur in $\inputs{\agraphB}$, but this is impossible by Condition \ref{subgraph:inputs_bis} of Definition \ref{def:subgraph}.
\end{proof}

\newcommand{\amap}{\mu}

In what follows, we define an equivalence relation on graphs to formalize the fact that the properties we prove do not depend on the actual nodes that occur in the graph.

\begin{definition} 
An {\em \nsubstitution} $\amap$ is a partial injective mapping from {\node}s to {\node}s, 
such that $\sortof(\anode) = \sortof(\amap(\anode))$
and 
 $\anode \eqnodes \amap(\anode)$. 
For technical convenience we assume that
$\amap(t) = t$ holds for every term $t$, i.e., 
every \nsubstitution is extended to the identity on terms.
\end{definition}

\begin{proposition}
If $\agraph$ is a \graph, then for every \nsubstitution $\amap$ with $\dom{\amap} \supseteq \nodes{\agraph}$,
$\amap(\agraph)$ is  a \graph. Moreover, $\nodes{\amap(\agraph)} = \amap(\nodes{\agraph})$, 
$\edges{\amap(\agraph)} = \amap(\edges{\agraph})$, 
$\inputs{\amap(\agraph)} = \amap(\inputs{\agraph})$
and
$\labs{\amap(\agraph)} = \amap(\labs{\agraph}) = \amap \circ \labs{\agraph} \circ \amap^{-1} = \labs{\agraph} \circ \amap^{-1}$ 
\end{proposition}
\begin{proof}
The result follows immediately from the definition of $\amap(\agraph)$ (see Section \ref{sect:basics}) and from the fact that $\amap$ is injective. \end{proof}


Similarly, because {\nsubstitution}s are injective, we have the following result:

\begin{proposition}
	\label{prop:mapsgr}
	If $\agraphB \sgr \agraph$ and $\nodes{\agraph} \subseteq \dom{\amap}$ for an \nsubstitution $\amap$, then $\amap(\agraphB) \sgr \amap(\agraph)$.
\end{proposition}

\begin{definition}
We write $\agraph \iso \agraph'$ if there exists $\amap$ such that $\dom{\amap} \supseteq \nodes{\agraph'}$ and
$\agraph = \amap(\agraph')$.	
\end{definition}

\begin{proposition}
The relation $\iso$ is an equivalence relation.
\end{proposition}
\begin{proof}\hfill
\begin{itemize}
\item{{\bf Reflexivity.}
The identity on $\nodes{\agraph}$ is an \nsubstitution, thus $\agraph \iso \agraph$.}
\item{{\bf Symmetry.}
If $\agraph \iso \agraph'$, then 
$\agraph = \amap(\agraph')$, where $\amap$ is an \nsubstitution such that $\dom{\amap} \supseteq \nodes{\agraph'}$.
Since $\amap$ is injective, $\amap^{-1}$ is well-defined, $\amap^{-1}(\agraph) = (\amap^{-1} \circ \amap)(\agraph') = \agraph'$ and $\dom{\amap^{-1}} = \img{\amap} \supseteq \amap(\nodes{\agraph'}) = \nodes{\agraph}$.
Hence $\agraph' \iso \agraph$.
}
\item{{\bf Transivity.} 
If $\agraph \iso \agraph' \iso \agraph''$ then there exist {\nsubstitution}s $\amap$ and $\amap'$ such that
$\agraph = \amap(\agraph')$, $\agraph' = \amap'(\agraph'')$, $\dom{\amap} \supseteq \nodes{\agraph'}$
and $\dom{\amap'} \supseteq \nodes{\agraph''}$.
Then since the composition of two injective functions is also injective, $\amap \circ \amap'$ is an \nsubstitution, 
and $\agraph = (\amap \circ \amap')(\amap'(\agraph''))$ 
with $\dom{\amap \circ \amap'} \supseteq \nodes{\agraph''}$. 
Hence $\agraph \iso \agraph''$.
}
\end{itemize}
\end{proof}



We now define a replacement operation on the considered class of graphs. Intuitively, this replacement operation applied to a graph permits to replace one of its subgraphs by another graph, provided a so-called \emph{substitutability} condition is satisfied. This condition will guarantee that the replacement operation on graphs enjoys properties similar to those of standard term rewriting.

\begin{definition} 
	\label{def:rootsimilar}
	Let $\agraph$, $\agraphB$ and $\agraphB'$ be {\graph}s such that
	$\agraphB \sgr \agraph$, $\inputs{\agraphB} = \tuple{\anode_1,\dots,\anode_n}$
	and
	$\inputs{\agraphB'} = \tuple{\anode_1',\dots,\anode_m'}$.
	The {\graph}s $\agraphB$ and $\agraphB'$ are {\em \rootsimilar}
	if $n = m$ and for every $i \in \interv{1}{n}$, we have $\sortof(\anode_i) = \sortof(\anode_i')$	and $\anode_i \eqnodes \anode_i'$.
	\end{definition}
	
	\begin{definition}
	\label{def:substitutable}
	A {\graph} $\agraphB'$ is {\em \substitutable for  $\agraphB$ in $\agraph$} if
	$\agraphB$ and $\agraphB'$ are {\em \rootsimilar}
	 and $\nodes{\agraph} \cap \nodes{\agraphB'} \subseteq \nodes{\agraphB}$.
\end{definition}

An immediate consequence of this definition is the following property:
\begin{proposition}
	\label{prop:substitutable}
	If $\agraphB'$ is \substitutable for  $\agraphB$ in $\agraph$ then $(\nodes{\agraph}\setminus \nodes{\agraphB}) \cap \nodes{\agraphB'} = \emptyset$.
\end{proposition}

\newcommand{\submap}[3]{mapping from \ensuremath{#1} to \ensuremath{#2} in \ensuremath{#3}\xspace}
\newcommand{\amapsub}[3]{\mu^{#1\rightarrow #2}_{#3}}

\begin{definition}
\label{def:replace}
Let $\agraph$, $\agraphB$ and $\agraphB'$ be {\graph}s such that 
 {\graph} $\agraphB'$ is \substitutable for  $\agraphB$ in $\agraph$.
Let $\agraph$ be a \graph. 
Let $\agraphB$ and $\agraphB'$ be two {\graph}s, with 
$\agraphB \sgr \agraph$, $\inputs{\agraphB} = \tuple{\anode_1,\dots,\anode_n}$
and
$\inputs{\agraphB'} = \tuple{\anode_1',\dots,\anode_m'}$.
The \emph{\submap{\agraphB}{\agraphB'}{\agraph}}, is the function
 $\amapsub{\agraphB}{\agraphB'}{\agraph} \isdef \{ \anode_i \mapsto \anode_i' \mid 1 \leq i \leq n\} \cup \{ \anode \mapsto \anode \mid 
\anode \in \nodes{\agraph} \setminus \nodes{\agraphB} \}$. 

Let $\amap \isdef \amapsub{\agraphB}{\agraphB'}{\agraph}$. We denote by $\replace{\agraph}{\agraphB}{\agraphB'}$ the \graph $\agraph'$ defined as follows:
\begin{enumerate}
\item{$\nodes{\agraph'} \isdef (\nodes{\agraph} \setminus \nodes{\agraphB}) \cup \nodes{\agraphB'}$. \label{replace:nodes}}
\item{$\inputs{\agraph'} \isdef \amap(\inputs{\agraph})$.
\label{replace:inputs}
}
\item{$\edges{\agraph'} \isdef  \amap(\edges{\agraph} \setminus \edges{\agraphB}) \cup \edges{\agraphB'}$.
\label{replace:edges}}
\item{For every $\anode \in \nodes{\agraph'}$, 
	$\labs{\agraph'}(\anode) \isdef \left\{ 
	\begin{array}{l} 
		\labs{\agraph}(\anode) 
		\text{\ if $\anode \in \innernodes{\agraph'}\setminus \nodes{\agraphB'}$} \\
		\labs{\agraphB'}(\anode) 
		\text{\ if $\anode \in \innernodes{\agraphB'}$} \\
		\labs{\agraph}({\inv{\amap}(\anode)})
		\text{\ if $\inv{\amap}(\anode) \in \innernodes{\agraph} \cap \inputs{\agraphB}$} \\
		\text{undefined otherwise}
	\end{array}
	\right.
	$
	\label{replace:labels}
}

%
\end{enumerate}
\end{definition}

\begin{example}\label{ex:subs}
	Consider the following {\graph}s (see also Figure \ref{fig:subs}):
	\begin{itemize}
		\item  $\agraph$ is defined as follows:
		\begin{eqnarray*}
			\nodes{\agraph} & \isdef & \set{\anode_1, \ldots, \anode_4}\\
			\inputs{\agraph} & \isdef & \tuple{}\\
			\edges{\agraph} & \isdef & \set{\mkedge{\anode_1}{\anode_2},\, \mkedge{\anode_2}{\anode_3},\, \mkedge{\anode_3}{\anode_4}}\\
			\labs{\agraph} & \isdef & \set{(\anode_1, a),\, (\anode_2, f(x)),\, (\anode_3, f(b)),\, (\anode_4, b)}
		\end{eqnarray*}
		\item $\agraphB$ is defined as follows:
		\begin{eqnarray*}
			\nodes{\agraph} & \isdef & \set{\anode_2, \anode_3}\\
			\inputs{\agraph} & \isdef &\tuple{\anode_2, \anode_3}\\
			\edges{\agraph} & \isdef & \set{\mkedge{\anode_2}{\anode_3}}\\
			\labs{\agraph} & \isdef & \emptyset
		\end{eqnarray*}
		\item $\agraphB'$ is defined as follows:
		\begin{eqnarray*}
			\nodes{\agraph} & \isdef & \set{\anodeB_1, \ldots, \anodeB_4}\\
			\inputs{\agraph} & \isdef & \tuple{\anodeB_1, \anodeB_2}\\
			\edges{\agraph} & \isdef & \set{\mkedge{\anode_1}{\anode_3},\, \mkedge{\anode_1}{\anode_4},\, \mkedge{\anode_2}{\anode_3},\, \mkedge{\anode_2}{\anode_4}}\\
			\labs{\agraph} & \isdef & \set{(\anode_3, c),\, (\anode_4, d)}
		\end{eqnarray*}
	\end{itemize}
	We assume that all nodes are of the same sort and that all nodes are distinct. Then $\agraphB$ is a subgraph of $\agraph$, and $\agraphB'$ is \substitutable for $\agraphB$ in $\agraph$. The \graph $\replace{\agraph}{\agraphB}{\agraphB'}$ is represented in Figure \ref{fig:subs2}.
\end{example}

\begin{figure}
	\begin{center}
		\begin{minipage}{0.4\textwidth}
			\begin{center}
				{\tiny
					\begin{tikzpicture}
						\draw[dashed] (7.5,0) ellipse (1.1cm and 0.4cm);
						\node[inputNode, thick] (i1) at (6, 0) {$\anode_1$};
						\node[inputNode, thick] (i2) at (7, 0) {$\anode_2$};
						\node[inputNode, thick] (i3) at (8, 0) {$\anode_3$};			
						\node[inputNode, thick] (i4) at (9, 0) {$\anode_4$};
						\node[draw=none](H) at (7.5, -0.75) {$H$};
						\node[draw=none] (h1) at (11, 0) {};
						\draw[stateTransition] (i1) -- (i2);
						\draw[stateTransition] (i2) -- (i3);
						\draw[stateTransition] (i3) -- (i4);
					\end{tikzpicture}
				}
			\end{center}
		\end{minipage}
		\begin{minipage}{0.4\textwidth}
			\begin{center}
				{\tiny
					\begin{tikzpicture}
						\node[inputNode, thick] (i1) at (6, 0) {$\anodeB_1$};
						\node[inputNode, thick] (i2) at (8, 0) {$\anodeB_2$};
						\node[inputNode, thick] (i3) at (7, -0.75) {$\anodeB_3$};			
						\node[inputNode, thick] (i4) at (7, -1.5) {$\anodeB_4$};
						\node[draw=none] (h2) at (5, 0) {};
						\draw[stateTransition] (i1) -- (i3);
						\draw[stateTransition] (i1) -- (i4);
						\draw[stateTransition] (i2) -- (i3);
						\draw[stateTransition] (i2) -- (i4);				
					\end{tikzpicture}
				}
			\end{center}
		\end{minipage}
	\end{center}
	\caption{Graphs $\agraph$, $\agraphB$ (dashed lines) and $\agraphB'$   in Example \ref{ex:subs}}\label{fig:subs}
\end{figure}

\begin{figure}
	\begin{center}
	{\tiny
		\begin{tikzpicture}
			\node[inputNode, thick] (i1) at (6, 0) {$\anode_1$};
			\node[inputNode, thick] (j1) at (7, 0) {$\anodeB_1$};
			\node[inputNode, thick] (j2) at (9, 0) {$\anodeB_2$};			
			\node[inputNode, thick] (i4) at (10, 0) {$\anode_4$};
			\node[inputNode, thick] (j3) at (8, -0.75) {$\anodeB_3$};			
			\node[inputNode, thick] (j4) at (8, -1.5) {$\anodeB_4$};
			\node[draw=none] (h1) at (11, 0) {};
			\draw[stateTransition] (i1) -- (j1);
			\draw[stateTransition] (j2) -- (i4);
			\draw[stateTransition] (j1) -- (j3);
			\draw[stateTransition] (j1) -- (j4);
			\draw[stateTransition] (j2) -- (j3);
			\draw[stateTransition] (j2) -- (j4);
		\end{tikzpicture}
	}
	\end{center}
	\caption{The graph $\replace{\agraph}{\agraphB}{\agraphB'}$ in Example \ref{ex:subs}}\label{fig:subs2}
\end{figure}


\begin{proposition}
	\label{prop:amap_image}
	If $\agraphB'$ is \substitutable for  $\agraphB$ in $\agraph$, then by letting $\amap \isdef \amapsub{\agraphB}{\agraphB'}{\agraph}$ and $\agraph' \isdef \replace{\agraph}{\agraphB}{\agraphB'}$, we have $\amap(\inputs{\agraph}) = \inputs{\agraph'}$ and $\amap(\inputs{\agraphB}) = \inputs{\agraphB'}$.
\end{proposition}


\begin{proposition}\label{prop:subs-minus}
	If $\agraphB'$ is \substitutable for  $\agraphB$ in $\agraph$, then by letting $\agraph' \isdef \replace{\agraph}{\agraphB}{\agraphB'}$, we have:
	\begin{itemize}
		\item $\nodes{\agraph'} \setminus \nodes{\agraphB'} = \nodes{\agraph}\setminus\nodes{\agraphB}$;
		\item $\innernodes{\agraph'}\setminus \nodes{\agraphB'} = \innernodes{\agraph}\setminus \nodes{\agraphB}$.
	\end{itemize}
\end{proposition}

\begin{proof}
	We have $\nodes{\agraph'} \setminus \nodes{\agraphB'} = \left[(\nodes{\agraph} \setminus \nodes{\agraphB}) \cup \nodes{\agraphB'}\right] \setminus \nodes{\agraphB'} = \left[\nodes{\agraph} \setminus \nodes{\agraphB}\right] \setminus \nodes{\agraphB'} = \nodes{\agraph}\setminus\nodes{\agraphB}$ by Proposition \ref{prop:substitutable}. Consider a node $\anode \in \nodes{\agraph}$; we show that $\anode \in \inputs{\agraph}\setminus \nodes{\agraphB}$ if and only if $\anode \in \inputs{\agraph'}\setminus \nodes{\agraphB'}$; 
together with the previous result this will prove the second item of the proposition.  First assume that $\anode \in \inputs{\agraph}\setminus \nodes{\agraphB}$ and let $\amap$ denote the \submap{\agraphB}{\agraphB'}{\agraph}. Then by definition, we have $\amap(\anode) = \anode \in \inputs{\agraph'}$ and by Proposition \ref{prop:substitutable}, $\anode \notin \nodes{\agraphB'}$; hence the result. Conversely, if $\anode \in \inputs{\agraph'}\setminus \nodes{\agraphB'}$ then $\anode \in \nodes{\agraph'}\setminus \nodes{\agraphB'} = \nodes{\agraph}\setminus\nodes{\agraphB}$ by the first item of the proposition, and by definition, $\anode = \amap(\anode')$ for some $\anode'\in \inputs{\agraph}$. Since $\anode\notin \nodes{\agraphB}$, necessarily $\anode' = \anode$, hence the result.
\end{proof}


Observe that, with the notations of Definition \ref{def:replace}, if $\anode \in \innernodes{\agraph'}\setminus \nodes{\agraphB'} = \innernodes{\agraph}\setminus \nodes{\agraphB}$ then necessarily $\amap(\anode) = \anode$ hence
$\inv{\amap}(\anode) = \anode$, thus 
$\labs{\agraph'}(\anode) = \labs{\agraph}(\inv{\amap}(\anode))$.

\begin{lemma}\label{lm:sub-subg}
Let $\agraph$ be a \graph, let $\agraphB$ be a \subgraph of $\agraph$
and let $\agraphB'$ be a \graph \substitutable for $\agraphB$ in $\agraph$.
Then $\replace{\agraph}{\agraphB}{\agraphB'}$ is a \graph and $\agraphB' \sgr \replace{\agraph}{\agraphB}{\agraphB'}$.
\end{lemma}


\begin{proof}

Observe that the function $\amap$ in Definition \ref{def:replace} is injective, since  $\inputs{\agraphB'}$ is repetition-free and 
no node in $\nodes{\agraph} \setminus \nodes{\agraphB}$ may occur in $\inputs{\agraphB'}$, because $\nodes{\agraph} \cap \nodes{\agraphB'} \subseteq \nodes{\agraphB}$ by hypothesis.
We show that the domain of $\amap$ contains  all {\node}s occurring either in 
$\inputs{\agraph}$ or in an edge in $\edges{\agraph} \setminus \edges{\agraphB}$.
By Condition \ref{subgraph:inputs_bis}  
of Definition \ref{def:subgraph}, the nodes in $\inputs{\agraph}$
occur in  either  $\inputs{\agraphB}$
or $\nodes{\agraph} \setminus \nodes{\agraphB}$. 

The function $\labs{\agraph'}$ is well-defined: if $\anode \in \innernodes{\agraphB'}$
and $\anode = \amap(\anode')$ for 
$\anode' \in \innernodes{\agraph}\setminus \nodes{\agraphB}$, then
we have $\anode' \not \in \nodes{\agraphB'}$ (because $\nodes{\agraph} \cap \nodes{\agraphB'} \subseteq \nodes{\agraph}$),
hence 
$\anode' \not = \anode$, and 
in this case, by definition of $\amap$, $\anode$ occurs in $\inputs{\agraphB'}$, hence may not occur in $\innernodes{\agraphB'}$, a contradiction.
Thus it is straightforward to check that $\replace{\agraph}{\agraphB}{\agraphB'}$ is a \graph.

We prove that 
$\agraphB' \sgr \agraph'$, with $\agraph' \isdef \replace{\agraph}{\agraphB}{\agraphB'}$:
\begin{enumerate}
\item{$\nodes{\agraphB'} \subseteq \nodes{\agraph'} = (\nodes{\agraph} \setminus \nodes{\agraphB}) \cup \nodes{\agraphB'}$.}
\item{
 If $\mkedge{\anode}{\anodeB}\in \edges{\agraphB'}$ then 
$\anyedge{\anode}{\anodeB}\in \edges{\agraph'} = \amap(\edges{\agraph} \setminus \edges{\agraphB}) \cup \edges{\agraphB'}$.}
\item{
If $\anode, \anodeB \in \nodes{\agraphB'}$ and 
$\mkedge{\anode}{\anodeB} \in \edges{\agraph'}$, then by definition of $\edges{\agraph'}$, either 
$\mkedge{\anode}{\anodeB} \in \edges{\agraphB'}$ and the proof is completed, 
or
$\mkedge{\anode}{\anodeB} \in   \amap(\edges{\agraph} \setminus \edges{\agraphB})$.
In the latter case, we have 
$\anode = \amap(\anode')$, $\anodeB = \amap(\anodeB')$ with 
$\mkedge{\anode'}{\anodeB'} \in \edges{\agraph} \setminus \edges{\agraphB}$.
We show that $\anode' \in \nodes{\agraphB}$.
If $\anode = \anode'$, then we have $\anode \in \nodes{\agraphB'} \cap \nodes{\agraph}$,
thus $\anode = \anode' \in \nodes{\agraphB}$ (since $\agraphB'$ is \substitutable for $\agraphB$ in $\agraph$). 
Otherwise,  $\anode \not = \anode'$, hence
by definition of $\amap$, $\anode'$ occurs in $\inputs{\agraphB}$, thus $\anode' \in \nodes{\agraphB}$.
Similarly, $\anodeB' \in \nodes{\agraphB}$, 
and since $\agraphB \sgr \agraph$ we deduce that
$\mkedge{\anode'}{\anodeB'} \in \edges{\agraphB}$, contradicting the fact that $\mkedge{\anode'}{\anodeB'} \in \edges{\agraph} \setminus \edges{\agraphB}$.
}

\item{Assume that $\anyedge{\anode}{\anodeB} \in \edges{\agraph'}$, $\anode \in \nodes{\agraphB'}$ and
$\anodeB \not \in \nodes{\agraphB'}$. By definition of $\edges{\agraph'}$ this entails 
that 
$\anyedge{\anode}{\anodeB} \in \amap(\edges{\agraph} \setminus \edges{\agraphB})$
and $\anode = \amap(\anode')$, $\anodeB = \amap(\anodeB')$, for some $\anode',\anodeB' \in \nodes{\agraph}$.
If $\anode = \anode'$ then $\anode\in \nodes{\agraphB'} \cup \nodes{\agraph}$, and
since $\agraphB'$ is \substitutable for $\agraphB$ in $\agraph$, we have 
$\anode \in \nodes{\agraphB}$, thus, by definition of $\amap$, 
$\anode$ occurs in $\inputs{\agraphB}$ (since $\anode\in \dom{\amap}$).
Otherwise,
$\anode \not = \anode'$, and  $\anode$ occurs in $\inputs{\agraphB'}$ by definition of $\amap$.
}

\item{Let $\anode\in \nodes{\agraphB'}$ be a \node occurring in $\inputs{\agraph'}$.
We have $\anode = \amap(\anode')$ for some \node occurring in $\inputs{\agraph}$.
By definition of $\amap$, either $\anode$ occurs in $\inputs{\agraphB'}$ and the proof is completed,
or $\anode = \anode'$ with $\anode\not \in \nodes{\agraphB}$, which contradicts the hypothesis
that $\agraphB'$ is \substitutable for $\agraphB$ in $\agraph$. 
} 
\item{By definition of $\labs{\agraph'}$, we have
$\labs{\agraphB'} = \restrict{\labs{\agraph'}}{\innernodes{\agraphB'}}$.}

\end{enumerate}
\end{proof}

\subsection{{Properties of the Replacement Operation}}

We establish some basic properties of the replacement operation, all of which are similar to their counterparts for the replacement operation on first-order terms.

\begin{proposition}
\label{prop:trivial_replacement}
Let $\agraph$ be a \graph and let $\agraphB\sgr \agraph$.
Then $\agraphB$ is \substitutable for $\agraphB$ in $\agraph$ and $\replace{\agraph}{\agraphB}{\agraphB} = \agraph$.
\end{proposition}
\begin{proof}
By definition, $\nodes{\agraph} \cap \nodes{\agraphB} \subseteq \nodes{\agraphB}$
thus $\agraphB$ is \substitutable for $\agraphB$ in $\agraph$.
The conditions of Definition \ref{def:replace} (with $\agraphB = \agraphB'$) entail that $\amap$ is the identity on $\nodes{\agraph} \setminus \innernodes{\agraphB}$
and that $\nodes{\agraph'} = (\nodes{\agraph} \setminus \nodes{\agraphB}) \cup \nodes{\agraphB} = \nodes{\agraph}$, 
$\inputs{\agraph'} = \inputs{\agraph}$,
$\edges{\agraph'}  = (\edges{\agraph} \setminus \edges{\agraphB}) \cup \edges{\agraphB} = \edges{\agraph}$
and $\labs{\agraph'} = \labs{\agraph}$ (since $\labs{\agraphB} = \restrict{\labs{\agraph}}{\innernodes{\agraphB}}$). 
\end{proof}

\newcommand{\rewriterule}{\graph rewrite rule\xspace}
\newcommand{\rules}{\mathit{R}}


\begin{proposition}
\label{prop:replace_iso}
Let $\agraph, \agraphB$ and $\agraphB'$ be {\graph}s, with $\agraphB \sgr \agraph$
and $\agraphB'$ is \substitutable for $\agraphB$ in $\agraph$.
Let $\amap$
be an \nsubstitution with domain $\nodes{\agraph} \cup \nodes{\agraphB'}$.
Then  $\amap(\agraphB')$ is \substitutable for $\amap(\agraphB)$ in $\amap(\agraph)$, and:
\[\amap(\replace{\agraph}{\agraphB}{\agraphB'}) = 
\replace{\amap(\agraph)}{\amap(\agraphB)}{\amap(\agraphB')}\]
\end{proposition}
\begin{proof}
We have $\nodes{\amap(\agraph)} \cap \nodes{\amap(\agraphB')} = \amap(\nodes{\agraph}) \cap \amap(\nodes{\agraphB'}) =  
\amap(\nodes{\agraph} \cap \nodes{\agraphB'})$ by injectivity of $\amap$,
thus 
$\nodes{\amap(\agraph)} \cap \nodes{\amap(\agraphB')} \subseteq \amap(\nodes{\agraphB}) = \nodes{\amap(\agraphB)}$ (since
$\nodes{\agraph} \cap \nodes{\agraphB'} \subseteq \nodes{\agraphB}$, as $\agraphB'$ is \substitutable for $\agraphB$ in $\agraph$). Consequently, 
$\amap(\agraphB')$ is \substitutable for $\amap(\agraphB)$ in $\amap(\agraph)$.

Let $\agraph' \isdef \replace{\agraph}{\agraphB}{\agraphB'}$
and $\agraph'' \isdef \replace{\amap(\agraph)}{\amap(\agraphB)}{\amap(\agraphB')}$.
We show that $\amap(\agraph') = \agraph''$.
Let $\inputs{\agraphB} = \tuple{\anode_1,\ldots,\anode_n}$, 
$\inputs{\agraphB'} = \tuple{\anode_1',\ldots,\anode_n'}$, and:
\[\amap' = \{ \anode_i \mapsto \anode_i' \mid 1 \leq i \leq n\} \cup \{ \anode \mapsto \anode \mid 
\anode \in \nodes{\agraph} \setminus \nodes{\agraphB} \}\]
\[\amap'' = \{ \amap(\anode_i) \mapsto \amap(\anode_i') \mid 1 \leq i \leq n\} \cup \{ \anode \mapsto \anode \mid 
\anode \in \nodes{\amap(\agraph)} \setminus \nodes{\amap(\agraphB)} \}\]
Observe that $\amap'' = \amap\circ\amap'\circ\amap^{-1}$. We have:

\begin{itemize}
\item{
\[
\begin{array}{lll}
\nodes{\amap(\agraph')} & = & \amap(\nodes{\agraph'}) 
 =  \amap((\nodes{\agraph} \setminus \nodes{\agraphB}) \cup \nodes{\agraphB'}) 
 =  \amap(\nodes{\agraph} \setminus \nodes{\agraphB}) \cup \amap(\nodes{\agraphB'}) \\
& = & (\amap(\nodes{\agraph}) \setminus \amap(\nodes{\agraphB})) \cup \amap(\nodes{\agraphB'}) \quad \text{by injectivity of $\amap$} \\
& = & \nodes{\agraph''} \\
\end{array}
\] 
}
\item{
\[
\begin{array}{lll}
\edges{\amap(\agraph')} & = & \amap(\edges{\agraph'}) 
= \amap(\amap'(\edges{\agraph} \setminus\edges{\agraphB}) \cup \edges{\agraphB'}) \\
& = & \amap(\amap'(\edges{\agraph} \setminus\edges{\agraphB})) \cup \amap(\edges{\agraphB'}))
= \amap(\amap'(\edges{\agraph} \setminus\edges{\agraphB})) \cup \edges{\amap(\agraphB')} \\
& = &  (\amap(\amap'(\edges{\agraph})) \setminus \amap(\amap'(\edges{\agraphB}))) \cup \edges{\amap(\agraphB')} \quad \text{by injectivity of $\amap,\amap'$} \\ & = &  (\edges{\amap(\amap'(\agraph))})) \setminus \edges{\amap(\amap'(\agraphB))}) \cup \edges{\amap(\agraphB')} \\ \end{array}
\]
Further, 
\[
\begin{array}{lll}
\edges{\agraph''} & = & \amap''(\edges{\amap(\agraph)} \setminus \edges{\amap(\agraphB)}) \cup \edges{\amap(\agraphB')}  =  (\edges{\amap''(\amap(\agraph))} \setminus \edges{\amap''(\amap(\agraphB))}) \cup \edges{\amap(\agraphB')} \\
& = & (\edges{\amap(\amap'(\agraph))} \setminus \edges{\amap(\amap'(\agraphB))}) \cup \edges{\amap(\agraphB')} \\
\end{array}
\]
Hence $\edges{\amap(\agraph')} = \edges{\agraph''}$.
}
\item{$\inputs{\amap(\agraph')}  = \amap(\inputs{\agraph'}) = 
\amap(\amap'(\inputs{\agraph})) = \amap''(\amap(\inputs{\agraph})) =  \amap''(\inputs{\amap(\agraph)})) = \inputs{\agraph''}$.
}
\item{Assume that $\anode = \amap''(\anode')$, with $\anode' \in \innernodes{\amap(\agraph)} \setminus \innernodes{\amap(\agraphB)}$.
Then $\anode' = \amap(\anode'')$, where $\anode'' \in \innernodes{\agraph} \setminus \innernodes{\agraphB}$, 
and $\amap(\amap'(\anode'')) = \amap''(\anode') = \anode$.
Furthermore, $\labs{\agraph''}(\anode) = \labs{\amap(\agraph)}(\anode') = 
\labs{\agraph}(\anode'') 
= \labs{\agraph'}(\amap'(\anode''))  = \labs{\amap(\agraph')}(\anode)$, because $\inv{\amap}(\anode) = \amap'(\anode'')$. 
If $\anode \in \innernodes{\amap(\agraphB')}$, then $\anode = \amap(\anode')$, with $\anode' \in \innernodes{\agraphB}$. Moreover, we have
$\labs{\agraph''}(\anode) = \labs{\amap(\agraphB')}(\anode) =  \labs{\agraphB'}(\anode') = \labs{\agraph'}(\anode) = 
\labs{\amap(\agraph')}(\anode)$.
Finally, if
$\anode = \amap(\anode')$ is not of one of the two forms above, then 
it is easy to check that 
$\anode' \not \in \innernodes{\agraphB'}$, 
and that
there is no $\anode''$ such that $\anode' = \amap'(\anode'')$, with $\anode'' \in \innernodes{\agraph} \setminus \nodes{\agraphB}$.
Thus $\labs{\agraph''}(\anode)$ and 
$\labs{\agraph'}(\anode')$ are both undefined.
}
\end{itemize}
\end{proof}

The next lemma shows that renaming the nodes of some subgraph does not affect the graphs obtained 
by replacing this subgraph.

\begin{lemma}
\label{lem:replace_id}
Let $\agraph, \agraphB$ and $\agraphB'$ be {\graph}s, with $\agraphB \sgr \agraph$
and $\agraphB'$ is \substitutable for $\agraphB$ in $\agraph$.
Let $\amap$ be an \nsubstitution of domain $\nodes{\agraph}$ such that 
$\amap(\anode)= \anode$ if $\anode \in \nodes{\agraph} \setminus \nodes{\agraphB}$.
Then:
\[\replace{\amap(\agraph)}{\amap(\agraphB)}{\agraphB'} = 
\replace{\agraph}{\agraphB}{\agraphB'}\]
\end{lemma}
\begin{proof}
Let $\agraph' \isdef \replace{\agraph}{\agraphB}{\agraphB'}$,
 $\agraph'' \isdef \replace{\amap(\agraph)}{\amap(\agraphB)}{\agraphB'}$, 
 $\inputs{\agraphB} \isdef \tuple{\anode_1,\dots,\anode_n}$, 
 $\inputs{\agraphB'}\isdef \tuple{\anode_1',\dots,\anode_n'}$ 
and consider the following {\nsubstitution}s:
\[\amap_1 \isdef \{ \anode_i \mapsto \anode_i' \mid 1 \leq i \leq n\} \cup \{ \anode \mapsto \anode \mid 
\anode \in \nodes{\agraph} \setminus \nodes{\agraphB} \} \]
and
\[\amap_2 \isdef \{ \amap(\anode_i) \mapsto \anode_i' \mid 1 \leq i \leq n\} \cup \{ \anode \mapsto \anode \mid 
\anode \in \nodes{\mu(\agraph)} \setminus \nodes{\mu(\agraphB)} \}. \]
By injectivity of $\mu$, we have $\nodes{\mu(\agraph)} \setminus \nodes{\mu(\agraphB)} = \mu(\nodes{\agraph}) \setminus \mu(\nodes{\agraphB})
 = \mu(\nodes{\agraph} \setminus\nodes{\agraphB}) = \nodes{\agraph} \setminus\nodes{\agraphB}$, since $\amap(\anode) = \anode$ if $\anode \in \nodes{\agraph} \setminus \nodes{\agraphB}$. 
Thus $\amap_2\circ\amap = \amap_1$.
We show that $\agraph' = \agraph''$:
\begin{itemize}
\item{By Definition \ref{def:replace} (\ref{replace:nodes}), 
 $\nodes{\agraph'} = (\nodes{\agraph} \setminus \nodes{\agraphB}) \cup \nodes{\agraphB'}$
and 
$\nodes{\agraph'} = (\nodes{\amap(\agraph)} \setminus \nodes{\amap(\agraphB)}) \cup \nodes{\agraphB'}
=   (\nodes{\agraph} \setminus \nodes{\agraphB})   \cup \nodes{\agraphB'} = 
\nodes{\agraph'}$ .}
\item{By Definition \ref{def:replace} (\ref{replace:inputs}),
$\inputs{\agraph'} = \amap_1(\inputs{\agraph})$
and
$\inputs{\agraph''} = \amap_2(\inputs{\amap(\agraph)}) = \amap_2(\amap(\inputs{\agraph})) 
= \amap_1(\inputs{\agraph}) = \inputs{\agraph'}$.
}
\item{By Definition \ref{def:replace} (\ref{replace:edges}),
$\edges{\agraph'} = \amap_1(\edges{\agraph} \setminus \edges{\agraphB}) \cup \edges{\agraphB'}$ and
$\edges{\agraph''} = \amap_2(\edges{\amap(\agraph)} \setminus \edges{\amap(\agraphB)}) \cup \edges{\agraphB'}
=  \amap_2(\amap(\edges{\agraph}) \setminus \amap(\edges{\agraphB})) \cup \edges{\agraphB'}$.
Since $\amap$ is injective, we deduce that
$\edges{\agraph''} = \amap_2(\amap(\edges{\agraph} \setminus \edges{\agraphB})) \cup \edges{\agraphB'} 
= \amap_1(\edges{\agraph} \setminus \edges{\agraphB}) \cup \edges{\agraphB'}  = \edges{\agraph'}$.
}
\item{Let $\anode \in \innernodes{\agraph} \setminus \innernodes{\agraphB}$, with $\anode = \amap_1(\anode')$, hence $\anode = \amap_2(\amap(\anode))$.
By Definition \ref{def:replace} (\ref{replace:labels}),
we have $\labs{\agraph'}(\anode) =  \labs{\agraph}(\anode')$
and
$\labs{\agraph''}(\anode) =  \labs{\amap(\agraph)}(\amap(\anode')) = \labs{\agraph}(\anode') = \labs{\agraph'}(\anode)$.
If $\anode \in \innernodes{\agraphB'}$, then, again by 
Definition \ref{def:replace} (\ref{replace:labels}),
we have $\labs{\agraph'}(\anode) = \labs{\agraph''}(\anode) = \labs{\agraphB'}(\anode)$.
}
\end{itemize}
\end{proof}

This entails that the replacement of a subgraph by an isomorphic subgraph preserves isomorphism: 

\begin{corollary}
\label{cor:replace_iso_subgraph}
Let $\agraph, \agraphB, \agraphB', \agraphB''$ be {\graph}s, with
$\agraphB \sgr \agraph$, $\agraphB'_1$ and $\agraphB'_2$ are \substitutable for $\agraphB$ in $\agraph$
and $\agraphB'_1 \iso \agraphB'_2$. Then:
\[\replace{\agraph}{\agraphB}{\agraphB_1'} \iso \replace{\agraph}{\agraphB}{\agraphB_2'}\]
\end{corollary}
\begin{proof}
Since $\agraphB'_1 \iso \agraphB'_2$, there exists an \nsubstitution $\amap$ of domain $\nodes{\agraphB'_2}$ such that
$\agraphB'_1 = \amap(\agraphB'_2)$.
Consider the extension $\amap'$ of $\amap$ to $\nodes{\agraph}$, such that: 
$\amap'(\anode) = \anode$ if $\anode \in \nodes{\agraph} \setminus \nodes{\agraphB}$
and all the nodes in $\nodes{\agraphB} \setminus \nodes{\agraphB_2'}$ 
are mapped to pairwise distinct nodes not occurring in $\agraph$ or $\agraphB_1'$.
Observe that $\amap$ is well-defined, since
$\nodes{\agraph} \cap \nodes{\agraphB_2'} \subseteq \nodes{\agraphB}$, as $\agraphB_2'$ is \substitutable for
$\agraphB$ in $\agraph$.
We show that $\amap'$ is injective. 
Let $\anode, \anode'$ such that $\amap'(\anode)=\amap'(\anode')$.
If $\anode,\anode' \in \nodes{\agraphB_2'}$ then $\amap'(\anode)=\amap(\anode)$ 
and $\amap'(\anode') = \amap(\anode')$ hence $\anode = \anode'$ as $\amap$ is injective.
If $\anode,\anode' \in \nodes{\agraph} \setminus \nodes{\agraphB}$ then 
  $\amap'(\anode)=\anode$ 
and $\amap'(\anode') = \anode'$ hence $\anode = \anode'$.
If $\anode,\anode'\in \nodes{\agraphB} \setminus \nodes{\agraphB_2'}$ then 
we have $\anode = \anode'$ by definition of $\amap'$.
If $\anode \in \nodes{\agraphB} \setminus \nodes{\agraphB_2'}$ and 
$\anode' \not \in \nodes{\agraphB} \setminus \nodes{\agraphB_2'}$, then by definition of $\amap'$, 
$\amap'(\anode') \in \nodes{\agraph} \cup \amap'(\nodes{\agraphB_2'}) = \nodes{\agraph} \cup \nodes{\agraphB_1'}$
and 
$\amap'(\anode) \not \in \nodes{\agraph} \cup \nodes{\agraphB_1'}$, which contradicts the fact that 
$\amap'(\anode)=\amap'(\anode')$.
The only remaining case is (by symmetry): 
$\anode \in \nodes{\agraphB_2'}$
and $\anode' \in \nodes{\agraph} \setminus \nodes{\agraphB}$.
Then $\amap'(\anode) = \amap(\anode) \in \nodes{\agraphB_1'}$
and $\amap'(\anode') = \anode' \in \nodes{\agraph}$, hence, since $\agraphB_1'$ is \substitutable for 
$\agraphB$ in $\agraph$, we must have $\anode' \in \nodes{\agraphB}$,  contradicting the fact that $\anode' \in \nodes{\agraph} \setminus \nodes{\agraphB}$.

We get, by Proposition \ref{prop:replace_iso}, since $\amap'$ coincides with $\amap$ on $\nodes{\agraphB_2'}$:
\[\amap'(\replace{\agraph}{\agraphB}{\agraphB_2'}) = \replace{\amap'(\agraph)}{\amap'(\agraphB)}{\amap'(\agraphB_2')} = \replace{\amap'(\agraph)}{\amap'(\agraphB)}{\amap(\agraphB_2')} = \replace{\amap'(\agraph)}{\amap'(\agraphB)}{\agraphB_1'} \]
By definition of $\amap'$, $\amap'(\anode) = \anode$ holds for any node $\anode \in \nodes{\agraph} \setminus \nodes{\agraphB}$, hence, 
by Lemma \ref{lem:replace_id}, we deduce:
\[\amap'(\replace{\agraph}{\agraphB}{\agraphB_2'}) = \replace{\agraph}{\agraphB}{\agraphB_1'} \]
Therefore, $\replace{\agraph}{\agraphB}{\agraphB_1'} \iso \replace{\agraph}{\agraphB}{\agraphB_2'}$.
\end{proof}

\begin{lemma}
\label{lem:replaceseq}
Let $\agraph$, $\agraphB$, $\agraphC$, $\agraphC'$, with
$\agraphB \sgr \agraph$, $\agraphC$ is \substitutable for $\agraphB$ in $\agraph$
and $\agraphC'$ is \substitutable for $\agraphC$ in $\replace{\agraph}{\agraphB}{\agraphC}$.
Then 
$\agraphC'$ is substitutable for $\agraphB$ in $\agraph$ and $\replace{\replace{\agraph}{\agraphB}{\agraphC}}{\agraphC}{\agraphC'} = \replace{\agraph}{\agraphB}{\agraphC'}$.
\end{lemma}
\begin{proof}
	Let $\agraph' \isdef \replace{\agraph}{\agraphB}{\agraphC}$.
	We show that $\agraphC'$ is substitutable for $\agraphB$ in $\agraph$, i.e., that $\nodes{\agraph} \cap \nodes{\agraphC'} \subseteq \nodes{\agraphB}$. 
	We have $\nodes{\agraph'} = (\nodes{\agraph}\setminus \nodes{\agraphB}) \cup \nodes{\agraphC}$, and show that $(\nodes{\agraph}\setminus \nodes{\agraphB}) \cap \nodes{\agraphC'} = \emptyset$. 
	We have
	\[(\nodes{\agraph}\setminus \nodes{\agraphB}) \cap \nodes{\agraphC'} \subseteq \left[(\nodes{\agraph}\setminus \nodes{\agraphB}) \cup \nodes{\agraphC}\right] \cap \nodes{\agraphC'} = \nodes{\agraph'}\cap \nodes{\agraphC'} \subseteq \nodes{\agraphC},\]
	because $\agraphC'$ is \substitutable for $\agraphC$ in $\agraph'$ by hypothesis. We deduce that $(\nodes{\agraph}\setminus \nodes{\agraphB}) \cap \nodes{\agraphC'} \subseteq (\nodes{\agraph}\setminus \nodes{\agraphB}) \cap \nodes{\agraphC}$. But $(\nodes{\agraph} \setminus \nodes{\agraphB}) \cap\nodes{\agraphC} = \emptyset$ by Proposition \ref{prop:substitutable}, hence the result.
	
	
	We show that $\replace{\agraph'}{\agraphC}{\agraphC'} = \replace{\agraph}{\agraphB}{\agraphC'}$. We define $\amap \isdef \amapsub{\agraphB}{\agraphC}{\agraph}$, $\amap' \isdef \amapsub{\agraphC}{\agraphC'}{\agraph'}$ and $\amap'' \isdef \amapsub{\agraphB}{\agraphC'}{\agraph}$. 
	Note that $\amap'' = \amap' \circ \amap$ 
	Indeed, by definition we have $\dom{\amap} = \dom{\amap'}$. Further, for every $\anode \in \dom{\amap}$, if $\anode\in \nodes{\agraph} \setminus \nodes{\agraphB}$ then 
	$\amap(\anode) = \amap'(\anode)$, 
	and $\anode \not \in \nodes{\agraphC}$ (since $\agraphC$ is \substitutable for $\agraphB$ 
	in $\agraph$), hence 
	$\amap''(\anode) = \anode$.  
	If $\anode \in \dom{\amap}$ and $\anode\not \in \nodes{\agraph} \setminus \nodes{\agraphB}$, then $\anode = \anode_i$, for some $i = 1,\dots,n$, with 
	 $\inputs{\agraphB} = \tuple{\anode_1,\dots,\anode_n}$. By definition,  
	$\amap(\anode_i) = \anode_i'$, with $\inputs{\agraphC} = \tuple{\anode_1',\dots,\anode_n'}$
	and  $\amap'(\anode_i')= \anode_i''$, $\amap''(\anode_i) = \anode_i''$ with $\inputs{\agraphC'} = \tuple{\anode_1'',\dots,\anode_n''}$.
	\begin{enumerate}
		\item 	We have 
	\begin{eqnarray*}
		\nodes{\replace{\agraph'}{\agraphC}{\agraphC'}} & = &
	(\nodes{\agraph'} \setminus \nodes{\agraphC}) \cup \nodes{\agraphC'}\\ 
	& =& \left[((\nodes{\agraph} \setminus \nodes{\agraphB}) \cup \nodes{\agraphC}) \setminus \nodes{\agraphC}\right] \cup \nodes{\agraphC'}\\
	& = & \left[(\nodes{\agraph} \setminus \nodes{\agraphB})\setminus \nodes{\agraphC}\right] \cup \nodes{\agraphC'}\\
	& = & (\nodes{\agraph} \setminus \nodes{\agraphB}) \cup \nodes{\agraphC'}\\
	& = & \nodes{\replace{\agraph}{\agraphB}{\agraphC'}},
	\end{eqnarray*}
where the second to last equality is obtained using the fact that, by Proposition \ref{prop:substitutable}, $(\nodes{\agraph} \setminus \nodes{\agraphB})\cap \nodes{\agraphC} = \emptyset$.
	\item We have $\inputs{\replace{\agraph'}{\agraphC}{\agraphC'}} = \amap'(\inputs{\agraph'}) = \amap'(\amap(\inputs{\agraph})) = \amap''(\inputs{\agraph})$.
	\item We show that $\amap(\edges{\agraph}\setminus \edges{\agraphB}) \cap \edges{\agraphC} = \emptyset$. 
	Suppose for a contradiction that $\anyedge{\anode'}{\anodeB'} \in \amap(\edges{\agraph}\setminus \edges{\agraphB}) \cap \edges{\agraphC}$. Then $\anode', \anodeB'\in \nodes{\agraphC}$, and there exist nodes $\anode, \anodeB$ such that $\anode' = \amap({\anode})$ and $\anodeB' = \amap(\anodeB)$, with $\anyedge{\anode}{\anodeB} \in \edges{\agraph}\setminus \edges{\agraphB}$. Since $\anyedge{\anode}{\anodeB} \notin \edges{\agraphB}$, necessarily, $\set{\anode, \anodeB} \not\subseteq \nodes{\agraphB}$, by Definition \ref{def:subgraph} (\ref{subgraph:edges_bis}). Assume w.l.o.g.\ that $\anode \in \nodes{\agraph}\setminus \nodes{\agraphB}$. Then by definition $\anode' = \anode \in \nodes{\agraphC}$. But this is impossible because $(\nodes{\agraph}\setminus \nodes{\agraphB}) \cap \nodes{\agraphC} = \emptyset$ by Proposition \ref{prop:substitutable}. We deduce that	
	\begin{eqnarray*}
		\edges{\replace{\agraph'}{\agraphC}{\agraphC'}} & = & \amap'(\edges{\agraph'} \setminus \edges{\agraphC}) \cup \edges{\agraphC'}\\
		& = & \amap'\left(\left[\amap(\edges{\agraph}\setminus \edges{\agraphB}) \cup \edges{\agraphC}\right] \setminus \edges{\agraphC} \right)  \cup \edges{\agraphC'}\\
		& = & \amap'\left(\amap(\edges{\agraph}\setminus \edges{\agraphB}) \setminus \edges{\agraphC}\right)  \cup \edges{\agraphC'}\\
		& = & \amap''(\edges{\agraph}\setminus \edges{\agraphB})   \cup \edges{\agraphC'}\\
		& = & \edges{\replace{\agraph}{\agraphB}{\agraphC'}},
	\end{eqnarray*}
where the second to last equality is obtained using the fact that $\amap(\edges{\agraph}\setminus \edges{\agraphB}) \cap \edges{\agraphC} = \emptyset$ and that $\amap'' = \amap' \circ \amap$.
\item Let $\anode \in \nodes{\replace{\agraph'}{\agraphC}{\agraphC'}}$. 
\begin{itemize}
	\item If $\anode \in \innernodes{\replace{\agraph'}{\agraphC}{\agraphC'}} \setminus \nodes{\agraphC'}$, then by Proposition \ref{prop:subs-minus} we have $\anode \in \innernodes{\replace{\agraph'}{\agraphC}{\agraphC'}} \setminus \nodes{\agraphC'}$ if and only if $\alpha \in \innernodes{\agraph'} \setminus \nodes{\agraphC}$ if and only if $\anode \in \innernodes{\agraph}\setminus \nodes{\agraphB}$ if and only if $\anode \in \innernodes{\replace{\agraph}{\agraphB}{\agraphC'}} \setminus \nodes{\agraphC'}$. Thus $\labs{\replace{\agraph'}{\agraphC}{\agraphC'}}(\anode) = \labs{\agraph}(\anode) = \labs{\replace{\agraph}{\agraphB}{\agraphC'}}(\anode)$.
	\item If $\anode \in \innernodes{\agraphC'}$ then $\labs{\replace{\agraph'}{\agraphC}{\agraphC'}}(\anode) = \labs{\agraphC'}(\anode) = \labs{\replace{\agraph}{\agraphB}{\agraphC'}}(\anode)$.
	\item If $\anodeB \isdef \inv{\amap'}(\anode) \in\innernodes{\agraph'} \cap \inputs{\agraphC}$ then by definition of $\amap$, we must have $\inv{\amap}(\beta) \in \inputs{\agraphB}$. We cannot have $\anodeB \in \inputs{\agraph}$ because otherwise we would have $\anodeB \in \inputs{G'}$; hence $\inv{\amap}(\anodeB) \in \innernodes{\agraph} \cap \inputs{\agraphB}$. Since $\amap'' = \amap' \circ \amap$, we deduce that ${\amap''}^{-1}(\anode) = \amap^{-1}(\beta)$ and $\labs{\replace{\agraph'}{\agraphC}{\agraphC'}}(\anode) = \labs{\agraph'}(\anodeB) = \labs{\agraph}(\inv{\amap}(\anodeB)) = \labs{\replace{\agraph}{\agraphB}{\agraphC'}}(\anode)$.
\end{itemize}

By Proposition \ref{prop:subs-minus} we have $\anode \in \innernodes{\replace{\agraph'}{\agraphC}{\agraphC'}} \setminus \nodes{\agraphC'}$ if and only if $\alpha \in \innernodes{\agraph'} \setminus \nodes{\agraphC}$ if and only if $\anode \in \innernodes{\agraph}\setminus \nodes{\agraphB}$. This proves that $\labs{\replace{\agraph'}{\agraphC}{\agraphC'}}(\anode) = \labs{\replace{\agraph}{\agraphB}{\agraphC'}}(\anode)$.
\end{enumerate}
%
%
\end{proof}
\begin{corollary}
	Assume $\agraphB'$ is substitutable for $\agraphB$ in $\agraph$ and let $\agraph' \isdef \replace{\agraph}{\agraphB}{\agraphB'}$. Then $\agraphB$ is substitutable for $\agraphB'$ in $\agraph'$ and $\replace{\agraph'}{\agraphB'}{\agraphB} = \agraph$.
\end{corollary}

\begin{proof}
	We have $\nodes{\agraph'} \cap \nodes{\agraphB} = \left[(\nodes{\agraph}\setminus \nodes{\agraphB}) \cup \nodes{\agraphB'}\right] \cap \nodes{\agraphB} = \nodes{\agraphB'} \cap \nodes{\agraphB} \subseteq \nodes{\agraphB'}$, which show that $\agraphB$ is substitutable for $\agraphB'$ in $\agraph'$. By Lemma \ref{lem:replaceseq} and Proposition \ref{prop:trivial_replacement}, we deduce that $\replace{\replace{\agraph}{\agraphB}{\agraphB'}}{\agraphB'}{\agraphB} = \replace{\agraph}{\agraphB}{\agraphB} = \agraph$.
\end{proof}

\begin{lemma}
	\label{lem:replace_subs}
	Consider the graphs $\agraph, \agraphB, \agraphB', \agraphC$, and assume that $\agraphB\sgr \agraph$, $\agraphC \sgr \agraph$ and $\agraphB'$ is \substitutable for $\agraphB$ in $\agraph$. If $\nodes{\agraphB} \cap \nodes{\agraphC} = \emptyset$ then $\agraphC \sgr \replace{\agraph}{\agraphB}{\agraphB'}$.
\end{lemma}

\begin{proof}
	Let $\agraph' \isdef \replace{\agraph}{\agraphB}{\agraphB'}$ and $\amap \isdef \amapsub{\agraphB}{\agraphB'}{\agraph}$. Note that since $\agraphB'$ is \substitutable for $\agraphB$ in $\agraph$, we have
	\[\nodes{\agraphB'} \cap \nodes{\agraphC}\  \subseteq \ \nodes{\agraphB'} \cap \nodes{\agraph} \ \subseteq\  \nodes{\agraphB} \cap \nodes{\agraphC}\  =\  \emptyset,\]
	and that since $\nodes{\agraphB} \cap \nodes{\agraphC} = \emptyset$, if $\anode \in \nodes{\agraphC}$ then $\amap(\anode) = \anode$.
	We verify that $\agraphC \sgr \agraph'$:
	\begin{enumerate}
		\item{Since $\agraphC \sgr \agraph$, by definition $\nodes{\agraphC} \subseteq \nodes{\agraph}$. By hypothesis $\nodes{\agraphB} \cap \nodes{\agraphC} = \emptyset$, thus $\nodes{\agraphC} \subseteq \nodes{\agraph} \setminus \nodes{\agraphB}  \subseteq \nodes{\agraph'}$.
		}
		\item{If $\mkedge{\anode}{\anodeB} \in \edges{\agraphC}$ then 
			$\mkedge{\anode}{\anodeB} \in \edges{\agraph}$, because $\agraphC \sgr \agraph$.
			Since $\nodes{\agraphB} \cap \nodes{\agraphC} = \emptyset$ by hypothesis, $\mkedge{\anode}{\anodeB} \not \in \edges{\agraphB}$. Thus
			$\mkedge{\anode}{\anodeB} \in (\edges{\agraph} \setminus \edges{\agraphB}) \cup \edges{\agraphB'}$, and since $\amap(\anode) = \anode$ and $\amap(\anodeB) = \anodeB$, we deduce that $\mkedge{\anode}{\anodeB} \in \amap(\edges{\agraph} \setminus \edges{\agraphB}) \cup \edges{\agraphB'} = \edges{\agraph'}$.}
		\item{Let $\anode,\anodeB\in \nodes{\agraphC}$, and assume that
			$\mkedge{\anode}{\anodeB} \in \edges{\agraph'} = \amap(\edges{\agraph} \setminus \edges{\agraphB}) \cup \edges{\agraphB'}$. Since $\nodes{\agraphB'} \cap \nodes{\agraphC} = \emptyset$, we cannot have $\mkedge{\anode}{\anodeB} \in \edges{\agraphB'}$.
			Since $\amap(\anode) = \anode$ and $\amap(\anodeB) = \anodeB$, necessarily $\mkedge{\anode}{\anodeB} \in \edges{\agraph}\setminus \edges{\agraphB} \subseteq \edges{\agraph}$, and since $\agraphC\sgr \agraph$, we deduce that $\mkedge{\anode}{\anodeB} \in \edges{\agraphC}$.
		}
		\item{Assume that $\anyedge{\anode}{\anodeB} \in \edges{\agraph'} = \amap(\edges{\agraph} \setminus \edges{\agraphB}) \cup \edges{\agraphB'}$, where 
			$\anode\in \nodes{\agraphC}$ and $\anodeB\not \in \nodes{\agraphC}$. Then $\amap(\anode) = \anode$ and since $\nodes{agraphB'} \cap \nodes{\agraphC} = \emptyset$, we cannot have $\anyedge{\anode}{\anodeB} \in \edges{\agraphB'}$. Let $\anodeB' \isdef \inv{\amap}(\anodeB)$, note that we cannot have $\anodeB \in \nodes{\agraphC}$ because otherwise we would have $\amap(\anodeB') = \anodeB' = \anodeB \notin \nodes{\agraphC}$. Thus, $\anyedge{\anode}{\anodeB'} \in \edges{\agraph} \setminus \edges{\agraphB} \subseteq \edges{\agraph}$, and since $\agraphC \sgr \agraph$, we deduce that $\anode \in \inputs{\agraphC}$.}

		\item{Assume $\anode\in \inputs{\agraph'} \cap \nodes{\agraphC} = \amap(\inputs{\agraph}) \cap \nodes{\agraphC}$. Then since $\amap(\anode) = \anode$, we have
			$\anode\in \inputs{\agraph} \cap \nodes{\agraphC}$, hence
			$\anode\in \inputs{\agraphC}$ because $\agraphC \sgr \agraph$.}
		
		\item{Consider $\anode \in \innernodes{\agraphC} \subseteq \innernodes{\agraph}$. Then by hypothesis, $\anode \in \innernodes{\agraph} \setminus \nodes{\agraphB}$ and by Proposition \ref{prop:subs-minus} we have $\anode \in \innernodes{\agraph'} \setminus \nodes{\agraphB'}$. By definition, $\labs{\agraph'}(\anode) = \labs{\agraph}(\anode) = \labs{\agraphC}(\anode)$, because $\agraphC\sgr \agraph$.}
	\end{enumerate}
\end{proof}

\begin{lemma}
	\label{lemma:replace_commute}
	Let $\agraph$, $\agraphB_i, \agraphB_i'$ be {\graph}s for $i = 1,2$, where 
	$\agraphB_i \sgr \agraph$ and
	$\agraphB_i'$ is \substitutable for $\agraphB_i$ in $\agraph$. 
	If $\nodes{\agraphB_1'} \cap \nodes{\agraphB_2'} = \nodes{\agraphB_1} \cap \nodes{\agraphB_2} = \emptyset$ then, for $i,j \in \set{1,2}$ with $i\neq j$,  $\agraphB_i'$  is \substitutable for $\agraphB_i$ in $\replace{\agraph}{\agraphB_{j}}{\agraphB_{j}'}$ and
	\[\replace{(\replace{\agraph}{\agraphB_{1}}{\agraphB_{1}'})}{\agraphB_{2}}{\agraphB_{2}'}
	=
	\replace{(\replace{\agraph}{\agraphB_{2}}{\agraphB_{2}'})}{\agraphB_{1}}{\agraphB_{1}'}.\]
\end{lemma}
\begin{proof}
	By Lemma \ref{lem:replace_subs}, $\agraphB_1 \sgr \replace{\agraph}{\agraphB_{2}}{\agraphB_{2}'}$ and $\agraphB_2 \sgr \replace{\agraph}{\agraphB_{1}}{\agraphB_{1}'}$.
	We show that $\agraphB_1'$ is \substitutable for $\agraphB_1$ in $\replace{\agraph}{\agraphB_{2}}{\agraphB_{2}'}$; the proof that $\agraphB_2'$  is \substitutable for
	$\agraphB_2$ in $\replace{\agraph}{\agraphB_{1}}{\agraphB_{1}'}$ is symmetric.	Let $\anode \in \nodes{\agraphB_1'} \cap 
	\nodes{\replace{\agraph}{\agraphB_{2}}{\agraphB_{2}'}} = \nodes{\agraphB_1'} \cap \left[(\nodes{\agraph} \setminus \nodes{\agraphB_2}) \cup \nodes{\agraphB_2'}\right]$. By hypothesis 	$\nodes{\agraphB_1'} \cap \nodes{\agraphB_2'} = \emptyset$, hence $\anode \in \nodes{\agraph} \setminus \nodes{\agraphB_2} \subseteq \nodes{\agraph}$, and	since $\agraphB_1'$ is \substitutable for $\agraphB_1$ in $\agraph$,	we deduce that $\anode\in \nodes{\agraphB_1}$, hence the result.
	
	Let $\agraph'_i \isdef \replace{\agraph}{\agraphB_{i}}{\agraphB_{i}'}$	and $\agraph''_i \isdef \replace{\agraph_{3-i}'}{\agraphB_{i}}{\agraphB_{i}'}$.	We show that $\agraph''_1 = \agraph''_2$.
	First note that $\nodes{\agraphB_1} \cap \nodes{\agraphB_2'} = \nodes{\agraphB_2} \cap \nodes{\agraphB_1'} = \emptyset$. Indeed, since $\agraphB_1 \sgr \agraph$ and $\agraphB_2'$ is \substitutable for $\agraphB_2$, we have $\nodes{\agraphB_1} \cap \nodes{\agraphB_2'} \subseteq \nodes{\agraphB_1} \cap \nodes{\agraph} \cap \nodes{\agraphB_2'} \subseteq \nodes{\agraphB_1} \cap \nodes{\agraphB_2} = \emptyset$. For $i = 1, 2$, we define 
	$\amap_i \isdef \amapsub{\agraphB_i}{\agraphB_i'}{\agraph}$ and $\amap_i' \isdef \amapsub{\agraphB_{i}}{\agraphB_{i}'}{\agraph_{3-i}'}$.
	Note that by definition, for $i = 1,2$, if $\anode\in \inputs{\agraphB_{i}}$ then $\amap_i(\anode) = \amap_i'(\anode) \in \inputs{\agraphB'_i}$, so that $\amap_{3-i}(\anode) = \anode$ and $\amap_{3-i}'(\amap_i(\anode)) = \amap_i(\anode)$. Similarly, if $\anode \in \nodes{\agraph} \setminus (\nodes{\agraphB_1} \cup \nodes{\agraphB_2})$ then for $i = 1, 2$, $\amap_i(\anode) = \amap_i'(\anode) = \anode$. 
	\begin{itemize}
		\item{We have $\nodes{\agraph''_1} = (\nodes{\agraph'_2} \setminus \nodes{\agraphB_1}) \cup \nodes{\agraphB_1'} 
			= ((\nodes{\agraph} \setminus \nodes{\agraphB_2}) \cup \nodes{\agraphB_2'})  \setminus \nodes{\agraphB_1}) \cup \nodes{\agraphB_1'}$.
			Since $\nodes{\agraphB_1} \cap \nodes{\agraphB_2'} = \nodes{\agraphB_2} \cap \nodes{\agraphB_1'} = \emptyset$, we deduce that
			\[
			\begin{array}{lll}
				\nodes{\agraph''_1} & = & ((\nodes{\agraph} \cup \nodes{\agraphB_1'} \cup \nodes{\agraphB_2'}) \setminus (\nodes{\agraphB_1} \cup \nodes{\agraphB_2}) \\
				& = & ((\nodes{\agraph} \setminus \nodes{\agraphB_1}) \cup \nodes{\agraphB_1'})  \setminus \nodes{\agraphB_2}) \cup \nodes{\agraphB_2'} \\
				& = & (\nodes{\agraph'_1} \setminus \nodes{\agraphB_2}) \cup \nodes{\agraphB_2'} 
				=  \nodes{\agraph''_2}.\\
			\end{array}
			\]
		}
		\item{By definition we have $\inputs{\agraph''_1} = \amap_1'(\inputs{\agraph'_2}) = \amap_1'(\amap_2(\inputs{\agraph}))$ and $\inputs{\agraph''_2} = \amap_2'(\inputs{\agraph'_1}) = \amap_2'(\amap_1(\inputs{\agraph}))$. Consider $\anode \in \nodes{\agraph}$. If $\anode \in \inputs{\agraphB_i}$ for $i= 1, 2$, then we have
			\[\amap_1'(\amap_2(\anode))\ =\ \amap_1'(\anode)\ =\ \amap_1(\anode)\ =\ \amap_2'(\amap_1(\anode)).\]
		Otherwise $\anode \in \inputs{\agraph} \setminus (\nodes{\agraphB_1} \cup \nodes{\agraphB_2})$, so that 
			\[\amap_1'(\amap_2(\anode))\ =\ \amap_1'(\anode)\ =\ \anode\ =\ \amap_1(\anode)\ =\ \amap_2'(\amap_1(\anode)).\]
			We conclude that $\inputs{\agraph_1''} = \inputs{\agraph_2''}$.
		}		
		\item{For $i = 1, 2$, we have $\amap_{3-i}(\edges{\agraphB_i}) = \edges{\agraphB_i}$ and $\amap_{3-i}'(\edges{\agraphB_i'}) = \edges{\agraphB_i'}$, hence 
			\begin{eqnarray*}
				\edges{\agraph''_1}& =& \amap_2'(\edges{\agraph_2'} \setminus \edges{\agraphB_2}) \cup \edges{\agraphB_2'}\\
				& = & \amap_2'( (\amap_1(\edges{\agraph} \setminus \edges{\agraphB_1}) \cup \edges{\agraphB_1'}) \setminus \edges{\agraphB_2}) \cup \edges{\agraphB_2'}\\
				& = & \amap_2'(\amap_1(\edges{\agraph} \setminus \edges{\agraphB_1}) \setminus \edges{\agraphB_2})  \cup \edges{\agraphB_1'} \cup \edges{\agraphB_2'}\\
				& = & \amap_2'(\amap_1(\edges{\agraph} \setminus (\edges{\agraphB_1} \cup \edges{\agraphB_2})) \cup \edges{\agraphB_1'} \cup \edges{\agraphB_2'}\\
				& = & \amap_1'(\amap_2(\edges{\agraph} \setminus (\edges{\agraphB_1} \cup \edges{\agraphB_2})) \cup \edges{\agraphB_1'} \cup \edges{\agraphB_2'}\\
				& = & \edges{\agraph''_2}
			\end{eqnarray*}
			The second to last line stems from the relation 
			$\forall \anode\in \nodes{\agraph}$, 
			$\amap_1'(\amap_2(\anode) = \amap_2'(\amap_1(\anode))$ established in the previous item.
		}
	\item Since $\nodes{\agraph_1''} = \nodes{\agraph_2''}$ and $\inputs{\agraph_1''} = \inputs{\agraph_2''}$, we have $\innernodes{\agraph_1''} = \innernodes{\agraph_2''}$. We show that for all $\anode \in \innernodes{\agraph_1''}$, $\labs{\agraph_1''}(\anode) = \labs{\agraph_2''}(\anode)$.
	\begin{itemize}
		\item By Proposition \ref{prop:subs-minus} we have 
		$\innernodes{\agraph_2''} \setminus (\nodes{\agraphB_1'} \cup \nodes{\agraphB_2'})\ =\ \innernodes{\agraph_1''} \setminus (\nodes{\agraphB_1'} \cup \nodes{\agraphB_2'})\ =\ \innernodes{\agraph_2'} \setminus (\nodes{\agraphB_1} \cup \nodes{\agraphB_2'})\ =\ \innernodes{\agraph} \setminus (\nodes{\agraphB_1} \cup \nodes{\agraphB_2})$. Thus, if $\anode \in \innernodes{\agraph_2''} \setminus (\nodes{\agraphB_1'} \cup \nodes{\agraphB_2'})$ then $\labs{\agraph_1''}(\anode) = \labs{\agraph}(\anode) = \labs{\agraph_2''}(\anode)$.
		\item If $\anode \in \innernodes{\agraphB_2'}$ then $\labs{\agraph_1''}(\anode) = \labs{\agraph_2'}(\anode) = \labs{\agraphB_2'}(\anode) = \labs{\agraph_2''}(\anode)$. The case where $\anode \in \innernodes{\agraphB_1'}$ is proved in a similar way.
		\item If $\anode \in \innernodes{\agraph_1''} \cap \inputs{\agraphB_1'}$ then $\anode \in\innernodes{\agraph_2''} \setminus \nodes{\agraphB_2'} $ and $\anode = \amap_1(\anodeB) = \amap_1'(\anodeB)$ for some element $\anodeB\in \inputs{\agraphB_1}$. Since $\anodeB \in \innernodes{\agraph_2'}\setminus \nodes{\agraphB_2'} = \innernodes{\agraph}\setminus \nodes{\agraphB_2}$ (Proposition \ref{prop:subs-minus}), we have $\labs{\agraph_1''}(\anode) = \labs{\agraph_2'}(\anodeB) = \labs{\agraph}(\anodeB) = \labs{\agraph_1'}(\anode) = \labs{\agraph_2''}(\anode)$. The case where $\anode \in \innernodes{\agraph_2''} \cap \inputs{\agraphB_2'}$ is proved in a similar way.
	\end{itemize}
		
	\end{itemize}  
\end{proof}

\begin{lemma}
\label{lem:replace_subgraph}
Let $\agraph$, $\agraphB$ and $\agraphC$ be {\graph}s such that 
$\agraphC \sgr \agraphB \sgr \agraph$, and let $\agraphC'$ be a \graph \substitutable for $\agraphC$ in $\agraph$.
Then $\agraphC'$ is \substitutable for $\agraphC$ in $\agraphB$,
$\replace{\agraphB}{\agraphC}{\agraphC'}$ 
is \substitutable for $\agraphB$ in $\agraph$,
and:
\[\replace{\agraph}{\agraphC}{\agraphC'} = 
\replace{\agraph}{\agraphB}{\replace{\agraphB}{\agraphC}{\agraphC'}}\]
\end{lemma}
\begin{proof}
We have $\nodes{\agraphB} \cap \nodes{\agraphC'} \subseteq \nodes{\agraph} \cap \nodes{\agraphC'}
\subseteq \nodes{\agraphC}$ (since $\agraphC'$ is \substitutable for $\agraphC$ in $\agraph$).
Therefore, $\agraphC'$ is \substitutable for $\agraphC$ in $\agraphB$.
Similarly, $\nodes{\agraph} \cap \nodes{\replace{\agraphB}{\agraphC}{\agraphC'}} \subseteq
\nodes{\agraph} \cap (\nodes{\agraphB} \cup \nodes{\agraphC'}) 
= \nodes{\agraphB} \cup (\nodes{\agraph} \cap \nodes{\agraphC'}) \subseteq \nodes{\agraphB} \cup \nodes{\agraphC} = \nodes{\agraphB}$, hence $\replace{\agraphB}{\agraphC}{\agraphC'}$ 
is \substitutable for $\agraphB$ in $\agraph$.

Let $\agraph' \isdef \replace{\agraph}{\agraphC}{\agraphC'}$, 
$\agraphB' \isdef \replace{\agraphB}{\agraphC}{\agraphC'}$ and 
$\agraph'' \isdef \replace{\agraph}{\agraphB}{\agraphB'}$; we show that $\agraph' = \agraph''$.

Let $\amap_1 \isdef \amapsub{\agraphC}{\agraphC'}{\agraph}$, $\amap_2 \isdef \amapsub{\agraphC}{\agraphC'}{\agraphB}$ and $\amap_3 \isdef \amapsub{\agraphB}{\agraphB'}{\agraph}$. Note that by definition, $\amap_2$ is the restriction of $\amap_1$ to the nodes in $(\nodes{\agraphB} \setminus \nodes{\agraphC}) \cup \inputs{\agraphC}  = \nodes{\agraphB} \setminus \innernodes{\agraphC}$. We show that for all $\anode\in \nodes{\agraph} \setminus \innernodes{\agraphB} = (\nodes{\agraph} \setminus \nodes{\agraphB}) \cup \inputs{\agraphB}$, we have $\amap_3(\anode) = \amap_1(\anode)$. If $\anode \in \nodes{\agraph} \setminus \nodes{\agraphB}$ then since $\nodes{\agraphC} \subseteq \nodes{\agraphB}$, we have $\amap_3(\anode) = \amap_1(\anode) = \anode$. Otherwise, $\anode \in \inputs{\agraphB}$ and by Proposition \ref{prop:amap_image}, $\amap_3(\anode) = \amap_2(\anode)$. Since $\inputs{\agraphB} \subseteq \nodes{\agraphB} \setminus \innernodes{\agraphC}$ by Proposition \ref{prop:sub_inputs_include}, we deduce that $\amap_3(\anode) = \amap_2(\anode) = \amap_1(\anode)$.

\begin{itemize}
\item{By Definition \ref{def:replace} (\ref{replace:nodes}), 
$\nodes{\agraph''}  = (\nodes{\agraph} \setminus \nodes{\agraphB}) \cup \nodes{\agraphB'} 
= (\nodes{\agraph} \setminus \nodes{\agraphB}) \cup (\nodes{\agraphB} \setminus \nodes{\agraphC}) \cup \nodes{\agraphC'} = (\nodes{\agraph} \setminus \nodes{\agraphC}) \cup \nodes{\agraphC'} = \nodes{\agraph'}$.}

\item{By Definition \ref{def:replace} (\ref{replace:inputs}),
 $\inputs{\agraph''} = \amap_3(\inputs{\agraph})$ and
$\inputs{\agraph'} = \amap_1(\inputs{\agraph})$.
Since $\inputs{\agraph} \subseteq \nodes{\agraph} \setminus \innernodes{\agraphB}$ by Proposition \ref{prop:sub_inputs_include}, we have $\amap_3(\inputs{\agraph}) = \amap_1(\inputs{\agraph})$ so that
$\inputs{\agraph''} = \inputs{\agraph'}$.
}
\item{Consider an edge $\anyedge{\anode}{\anodeB}$ in $\edges{\agraphB} \setminus \edges{\agraphC}$. Then since $\agraphC \sgr \agraphB$, we cannot have $\set{\anode, \anodeB} \subseteq \nodes{\agraphC}$ and if  $\anode \in \nodes{\agraphC}$ then necessarily $\anode \in \inputs{\agraphC}$. This shows that $\set{\anode, \anodeB} \subseteq \nodes{\agraphB} \setminus \innernodes{\agraphC}$ and therefore, $\amap_2(\edges{\agraphB} \setminus \edges{\agraphC}) = \amap_1(\edges{\agraphB} \setminus \edges{\agraphC})$. The fact that $\amap_3(\edges{\agraph} \setminus \edges{\agraphB}) = \amap_1(\edges{\agraph} \setminus \edges{\agraphB})$ is proved in a similar way. We have
	\begin{eqnarray*}
		\edges{\agraph''} &= &\amap_3(\edges{\agraph} \setminus \edges{\agraphB}) \cup \edges{\agraphB'}\\
		& = & \amap_3(\edges{\agraph} \setminus \edges{\agraphB}) \cup \amap_2(\edges{\agraphB}\setminus \edges{\agraphC}) \cup \edges{\agraphC'}\\
		& = & \amap_1(\edges{\agraph} \setminus \edges{\agraphB}) \cup \amap_1(\edges{\agraphB}\setminus \edges{\agraphC}) \cup \edges{\agraphC'}\\
		& = & \amap_1\left((\edges{\agraph} \setminus \edges{\agraphB}) \cup (\edges{\agraphB}\setminus \edges{\agraphC})\right) \cup \edges{\agraphC'} \text{\quad because $\amap_1$ is injective}\\
		& = & \amap_1\left(\edges{\agraph} \setminus \edges{\agraphC}\right) \cup \edges{\agraphC'}\\
		& = & \edges{\agraph'}
	\end{eqnarray*}	
 }
 \item{Consider a node $\anode \in \innernodes{\agraph'}$. Since $\nodes{\agraph'}  = \nodes{\agraph''}$ and $\inputs{\agraph'} = \inputs{\agraph''}$, we also have $\anode \in \innernodes{\agraph''}$. First assume that $\anode \in \innernodes{\agraph'}\setminus \nodes{\agraphC'}$, so that $\labs{\agraph'}(\anode) = \labs{\agraph}(\anode)$.
 	\begin{itemize}
 		\item If $\anode \in \innernodes{\agraph''} \setminus \nodes{\agraphB'}$ then by definition $\labs{\agraph''}(\anode) = \labs{\agraph}(\anode)$.
 		\item If $\anode \in \innernodes{\agraphB'}$ then since $\anode \notin \nodes{\agraphC'}$ by hypothesis and $\agraphB \sgr \agraph$, we have $\labs{\agraph''}(\anode) = \labs{\agraphB'}(\anode) = \labs{\agraphB}(\anode) = \labs{\agraph}(\anode)$.
 		\item If $\inv{\amap_3}(\anode) \in \innernodes{\agraph} \cap \inputs{\agraphB}$ then by definition $\labs{\agraph''}(\anode) = \labs{\agraph}(\inv{\amap_3}(\anode))$. But since $\anode \notin \nodes{\agraphC}$, we must have $\inv{\amap_3}(\anode) = \inv{\amap_2}(\anode) = \anode$, hence the result.
 	\end{itemize}
 Assume that $\anode \in \innernodes{\agraphC'}$, so that $\labs{\agraph'}(\anode) = \labs{\agraphC'}(\anode)$.
 \begin{itemize}
 	\item We cannot have $\anode \in \innernodes{\agraph''} \setminus \nodes{\agraphB'}$ because $\innernodes{\agraphC'} \subseteq \nodes{\agraphB'}$.
 	\item If $\anode \in \innernodes{\agraphB'}$ then  $\labs{\agraph''}(\anode) = \labs{\agraphB'}(\anode) = \labs{\agraphC'}(\anode)$ because $\anode \in \innernodes{\agraphC'}$ by hypothesis.
 	\item We cannot have $\inv{\amap_3}(\anode) \in \innernodes{\agraph} \cap \inputs{\agraphB}$ because otherwise we would have $\anode \in \amap_3(\inputs{\agraphB}) = \inputs{\agraphB'}$ and since $\agraphC' \sgr \agraphB'$, necessarily $\anode \in \inputs{\agraphC'}$.
 \end{itemize}
Now assume that 
$\anode \in \innernodes{\agraph'} \cap \inputs{\agraphC'}$, i.e.,  $\inv{\amap_1}(\anode) \in \innernodes{\agraph} \cap \inputs{\agraphC}$, so that $\labs{\agraph'}(\anode) = \labs{\agraph}(\inv{\amap_1}(\anode))$.
\begin{itemize}
	\item We cannot have $\anode \in \innernodes{\agraph''} \setminus \nodes{\agraphB'}$ because $\inputs{\agraphC'} \subseteq \nodes{\agraphB'}$.
	\item If $\anode \in \innernodes{\agraphB'}$ then $\labs{\agraph''}(\anode) = \labs{\agraphB'}(\anode) = \labs{\agraphB}(\inv{\amap_2}(\anode)) = \labs{\agraphB}(\inv{\amap_1}(\anode)) = \labs{\agraph}(\inv{\amap_1}(\anode))$.
	\item If $\inv{\amap_3}(\anode) \in \innernodes{\agraph} \cap \inputs{\agraphB} \subseteq \innernodes{\agraph} \setminus \innernodes{\agraphB}$  then $\labs{\agraph''}(\anode) = \labs{\agraph}(\inv{\amap_3}(\anode)) = \labs{\agraph}(\inv{\amap_1}(\anode))$ because $\amap_1$ and $\amap_3$ coincide on $\innernodes{\agraph} \setminus \innernodes{\agraphB}$.
\end{itemize}
 	
%
%
}
 \end{itemize}
\end{proof}

\newcommand{\compatible}{label-compatible\xspace}

\subsection{\capitalisewords{graph rewriting systems}}

In what follows, we define the notion of a \emph{graph rewrite system} that is based on the replacement operation. In order to define critical pairs on such a system, we introduce the notion of a \emph{merge} of two graphs. Intuitively, the merge of graphs $\agraph_1$ and $\agraph_2$ will be any graph $\agraph$ containing all nodes and edges occurring in either $\agraph_1$ or $\agraph_2$, and possibly some additional edges.

\begin{definition}
	\label{def:merge}
	Two {\graph}s $\agraph_1$ and $\agraph_2$ are {\em \compatible} iff
	the set of pairs $\set{\tuple{\labs{\agraph_1}(\anodeB),\labs{\agraph_2}(\anodeB))} \mid \anodeB \in \nodes{\agraph_1} \cap\nodes{\agraph_2}}$ admits 
	an mgu. 
	
	
	Given two \compatible\ {\graph}s $\agraph_1$, $\agraph_2$ with mgu $\sigma$ and a set of edges $E$ of the form $\anyedge{\anode}{\anodeB}$ such that $\anode \in \inputs{\agraph_1} \setminus \nodes{\agraph_2}$ and $\anodeB \in \inputs{\agraph_2} \setminus \nodes{\agraph_1}$, the \emph{$E$-\merge of $\agraph_1$ and $\agraph_2$} is the graph $\agraph$ defined as follows:	
	\begin{itemize}
		\item $\nodes{\agraph} \isdef \nodes{\agraph_1} \cup \nodes{\agraph_2}$;
		\item 
		$\inputs{\agraph} \isdef (\inputs{\agraph_1} \cup \inputs{\agraph_2}) \setminus (\innernodes{\agraph_1} \cup \innernodes{\agraph_2})$ (the order of the nodes in $\inputs{\agraph}$ is chosen arbitrarily);
		\item $\edges{\agraph} \isdef \edges{\agraph_1} \cup \edges{\agraph_2} \cup E$;
		\item for every node $\anode \in \nodes{\agraph_i}$ with $i = 1,2$,	$\labs{\agraph}(\anode) \isdef \sigma(\labs{\agraph_i}(\anode))$.
	\end{itemize}
	A {\em \merge} of two \compatible\ {\graph}s $\agraph_1$, $\agraph_2$ is a \graph $\agraph$  that is an $E$-merge of $\agraph_1$ and $\agraph_2$, for some set $E$ of edges of the form $\anyedge{\anode}{\anodeB}$ such that $\anode \in \inputs{\agraph_1} \setminus \nodes{\agraph_2}$ and $\anodeB \in \inputs{\agraph_2} \setminus \nodes{\agraph_1}$.		
\end{definition}

\begin{proposition}
	\label{prop:merge}
	Assume $\agraphB,\agraphB' \sgr \agraph$ and consider the set of edges
	\[E\ \isdef\ \setof{\anyedge{\anode}{\anodeB}}{\anode \in \inputs{\agraphB} \setminus \nodes{\agraphB'},\, \anodeB\in \inputs{\agraphB'} \setminus \nodes{\agraphB},\, \anyedge{\anode}{\anodeB} \in \edges{\agraph}}.\]
	 Then
	$\agraphB$ and $\agraphB'$ are \compatible and if $\agraph'$ is the $E$-\merge of $\agraphB$ and $\agraphB'$, then $\agraph' \sgr \agraph$. 
\end{proposition}
\begin{proof}
	It is straightforward to verify that $\agraphB$ and $\agraphB'$ are \compatible, with the empty mgu. 
	We show that $\agraph' \sgr \agraph$, using the fact that, by hypothesis, $\agraphB,\agraphB' \sgr \agraph$.
	\begin{enumerate}
		\item By Definition \ref{def:merge}, we have $\nodes{\agraph'} = \nodes{\agraphB} \cup \nodes{\agraphB'} \subseteq \nodes{\agraph}$.
		\item All the edges in $E$ are also in $\edges{\agraph}$, it is thus straightforward to verify that $\edges{\agraph'} \subseteq \edges{\agraph}$.
		\item Assume that $\mkedge{\anode}{\anodeB} \in \edges{\agraph}$, where $\anode, \anodeB \in \nodes{\agraph'}$. If $\anode, \anodeB \in \nodes{\agraphB}$ or $\anode, \anodeB \in \nodes{\agraphB'}$ then  $\mkedge{\anode}{\anodeB} \in \edges{\agraph'}$ since $\agraphB,\agraphB' \sgr \agraph$. Suppose that $\anode \in \nodes{\agraphB} \setminus \nodes{\agraphB'}$ and $\anodeB \in \nodes{\agraphB'} \setminus \nodes{\agraphB}$.  Since $\agraphB,\agraphB' \sgr \agraph$, by Definition \ref{def:subgraph} (\ref{subgraph:inputs}),  necessarily $\anode \in \inputs{\agraphB}$ and $\anodeB \in \inputs{\agraphB'}$, hence $\mkedge{\anode}{\anodeB} \in E \subseteq \edges{\agraph'}$. The case where $\anode \in \nodes{\agraphB} \setminus \nodes{\agraphB'}$ and $\anodeB \in \nodes{\agraphB'} \setminus \nodes{\agraphB}$ is proved in a similar way. 
		\item Consider $\anyedge{\anode}{\anodeB} \in \edges{\agraph}$, where $\anode \in \nodes{\agraph'}$ and $\anodeB \notin \nodes{\agraph'}$. Assume $\anode \in \nodes{\agraphB}$, the case where $\anode \in \nodes{\agraphB'}$ is proved in a similar way. Since $\anodeB \notin \nodes{\agraph'}$ we have $\anodeB \notin \nodes{\agraphB}$, hence $\anode \in \inputs{\agraphB}$ by Definition \ref{def:subgraph} (\ref{subgraph:inputs}). It cannot be the case that $\anode \in \innernodes{\agraphB'}$ because otherwise we would have $\anode \in \nodes{\agraphB'}$ and $\anodeB \notin \nodes{\agraphB'}$, so that $\anode \in \inputs{\agraphB'}$, a contradiction. We conclude that $\anode \in \inputs{\agraph'}$.
		\item Consider a node $\anode \in \inputs{\agraph} \cap \nodes{\agraph'}$ and suppose w.l.o.g.\ that $\anode \in \nodes{\agraphB}$. Since $\agraphB \sgr \agraph$, by Definition \ref{def:subgraph} (\ref{subgraph:inputs_bis}) necessarily $\anode \in \inputs{\agraphB}$. If $\anode \in \nodes{\agraphB'}$ then $\anode \in \inputs{\agraphB'}$, we deduce that $\anode \in \inputs{\agraphB} \setminus (\innernodes{\agraphB} \cup \innernodes{\agraphB'}) \subseteq \inputs{\agraph'}$.
		\item By definition (since the considered \substitution $\sigma$ is empty)
		$\labs{\agraph'}$ is the restriction of $\labs{\agraph}$ to $\innernodes{\agraph'}$.
	\end{enumerate}
\end{proof}

 \newcommand{\gatesort}[2]{\mathtt{gate}_{#1}^{#2}}
 \newcommand{\inportsort}{\mathtt{into}}
  \newcommand{\outportsort}{\mathtt{from}}


\newcommand{\crel}{$\aclass$-relation\xspace}

\begin{definition}
\label{def:sgrc}
Let $\aclass$ be a  set of {\graph}s, such that $\agraph \in \aclass \wedge \agraph \iso \agraph' \Rightarrow \agraph' \in \aclass$.
A {\em \crel} $\sgrc$ is a reflexive and transitive binary relation included in $\sgr$ and satisfying  the following properties:
\begin{enumerate}
\item{If $\agraphB \sgrc \agraph$ then $\agraph, \agraphB\in \aclass$. 
}
\item{If $\agraphB \sgrc \agraph$, $\agraphB'\in \aclass$ and $\agraphB'$ is \substitutable for $\agraphB$ in $\agraph$ then 
$\agraphB' \sgrc \replace{\agraph}{\agraphB}{\agraphB'}$ (thus $\replace{\agraph}{\agraphB}{\agraphB'} \in \aclass$). \label{sgrc:preserve}}
\item{If $\agraphB \sgrc \agraph$, then $\amap(\agraphB) \sgrc \amap(\agraph)$, for every \nsubstitution $\amap$. \label{sgrc:iso}}
\item{If $\agraphB,\agraphC \sgrc \agraph$, $\agraphB$ and $\agraphC$ are disjoint, $\agraphC' \in \aclass$ and 
$\agraphC'$ 
is \substitutable for $\agraphC$ in $\agraph$, then 
$\agraphB \sgrc \replace{\agraph}{\agraphC}{\agraphC'}$. 
\label{sgrc:replace}
}
\item{If $\agraph_i \sgrc \agraph$, for $i = 1,2$, 
	then there exists a \merge $\agraph'$ of
$\agraph_1$ and $\agraph_2$, such that $\agraph' \sgrc \agraph$\label{sgrc:merge} 
}
\item{
	If $\agraph \sgrc \agraphB$ then for every \substitution $\sigma$,
$\sigma(\agraph) \sgrc \sigma(\agraphB)$.\label{sgrc:subst}}
\end{enumerate}
\end{definition}

In the following we consider a fixed set $\aclass$ and a fixed \crel $\sgrc$. The rewrite relation is parameterized by $\sgrc$.

\newcommand{\system}{\graph rewrite system\xspace}

\begin{definition}
\label{def:rr}
A {\em \rewriterule} is a pair written $\agraph \rightarrow \agraphB$, where 
$\agraph$ and $\agraphB$ are   {\graph}s such that 
 $\inputs{\agraph}$ and $\inputs{\agraphB}$ are \rootsimilar. 
A {\em \system} is a set of {\rewriterule}s.
\end{definition}


\begin{definition}
\label{def:rw}
If $\rules$ is a \system
then we write $\agraph \rightarrow_{\rules} \agraph'$
iff there exists a rule $\agraphB \rightarrow \agraphB' \in \rules$, a \substitution $\sigma$ and an \nsubstitution $\amap$ of domain $\nodes{\agraphB} \cup \nodes{\agraphB'}$ such that
$\sigma(\amap(\agraphB)) \sgrc \agraph$,
$\sigma(\amap(\agraphB'))$ is \substitutable for $\sigma(\amap(\agraphB))$ in $\agraph$, 
 and $\agraph' \iso \replace{\agraph}{\sigma(\amap(\agraphB))}{\sigma(\amap(\agraphB'))}$.
We denote by  $\rightarrow_\rules^+$
the transitive closure of $\rightarrow_\rules$, respectively.
We write $\agraph \rightarrow^{k}_\rules \agraph'$ 
 if there exists 
 a sequence of {\graph}s $\agraphB_i$ (for $i = 1,\dots,k$)
 with $\agraphB_1 \iso \agraph$, $\agraph_k \iso \agraph'$
 and for every $i \in \interv{1}{k-1}$, 
 $\agraphB_i \rightarrow_{\rules} \agraphB_{i+1}$.
By a slight abuse of notation, we also  denote by $\rightarrow_\rules^*$ the least relation containing $\rightarrow_\rules^+$ and $\iso$. 
\end{definition}

Note that $\rightarrow_\rules^*$ is the reflexive and transitive closure of $\rightarrow_\rules$, if the latter is viewed as a relation on equivalence classes of {\graph}s w.r.t.\ $\iso$.

\begin{proposition}
\label{prop:rw_iso}
If $\agraph \rightarrow_{\rules} \agraph'$
and $\agraph'' \iso \agraph$, then 
$\agraph''\rightarrow_{\rules} \agraph'$. 
\end{proposition}
\begin{proof}
By definition, there exists $\agraphB \rightarrow \agraphB' \in \rules$, a \substitution $\sigma$ and an \nsubstitution $\amap$ of domain $\nodes{\agraphB} \cup \nodes{\agraphB'}$ such that
$\sigma(\amap(\agraphB)) \sgrc \agraph$,
$\sigma(\amap(\agraphB'))$ is \substitutable for $\sigma(\amap(\agraphB))$ in $\agraph$, 
 and $\agraph' \iso \replace{\agraph}{\sigma(\amap(\agraphB))}{\sigma(\amap(\agraphB'))}$.
 Since $\agraph'' \iso \agraph$, there exists an \nsubstitution 
 $\amap'$ such that $\agraph'' = \amap'(\agraph)$.
 By Definition \ref{def:sgrc} (\ref{sgrc:iso}), 
 we have $\amap'(\sigma(\amap(\agraphB))) \sgrc \amap'(\agraph)$.
 By Proposition \ref{prop:replace_iso}, 
 $\amap'(\sigma(\amap(\agraphB')))$ is \substitutable for $\amap'(\sigma(\amap(\agraphB)))$
 in $\amap'(\agraph)$, and
 \begin{eqnarray*}
 	\amap'(\replace{\agraph}{\sigma(\amap(\agraphB))}{\sigma(\amap(\agraphB'))})
 &= & \replace{\amap'(\agraph)}{\amap'(\sigma(\amap(\agraphB)))}{\amap'(\sigma(\amap(\agraphB')))}\\
 & = & \replace{\agraph''}{\amap'(\sigma(\amap(\agraphB)))}{\amap'(\sigma(\amap(\agraphB')))}.
 \end{eqnarray*}
 Thus 
 $\agraph' \iso \replace{\agraph''}{\amap'(\sigma(\amap(\agraphB)))}{\amap'(\sigma(\amap(\agraphB')))}$.
Since $\amap'(\sigma(\amap(\agraphB))) = \sigma(\amap'\circ\amap(\agraphB))$ and $\amap'(\sigma(\amap(\agraphB'))) = \sigma(\amap'\circ\amap(\agraphB'))$, by Definition \ref{def:rw} we deduce that  
 $\agraph'' \rightarrow_{\rules} \agraph'$.

  \end{proof}
 
\begin{lemma}
\label{lem:rw_subgraph}
Let $\rules$ be a \system.
Let $\agraph$ be a \graph and  $\agraphC$ be a \subgraph of $\agraph$, with $\agraphC \sgrc \agraph$.
If $\agraphC \rightarrow_{\rules}^k \agraphC'$ for some $k \geq 0$ and $\agraphC'$ is \substitutable for $\agraphC$ in $\agraph$, 
then $\agraph \rightarrow_{\rules}^k \replace{\agraph}{\agraphC}{\agraphC'}$.
\end{lemma}
\begin{proof}
The proof is by induction on $k$. If $k = 0$, then the proof is an immediate consequence of Proposition \ref{prop:trivial_replacement}.
Assume that $k > 0$. By Definition \ref{def:rr}, 
there exist a rule $\agraphB \rightarrow \agraphB' \in \rules$, a \substitution $\sigma$ and an \nsubstitution $\amap$ of domain $\nodes{\agraphB} \cup \nodes{\agraphB'}$ such that
$\sigma(\amap(\agraphB)) \sgrc \agraphC$,
$\sigma(\amap(\agraphB'))$ is \substitutable for $\sigma(\amap(\agraphB))$ in $\agraphC$, 
 and $\agraphC'' \iso \replace{\agraphC}{\sigma(\amap(\agraphB))}{\sigma(\amap(\agraphB'))}$
 with $\agraphC'' \rightarrow_{\rules}^{k-1} \agraphC'$. 
 We assume w.l.o.g.\ that:
 \begin{itemize}
 	\item 
 	The \nsubstitution $\amap$ is such that $\nodes{\amap(\agraphB')} \cap \nodes{\agraph} \subseteq \nodes{\agraphC}$, 
 	so that the graph $\replace{\agraphC}{\sigma(\amap(\agraphB))}{\sigma(\amap(\agraphB'))}$ is \substitutable for $\agraphC$ in $\agraph$ (such an \nsubstitution is guaranteed to exists because 
 	$\nodes{\agraphB} \subseteq \nodes{\agraphC}$);
 	\item $\nodes{\agraphC''} \cap \nodes{\agraph} = \emptyset$, so that $\agraphC''$ is \substitutable for $\agraphC$ in $\agraph$.
 \end{itemize} 
By Corollary \ref{cor:replace_iso_subgraph} we have $\replace{\agraph}{\agraphC}{\agraphC''} \iso \replace{\agraph}{\agraphC}{\replace{\agraphC}{\sigma(\amap(\agraphB))}{\sigma(\amap(\agraphB'))}}$, and by Lemma \ref{lem:replace_subgraph},
\[ \replace{\agraph}{\sigma(\amap(\agraphB))}{\sigma(\amap(\agraphB'))} 
= \replace{\agraph}{\agraphC}{\replace{\agraphC}{\sigma(\amap(\agraphB))}{\sigma(\amap(\agraphB'))}}.   \]
We deduce that $\replace{\agraph}{\agraphC}{\agraphC''}  \iso \replace{\agraph}{\sigma(\amap(\agraphB))}{\sigma(\amap(\agraphB'))}$.
  
 By Definition \ref{def:sgrc} $\sgrc$ is transitive, and since 
 $\agraphC \sgrc \agraph$ and $\sigma(\amap(\agraphB)) \sgrc \agraphC$,
 we deduce that $\sigma(\amap(\agraphB)) \sgrc \agraph$, which proves that $\agraph  \rightarrow_{\rules} \replace{\agraph}{\agraphC}{\agraphC''}$.
%
By Definition \ref{def:sgrc} (\ref{sgrc:preserve}) we have $\agraphC'' \sgrc \agraph$ and by the induction hypothesis, 
  $\replace{\agraph}{\agraphC}{\agraphC''} \rightarrow_{\rules}^{k-1} \replace{\agraph}{\agraphC}{\agraphC'}$.
 Therefore, $\agraph \rightarrow_{\rules}^k \replace{\agraph}{\agraphC}{\agraphC'}$.

\end{proof}

\newcommand{\criticalpair}{critical pair\xspace}
\newcommand{\joinable}{joinable\xspace}


\begin{definition}
\label{def:critical_pair} 
A {\em \criticalpair} for a set of {\rewriterule}s $\rules$
is a pair $(\agraph_1,\agraph_2)$ such that there exist rules $\agraphB_i \rightarrow \agraphB'_i$ and  
{\nsubstitution}s
$\amap_i$  for $i = 1,2$, such that:
\begin{itemize}
	\item  $\amap_1(\agraphB_1)$ and $\amap_2(\agraphB_2)$ are \compatible with  mgu $\sigma$;
	\item $(\amap_1(\nodes{\agraphB_1'})\cup \amap_2(\nodes{\agraphB_2'})) \cap (\amap_1(\nodes{\agraphB_1}) \cup \amap_2(\nodes{\agraphB_2})) = \emptyset$;
	\item there exists a \merge $\agraph$  of $\amap_1(\agraphB_1)$	and $\amap_2(\agraphB_2)$ and for $i = 1,2$, $\amap_i(\agraphB_i) \sgrc \agraph$; 
	\item $\agraph_i = \replace{\agraph}{\sigma(\amap_i(\agraphB_i))}{\sigma(\amap_i(\agraphB'_i))}$.
\end{itemize}


The critical pair $(\agraph_1,\agraph_2)$ is {\em trivial} if
$\amap_1(\nodes{\agraphB_1}) \cap \amap_2(\nodes{\agraphB_2})  = \emptyset$.
\end{definition}

Note that if the critical pair $(\agraph_1,\agraph_2)$ is {trivial}, then this entails that $\sigma$ is the identity on the set of variables $\fv{\agraphB_1} \cup \fv{\agraphB_2}$, since the set $\set{\tuple{\labs{\amap_1(\agraphB_1)}(\anodeB),\labs{\amap_2(\agraphB_2)}(\anodeB)} \mid \anodeB \in \nodes{\amap_1(\agraphB_1)} \cap\nodes{\amap_2(\agraphB_2)}}$ is empty.

\begin{definition}
\label{def:joinable}
A \criticalpair $(\agraph_1,\agraph_2)$ is {\em \joinable} if 
there exists a graph $\agraph'$ such that $\agraph_i \rightarrow^*_\rules \agraph'$ for $i = 1,2$.
\end{definition}

\begin{proposition}
Any trivial critical pair is \joinable.
\end{proposition}
\begin{proof}
Let $(\agraph_1,\agraph_2)$ be a trivial critical pair.
We use the same notations as in Definition \ref{def:critical_pair}.
By Lemma \ref{lemma:replace_commute}, for every $i,j\in \{1,2\}$ with $i \not = j$, 
we have $\sigma(\amap_i(\agraphB_i)) \sgr \agraph_j$, 
 $\sigma(\amap_i(\agraphB_i'))$ is \substitutable for $\sigma(\amap_i(\agraphB_i))$ in $\agraph_j$ and:
 \begin{equation}
\replace{\agraph_2}{\sigma(\amap_1(\agraphB_1))}{\sigma(\amap_1(\agraphB_1'))} 
= 
\replace{\agraph_1}{\sigma(\amap_2(\agraphB_2))}{\sigma(\amap_2(\agraphB_2'))}.
\label{eq:commute}
\end{equation}
Furthermore, by Definition \ref{def:sgrc} (\ref{sgrc:replace}), 
we have $\sigma(\amap_i(\agraphB_i)) \sgrc \agraph_j$, thus
$\agraph_i \rightarrow_{\rules} \replace{\agraph_i}{\sigma(\amap_j(\agraphB_j))}{\sigma(\amap_j(\agraphB_j'))}$, for every $i,j = 1, 2$, $i \not = j$. By Equation \ref{eq:commute}, we obtain the  result.
\end{proof}

A \system is {\em terminating} if $\rightarrow^*_\rules$ is well-founded, 
and 
{\em confluent} (resp. locally confluent) if for all {\graph}s $\agraph,\agraph_1,\agraph_2$ such that
$\agraph \rightarrow_\rules^* \agraph_1$ 
and
$\agraph \rightarrow_\rules^* \agraph_2$ (resp. $\agraph \rightarrow_\rules \agraph_1$ 
and
$\agraph \rightarrow_\rules \agraph_2$), there exists $\agraph'$ such that
$\agraph_1 \rightarrow_\rules^* \agraph'$ 
and
$\agraph_2 \rightarrow^* \agraph'$.

We recall the following well-known result:
\begin{lemma}(Newman's lemma)
Any binary relation that is well-founded and locally confluent is confluent.
\end{lemma}
\begin{proof}
See for instance \cite{baader1998}.
\end{proof}

\begin{theorem}
\label{theo:confluence}
A \system is locally confluent iff
all its (nontrivial) {\criticalpair}s are \joinable.
\end{theorem}
\begin{proof}
Let $\rules$ be a \system and 
assume that all the {\criticalpair}s in $\rules$ are \joinable.
Let $\agraph, \agraph_1, \agraph_2$ be {\graph}s such that
$\agraph \rightarrow_{\rules} \agraph_i$ (for $i = 1,2$).
By definition there exist two rules $\agraphB_i \rightarrow \agraphB_i'$ in $\rules$, 
two {\substitution}s $\sigma_i$ and two {\nsubstitution}s  $\amap_i$ ($i = 1,2$) 
such that $\sigma_i(\amap_i(\agraphB_i)) \sgrc \agraph$, 
$\sigma_i(\amap_i(\agraphB_i'))$ is \substitutable for 
$\sigma_i(\amap_i(\agraphB_i))$
in $\agraph$  and
$\agraph_i \iso \replace{\agraph}{\sigma_i(\amap_i(\agraphB_i))}{\sigma_i(\amap_i(\agraphB_i'))}$.
Observe that $\sigma_1(\amap_1(\agraphB_1))$ and $\sigma_2(\amap_2(\agraphB_2))$ are necessarily \compatible.
 By Definition \ref{def:sgrc} (\ref{sgrc:merge}), there exists a \merge  $\agraph'$ of $\amap_1(\agraphB_1)$ and $\amap_2(\agraphB_2)$, with an mgu 
 $\theta$, such that $\theta'(\agraph') \sgrc \agraph$, with $\sigma_i = \theta\theta'$.
By Lemma \ref{lem:replace_subgraph}, by letting 
$\agraph_i' \isdef \replace{\agraph'}{\theta(\amap_i(\agraphB_i))}{\theta(\amap_i(\agraphB_i'))}$, we have $\agraph_i = \replace{\agraph}{\theta'(\agraph')}{\theta'(\agraph'_i)})$.

$(\agraph_1',\agraph_2')$ is a \criticalpair. 
Since all {\criticalpair}s are \joinable, there exists
$\agraph''$ such that 
$\agraph_i' \rightarrow_{\rules}^* \agraph''$, for all $i = 1,2$.
Without loss of generality we may assume that $\agraph''$ is disjoint from $\agraph_1'$ and $\agraph_2'$, hence \substitutable for $\agraph'_1$ and $\agraph_2'$ in $\agraph_i$.
By Lemma \ref{lem:rw_subgraph}, we deduce that
$\agraph_i \rightarrow^*_{\rules} \replace{\agraph_i}{\theta'(\agraph_i')}{\theta'(\agraph'')}$. 
By Lemma \ref{lem:replaceseq}, this entails that $\agraph_i \rightarrow^* \replace{\agraph}{\theta'(\agraph')}{\theta'(\agraph'')}$, and the proof is completed.  
\end{proof}

\newcommand{\gequation}{\graph equation\xspace}
\newcommand{\gdisequation}{\graph disequation\xspace}
\newcommand{\gliteral}{\graph literal\xspace}
\newcommand{\formula}{\graph formula\xspace}
\newcommand{\false}{\bot}
\newcommand{\iseq}{\approx}
\newcommand{\niseq}{\not \approx}
\newcommand{\alit}{{\cal L}}
\newcommand{\arel}{\bowtie}

\section{Graph Literals and Graph Congruences}\label{sec:lits}

In this section we introduce \emph{{\gliteral}s}, which are either equations or disequations between graphs, along with \emph{graph congruences}, which permit to define a semantics on a set of {\gliteral}s.

\begin{definition}	
A {\em \gequation} is an unordered pair written $\agraph \iseq \agraphB$, where 
$\agraph, \agraphB$ are \rootsimilar\ {\graph}s.
A {\em \gdisequation} is the negation of a \graph equation, written $\agraph \niseq \agraphB$.
A {\em \gliteral} is either a \gequation or a \gdisequation or $\false$.
The set of variables occurring in a \gliteral $\alit$ is defined as follows:
$\fv{\agraph \iseq \agraphB} = \fv{\agraph \niseq \agraphB} = \fv{\agraph} \cup \fv{\agraphB}$
and $\fv{\false} = \emptyset$.
All {\substitution}s
 and {\nsubstitution}s $\amap$ are extended to {\gliteral}s as follows:
 $\amap(\agraph \arel \agraphB) = \amap(\agraph) \arel \amap(\agraphB)$, for $\arel \in \{ \iseq, \niseq \}$,
 and $\amap(\false) = \false$. 
 
 If $\alit, \alit'$ are {\gliteral}s, we write $\alit \iso \alit'$ iff either $\alit = \alit'$ or 
 $\alit = (\agraph \arel \agraphB)$,
 $\alit' = (\agraph' \arel \agraphB')$, $\agraph \iso \agraph'$ and $\agraphB \iso \agraphB'$.
\end{definition}
 
A {\em \formula} is a set of {\gliteral}s.

\newcommand{\IF}[1]{I_g(#1)}

\begin{definition}
\label{def:IF}
For any \formula $S$, we denote by $\IF{S}$ the set of 
ground instances of a \gliteral in $S$, i.e., the set of 
{\gliteral}s $\sigma(\alit)$ where $\alit \iso \alit'$ for some $\alit' \in S$ and $\sigma$ is a ground \substitution 
of domain $\fv{\alit}$.
\end{definition}

We write $\alit \rightarrow_{\rules} \alit'$
if $\alit = \agraph \arel \agraphB$, 
$\alit' = \agraph' \arel \agraphB$
and $\agraph \rightarrow_{\rules} \agraph'$.



\newcommand{\gcongruence}{\graph congruence\xspace}

\newcommand{\rclosed}{closed under  embeddings\xspace}
\newcommand{\nclosed}{closed under isomorphisms\xspace}
\newcommand{\nrclosed}{closed under isomorphisms and embeddings\xspace}

\begin{definition}
\label{def:nclosed}
 A binary relation $\arel$ on {\graph}s is {\em \nclosed} if $\agraph \arel \agraph' \wedge \agraphB \iso \agraph \wedge \agraphB' \iso \agraph' \Rightarrow \agraphB \arel \agraphB'$.  
\end{definition}

\begin{definition}
\label{def:rclosed}
A binary relation $\arel$ on {\graph}s is {\em \rclosed} 
if 
for every {\graph}s $\agraph,\agraphB,\agraphB'$ such that
$\agraphB \sgrc \agraph$, 
 $\agraphB'$ is \substitutable for $\agraphB$ in 
$\agraph$ and $\agraphB \arel \agraphB'$ then $\agraph  \arel\replace{\agraph}{\agraphB}{\agraphB'}$
\end{definition}

 
\begin{definition}
\label{def:cong}
A {\em \gcongruence} $\acong$ is an equivalence relation 
between ground {\graph}s that is \nrclosed.
\end{definition}

\begin{definition}
\label{def:eval}
A \gcongruence $\acong$ {\em validates} $\phi$ if:
\begin{itemize}
\item{$\phi$ is a ground \gequation $\agraph \iseq \agraphB$ and 
$\agraph \acong \agraphB$;}
\item{$\phi$ is a ground \gdisequation $\agraph \niseq \agraphB$ and 
$\agraph \not \acong \agraphB$;}
\item{$\phi$ is a formula and $\acong$ validates all {\gliteral}s in $\IF{\phi}$.}
\end{itemize}
If a \gcongruence validates $\phi$ then it is a \emph{model} of $\phi$, and we say that $\phi$ is {\em satisfiable}. 
We write $\phi \models \psi$ ($\psi$ is a {\em logical consequence of $\phi$}) if every model of $\phi$ is a model of $\phi'$, and $\phi \equiv \psi$ 
if $\phi \models \psi$
and $\psi \models \phi$.
\end{definition}

\begin{proposition}
For any set of {\gequation}s $E$, there exists a \gcongruence $\acong_E$ that is the minimal (w.r.t.\ $\subseteq$) model of
$E$.
\end{proposition}
\begin{proof}
It is easy to check that if ${\cal M}$ is a set of models of $E$, then the relation defined as the  intersection of all the relations in ${\cal M}$ is also a model of $E$. Further, $E$ has at least one model, containing all pairs of \rootsimilar ground {\graph}s. Hence $\acong_E$ can be defined as the intersection of all the models of $E$.  
\end{proof}

We first prove that the satisfiability problem
is undecidable for ground {\formula}s. The result is straightforward, but it is given here to emphasize the 
difference with the usual ground equational logic (on terms) 
which is well-known to be  decidable

\begin{proposition}
The satisfiability problem is undecidable for ground {\formula}s.
\end{proposition}

\begin{proof}
The proof goes by a straightforward reduction from the halting problem for Turing machines (TM).
Let $M$ be a deterministic TM $(Q,\Gamma,b,\Sigma,q_0,F,\delta)$, where $Q$ is the set of states, $q_0$ is the initial state, $F \subseteq Q$ is the set of final states, $b$ is the blank symbol, $\Gamma$ is the alphabet, $\Sigma \subseteq \Gamma$ is the set of input symbols and $\delta: (Q \setminus F) \times \Gamma \rightarrow Q \times \Gamma \times \{ L,R \}$ the transition function.
We assume, w.l.o.g., that $Q \cap \Gamma = \emptyset$.
The configurations are  tuples $c = (q,w,w')$, where $q$ is a state and $w,w' \in \Gamma^+$ denotes the part of the tape before and after the head, respectively (we assume that $w,w'$ are non empty for technical convenience).
Any such configuration, with $w = w_1.\dots.w_n$, $w' = w_1'.\dots.w_m'$, with $w_i,w_j' \in \Gamma$, may be encoded into a graph $\agraph(c)$ defined as follows (where $s,h,e$ are pairwise distinct symbols not occurring in $\Gamma \cup Q$):
\[
\begin{tabular}{lll}
$\nodes{\agraph(c)}$ & $=$ & $\{ \anode_0, \anode_1, \anode_2, \anode_3, \anodeB_1,\dots,\anodeB_n,  \anodeB_1',\dots,\anodeB_m' \}$ \\
 $\inputs{\agraph(c)}$ & $=$ & $\tuple{}$ \\ 
$\edges{\agraph{c}}$ & $=$ &$\{ \mkedge{\anodeB_i}{\anodeB_{i+1}} \mid i = 1,\dots,n-1 \}$ \\
& &  $\cup\;
\{ \mkedge{\anodeB_j'}{\anodeB_{j+1}'} \mid j = 1,\dots,m-1 \}$ \\
& & $ \cup\;
 \{ \mkedge{\anode_1}{\anodeB_1}, \mkedge{\anodeB_n}{\anode_2},  \mkedge{\anode_2}{\anodeB_1'}, \mkedge{\anodeB_m'}{\anode_3} \}$ \\
$\labs{\agraph(c)}(\anodeB_i)$ & $ = $ & $w_i$ \text{\ (for $i =1,\dots,n$)} \\
$\labs{\agraph(c)}(\anodeB_j')$ & $=$&  $w_j'$ \text{\ (for $j =1,\dots,m$)} \\
\multicolumn{3}{l}{$\labs{\agraph(c)}(\anode_0) = q$ \quad
$\labs{\agraph(c)}(\anode_1) = s$ \quad
$\labs{\agraph(c)}(\anode_2) = h$ \quad
$\labs{\agraph(c)}(\anode_3) = e$} 
\end{tabular}
\]
The transition function is encoded by the following {\gequation}s:

\[
\begin{tabular}{ccc}
\begin{minipage}{0.30\linewidth}
\begin {tikzpicture}[-latex ,auto ,node distance =0.75 cm and 0.75cm ,on grid ,
semithick]
\node[inputNode] (Q) {$q_1$};
\node[] (I) [right =of Q] {};
\node[inputNode] (A) [right =of I] {$h$};
\node[inputNode] (B) [right =of A] {$i$};
\node[inputNode] (C) [right =of B] {$k$};
\node[] (O) [right =of C] {};
\draw (I) [->] to (A);
\draw (A) [->] to (B);
\draw (B) [->] to (C);
\draw (C) [->] to (O);
\end{tikzpicture}
\end{minipage}& \qquad $=$ \qquad &
\begin{minipage}{0.30\linewidth}
\begin {tikzpicture}[-latex ,auto ,node distance =0.75 cm and 0.75cm ,on grid ,
semithick]
\node[inputNode] (Q) {$q_2$};
\node[] (I) [right =of Q] {};
\node[inputNode] (B) [right =of I] {$j$};
\node[inputNode] (A) [right =of B] {$h$};
\node[inputNode] (C) [right =of A] {$k$};
\node[] (O) [right =of C] {};
\draw (I) [->] to (B);
\draw (B) [->] to (A);
\draw (A) [->] to (C);
\draw (C) [->] to (O);
\end{tikzpicture}
\end{minipage} \\
\begin{minipage}{0.30\linewidth}
\begin {tikzpicture}[-latex ,auto ,node distance =0.75 cm and 0.75cm ,on grid ,
semithick]
\node[inputNode] (Q) {$q_1$};
\node[] (I) [right =of Q] {};
\node[inputNode] (A) [right =of I] {$h$};
\node[inputNode] (B) [right =of A] {$i$};
\node[inputNode] (C) [right =of B] {$e$};
\draw (I) [->] to (A);
\draw (A) [->] to (B);
\draw (B) [->] to (C);
\end{tikzpicture}
\end{minipage}& \qquad $=$ \qquad &
\begin{minipage}{0.30\linewidth}
\begin {tikzpicture}[-latex ,auto ,node distance =0.75 cm and 0.75cm ,on grid ,
semithick]
\node[inputNode] (Q) {$q_2$};
\node[] (I) [right =of Q] {};
\node[inputNode] (B) [right =of I] {$j$};
\node[inputNode] (A) [right =of B] {$h$};
\node[inputNode] (C) [right =of A] {$b$};
\node[inputNode] (D) [right =of C] {$e$};
\draw (I) [->] to (B);
\draw (B) [->] to (A);
\draw (A) [->] to (C);
\draw (C) [->] to (D);
\end{tikzpicture}
\end{minipage} \\
\end{tabular}
\]
for all $q_1,i,q_2,j,k$ such that $\delta(q_1,i) = (q_2,j,R)$ and $k \in \Gamma$.
 The second rule encodes the fact that the tape is infinite: if the head is at the end of the word, a new blank symbol must be created to ensure that there are always symbols after the head (i.e., $w'$ must be nonempty in any configuration $(q,w,w')$).
Similar rules are defined for the left movement (for all $q_1,i,q_2,j,k$ such that $\delta(q_1,i) = (q_2,j,L)$ and $k,l \in \Gamma$):
\[
\begin{tabular}{ccc}
\begin{minipage}{0.30\linewidth}
\begin {tikzpicture}[-latex ,auto ,node distance =0.75 cm and 0.75cm ,on grid ,
semithick]
\node[inputNode] (Q) {$q_1$};
\node[] (I) [right =of Q] {};
\node[inputNode] (J) [right =of I] {$l$};
\node[inputNode] (A) [right =of J] {$k$};
\node[inputNode] (B) [right =of A] {$h$};
\node[inputNode] (C) [right =of B] {$i$};
\node[] (O) [right =of C] {};
\draw (I) [->] to (J);
\draw (J) [->] to (A);
\draw (A) [->] to (B);
\draw (B) [->] to (C);
\draw (C) [->] to (O);
\end{tikzpicture}
\end{minipage}& \qquad $=$ \qquad &
\begin{minipage}{0.30\linewidth}
\begin {tikzpicture}[-latex ,auto ,node distance =0.75 cm and 0.75cm ,on grid ,
semithick]
\node[inputNode] (Q) {$q_2$};
\node[] (I) [right =of Q] {};
\node[inputNode] (J) [right =of I] {$l$};
\node[inputNode] (B) [right =of J] {$h$};
\node[inputNode] (A) [right =of B] {$k$};
\node[inputNode] (C) [right =of A] {$j$};
\node[] (O) [right =of C] {};
\draw (I) [->] to (J);
\draw (J) [->] to (B);
\draw (B) [->] to (A);
\draw (A) [->] to (C);
\draw (C) [->] to (O);
\end{tikzpicture}
\end{minipage} \\
\begin{minipage}{0.30\linewidth}
\begin {tikzpicture}[-latex ,auto ,node distance =0.75 cm and 0.75cm ,on grid ,
semithick]
\node[inputNode] (Q) {$q_1$};
\node[inputNode] (A) [right =of Q] {$s$};
\node[inputNode] (J) [right =of A] {$b$};
\node[inputNode] (B) [right =of J] {$h$};
\node[inputNode] (C) [right =of B] {$i$};
\node[] (I) [right =of C] {};
\draw (A) [->] to (J);
\draw (J) [->] to (B);
\draw (B) [->] to (C);
\draw (C) [->] to (I);
\end{tikzpicture}
\end{minipage}& \qquad $=$ \qquad &
\begin{minipage}{0.30\linewidth}
\begin {tikzpicture}[-latex ,auto ,node distance =0.75 cm and 0.75cm ,on grid ,
semithick]
\node[inputNode] (Q) {$q_2$};
\node[inputNode] (B) [right =of Q] {$s$};
\node[inputNode] (J) [right =of B] {$b$};
\node[inputNode] (A) [right =of J] {$h$};
\node[inputNode] (C) [right =of A] {$b$};
\node[inputNode] (D) [right =of C] {$j$};
\node[] (I) [right =of D] {};
\draw (B) [->] to (J);
\draw (J) [->] to (A);
\draw (A) [->] to (C);
\draw (C) [->] to (D);
\draw (D) [->] to (I);
\end{tikzpicture}
\end{minipage} \\
\end{tabular}
\]
We also need rules to delete useless blank symbols at the end or at the beginning of the tape:
\[
\begin{tabular}{ccc}
\begin{minipage}{0.30\linewidth}
\begin {tikzpicture}[-latex ,auto ,node distance =0.75 cm and 0.75cm ,on grid ,
semithick]
\node[inputNode] (A)  {$s$};
\node[inputNode] (B) [right =of A] {$b$};
\node[inputNode] (C) [right =of B] {$b$};
\node[] (O) [right =of C] {};
\draw (A) [->] to (B);
\draw (B) [->] to (C);
\draw (C) [->] to (O);
\end{tikzpicture}
\end{minipage}& \qquad $=$ \qquad &
\begin{minipage}{0.30\linewidth}
\begin {tikzpicture}[-latex ,auto ,node distance =0.75 cm and 0.75cm ,on grid ,
semithick]
\node[inputNode] (A)  {$s$};
\node[inputNode] (C) [right =of A] {$b$};
\node[] (O) [right =of C] {};
\draw (A) [->] to (C);
\draw (C) [->] to (O);
\end{tikzpicture}
\end{minipage} \\
\begin{minipage}{0.30\linewidth}
\begin {tikzpicture}[-latex ,auto ,node distance =0.75 cm and 0.75cm ,on grid ,
semithick]
\node[] (I)  {};
\node[inputNode] (A) [right =of I] {$b$};
\node[inputNode] (B) [right =of A] {$b$};
\node[inputNode] (C) [right =of B] {$e$};
\draw (I) [->] to (A);
\draw (A) [->] to (B);
\draw (B) [->] to (C);
\end{tikzpicture}
\end{minipage}& \qquad $=$ \qquad &
\begin{minipage}{0.30\linewidth}
\begin {tikzpicture}[-latex ,auto ,node distance =0.75 cm and 0.75cm ,on grid ,
semithick]
\node[] (I)  {};
\node[inputNode] (A) [right =of I] {$b$};
\node[inputNode] (C) [right =of A] {$e$};
\draw (I) [->] to (A);
\draw (A) [->] to (C);
\end{tikzpicture}
\end{minipage} \\
\end{tabular}
\]
Let $E$ be the set of such {\gequation}s. 
It is easy to check that, for all configurations $c,c'$, 
$\agraph(c) \acong_E \agraph(c')$ if either $c$ is reachable from $c'$
or  $c'$ is reachable from $c$.
We assume, w.l.o.g., that the considered TM only ends in a configuration $(q_f,b,b)$ (i.e., the TM clears the tape and ends in a single fixed final state $q_f$).
Then, the TM terminates on the empty word iff $E$, enriched by the following {\gdisequation}, is unsatisfiable:

\[
\begin{tabular}{ccc}
\begin{minipage}{0.4\linewidth}
\begin {tikzpicture}[-latex ,auto ,node distance =0.75 cm and 0.75cm ,on grid ,
semithick]
\node[inputNode] (Q) {$q_0$};
\node[inputNode] (S) [right =of Q] {$s$};
\node[inputNode] (B1) [right =of S] {$b$};
\node[inputNode] (H) [right =of B1] {$h$};
\node[inputNode] (B2) [right =of H] {$b$};
\node[inputNode] (E) [right =of B2] {$e$};
\draw (S) [->] to (B1);
\draw (B1) [->] to (H);
\draw (H) [->] to (B2);
\draw (B2) [->] to (E);
\end{tikzpicture}
\end{minipage}& $\not =$ &
\begin{minipage}{0.4\linewidth}
\begin {tikzpicture}[-latex ,auto ,node distance =0.75 cm and 0.75cm ,on grid ,
semithick]
\node[inputNode] (Q) {$q_f$};
\node[inputNode] (S) [right =of Q] {$s$};
\node[inputNode] (B1) [right =of S] {$b$};
\node[inputNode] (H) [right =of B1] {$h$};
\node[inputNode] (B2) [right =of H] {$b$};
\node[inputNode] (E) [right =of B2] {$e$};
\draw (S) [->] to (B1);
\draw (B1) [->] to (H);
\draw (H) [->] to (B2);
\draw (B2) [->] to (E);
\end{tikzpicture}
\end{minipage} \\
\end{tabular}
\]

\end{proof}

\section{\capitalisewords{a proof procedure}}\label{sec:pf-proc}

\subsection{Inference rules and redundancy}

\newcommand{\PS}{${\mathtt S}^+$}
\newcommand{\NS}{${\mathtt S}^-$}
\newcommand{\REFL}{${\mathtt R}$}

\newcommand{\psup}{Positive Unit Superposition\xspace}
\newcommand{\nsup}{Negative Unit Superposition\xspace}
\newcommand{\reflect}{Reflection\xspace}
\newcommand{\sclosed}{closed under {\substitution}s}
\newcommand{\srclosed}{closed under {\substitution}s and embeddings}
\newcommand{\allclosed}{closed under embeddings, {\substitution}s and isomorphisms}
\newcommand{\glt}{<}
\newcommand{\gleq}{\leq}
\newcommand{\ggt}{>}
\newcommand{\ggeq}{\geq}
\newcommand{\gequiv}{\simeq}

\newcommand{\gorder}{graph reduction order}

We define a set of inference rules that can be viewed as an adaptation of the Superposition calculus to graph formulas. Similarly to the Superposition calculus, this set of inference rules is designed to saturate a set of {\gliteral}s in order to derive $\bot$ when the set is unsatisfiable.

\begin{definition}
\label{def:sclosed}
A binary relation $\arel$ on {\graph}s is {\em \sclosed} 
if 
for all {\graph}s $\agraph,\agraphB$ such that
$\agraph \arel \agraphB$, and for every \substitution $\sigma$, we have
$\sigma(\agraph) \arel \sigma(\agraphB)$. 
\end{definition}

\begin{definition}
\label{def:gorder}
A {\em \gorder} $\ggeq$ is a preorder on {\graph}s that satisfies the following properties:
\begin{itemize}
\item{The associated order $\ggt$, defined as $\agraph \ggt \agraphB \iff (\agraph \ggeq \agraphB \wedge \agraphB \not \ggeq \agraph)$, is well-founded.}
\item{$\ggeq$ and $\ggt$ are \allclosed.}
\item{For all ground {\graph}s $\agraph,\agraphB$, we have either 
$\agraph \ggeq \agraphB$
or
$\agraphB \ggeq \agraph$.}

\end{itemize}
We write $\agraph \gequiv \agraphB$ 
for $\agraph \ggeq \agraphB \wedge
\agraph \gleq \agraphB$.
\end{definition}

{
\newcommand{\card}[1]{|#1|}

\begin{example}
For instance the relation $\agraph \ggeq \agraphB \iff \card{\nodes{\agraph}} \geq \card{\nodes{\agraphB}}$ is a \gorder.
\end{example}
}

Let $\gleq$ be a \gorder.
We consider the following $3$ rules, which apply modulo a renaming of {\node}s and variables.

\paragraph*{\psup (\PS)}

\[
\begin{tabular}{c}
$\agraph \iseq \agraph'$ \qquad $\agraphB \iseq \agraphB'$ \\
\hline
$\replace{\agraphC}{\sigma(\agraph)}{\sigma(\agraph')} \iseq 
\replace{\agraphC}{\sigma(\agraphB)}{\sigma(\agraphB')}$
\end{tabular}
\]
\noindent
If $\fv{\agraph\iseq \agraph'} \cap \fv{\agraphB\iseq \agraphB'} = \emptyset$,
$\nodes{\agraph} \cap \nodes{\agraphB} \not = \emptyset$,
$\agraphC$ is a \merge of $\agraph$ and $\agraphB$, with an mgu $\sigma$,
$\agraph'$ and $\agraphB'$ are respectively \substitutable for $\agraph$ and $\agraphB$ in $\agraphC$,
$\replace{\agraphC}{\sigma(\agraph)}{\sigma(\agraph')} \not \ggt 
\agraphC$
and
$\replace{\agraphC}{\sigma(\agraphB)}{\sigma(\agraphB')} \not \ggt
\agraphC$.

\paragraph*{\nsup (\NS)}
\[
\begin{tabular}{c}
$\agraph \not \iseq \agraph'$ \qquad $\agraphB \iseq \agraphB'$ \\
\hline
$\replace{\sigma(\agraph)}{\sigma(\agraphB)}{\sigma(\agraphB')}\not  \iseq 
\sigma(\agraph')$
\end{tabular}
\]
\noindent
If $\fv{\agraph\not\iseq \agraph'} \cap \fv{\agraphB\iseq \agraphB'} = \emptyset$,
$\sigma$ is a most general \substitution such that $\sigma(\agraphB) \sgrc \sigma(\agraph)$,
$\agraphB'$ is \substitutable for $\agraphB$ in $\agraph$, 
$\sigma(\agraph') \not \ggt 
\sigma(\agraph)$
and
$\sigma(\agraphB') \not \ggt \sigma(\agraphB)$.

\paragraph*{\reflect (\REFL)}

\newcommand{\unifiable}{unifiable\xspace}

Two {\graph}s $\agraph$ and $\agraphB$ are {\em \unifiable} iff there exists a \substitution $\sigma$
such that $\sigma(\agraph) = \sigma(\agraphB)$.

\[
\begin{tabular}{c}
$\agraph \not \iseq \agraphB$ \\
\hline
$\false$
\end{tabular}
\quad \text{If $\agraph$ and $\agraphB$ are \unifiable.}
\]

\newcommand{\OneStepSup}[1]{\mathtt{Sup}(#1)}
\newcommand{\NStepSup}[2]{\mathtt{Sup}^{#2}(#1)}
\newcommand{\AnyStepSup}[1]{\mathtt{Sup}^{*}(#1)}

For every \formula $S$, we denote by 
$\OneStepSup{S}$ the set of the {\gliteral}s deducible from $S$ by one of the above rules (in one step).
Let $\AnyStepSup{S} \isdef \bigcup_{i=0}^{\infty} \NStepSup{S}{i}$, with:
\[\NStepSup{S}{i} \isdef \left\{ \begin{array}{l} 
S \text{\ if $i = 0$} \\
S \cup \OneStepSup{S} \text{\ if $i = 1$} \\
\NStepSup{\NStepSup{S}{i-1}}{1} \text{\ otherwise} 
\end{array}\right.\] 

\begin{theorem}[Soundness]
The rules \PS, \NS and {\REFL} are sound, i.e., 
for every \formula $S$, $S \models \OneStepSup{S}$. Thus if $\false \in \AnyStepSup{S}$ then $S$ is unsatisfiable.
\end{theorem}

\newcommand{\Rw}[1]{\rules_{#1}}

\newcommand{\sRw}[1]{\rules_{#1}^{\ggt}}
\newcommand{\redundant}{redundant\xspace}
\newcommand{\saturated}{saturated\xspace}

We adapt the notions of redundancy and saturation to our setting:

\begin{definition}	
For every \formula $S$ we denote by 
$\sRw{S}$ (resp.\ $\Rw{S}$) the set of rules 
$\agraph \rightarrow \agraphB$
where $\agraph \iseq \agraphB \in \IF{S}$
and $\agraph \ggt \agraphB$ (resp.\ $\agraph \ggeq \agraphB$).
\end{definition}

\begin{proposition}\label{prop:srw-ggeq}
	If $\agraph \rightarrow_{\Rw{S}}^* \agraphB$ then $\agraph \ggeq \agraphB$, and if $\agraph \rightarrow_{\sRw{S}}^+ \agraphB$ 
	then $\agraph \ggt \agraphB$.
\end{proposition}

\begin{definition}
A ground \gliteral $\alit$ is {\em \redundant}
w.r.t.\ a \formula $S$
if one of the following conditions hold.
\begin{enumerate}
\item{$\false \in S$}
\item{$\alit$ is of the form $\agraph \iseq \agraphB$ where 
$\agraph \iso \agraphB$.}
\item{There exist a \gliteral $\alit' \in S$ and a \substitution 
$\sigma$ such that $\sigma(\alit') \iso \alit$.}
\item{There exists a \gliteral $\alit'$ such that 
$\alit \rightarrow_{\sRw{S}} \alit'$ and $\alit'$ is \redundant w.r.t.\ $S$.}
\end{enumerate}
\end{definition}

\begin{definition}
A \formula $S$ is {\em \saturated} if all {\gliteral}s in $\IF{\OneStepSup{S}}$
are \redundant w.r.t.\ $S$.
\end{definition}

\section{Completeness}\label{sec:complete}

\newcommand{\rr}[1]{\rightarrow^{#1}_{|\agraph}}
In this section we prove the completeness of our calculus, in other words, that if a saturated graph formula is unsatisfiable, then it necessarily contains $\bot$.
Throughout this section, we assume that $S$ denotes a fixed  \formula.

\begin{definition}
Let $\agraph$ be a \graph.
We write $\agraphB \rr{n} \agraphB'$ iff
there exists a sequence $\agraphC_i$ ($i = 1,\dots,m$)
with 
$\agraphC_1 = \agraphB$,
$\agraphC_{m} = \agraphB'$,
$\agraphC_i \rightarrow_{\Rw{S}} \agraphC_{i+1}$ for all $i =1,\dots,m-1$
and there are at most $n$ indices $i$ in $\interv{1}{m-1}$
such that $\agraphC_{i+1} \not \glt \agraph$.
\end{definition}

\begin{proposition}
\label{prop:rr_basic} We have the following properties:
\begin{enumerate}
\item{
If $\agraphB \rightarrow_{\Rw{S}}^* \agraphC$ then
$\agraphB \rr{n} \agraphC$, for some $n \in {\Bbb N}$. \label{rr:n}}
\item 
If $\agraphB \rr{n} \agraphC$ then $\agraphB \rr{m} \agraphC$ for all $m\geq n$.\label{rr:mgn}
\item{
If  $\agraphB \glt \agraph$ and $\agraphB \rightarrow_{\Rw{S}}^* \agraphC$
then
$\agraphB \rr{0} \agraphC$. \label{rr:0_lt}}
\item{
If $\agraphB \gleq \agraph$ and $\agraphB \rightarrow_{\sRw{S}}^* \agraphC$
then
$\agraphB \rr{0} \agraphC$. \label{rr:0_sRw} }
\item 
If $\agraphB \rr{n} \agraphB'$ and $\agraphB' \rr{m} \agraphC$ then $\agraphB \rr{n+m} \agraphC$. \label{rr:concat}
\end{enumerate}
\end{proposition}
\begin{proof}
The results follow immediately from the definition of $\rr{n}$ and Proposition \ref{prop:srw-ggeq}. 
\end{proof}

\newcommand{\amapB}{\nu}
\newcommand{\agraphD}{J}

\begin{lemma}
\label{lem:loc_confluence}
Assume that $S$ is \saturated and $\false \not \in S$.
Let $\agraph, \agraph_1,\agraph_2$ be 
ground {\graph}s such that
$\agraph \rightarrow_{\Rw{S}} \agraph_i$ (for $i = 1,2$).
There exists a \graph $\agraph'$ and numbers $\eta_i \in \{ 0,1\}$ such
that 
$\agraph_i \rr{\eta_i} \agraph'$, for all $i = 1,2$.
Furthermore, 
if $\agraph \ggt \agraph_i$ for some $i = 1,2$ 
then $\eta_1 = \eta_2 = 0$.
\end{lemma}
\begin{proof}
By definition of $\rightarrow_{\Rw{S}}$ there exist rules 
$\agraphB_i \rightarrow \agraphB_i'$ in 
$\Rw{S}$, {\substitution}s $\sigma_i$ and {\nsubstitution}s $\amap_i$ such that, for $i = 1,2$,
$\sigma_i(\amap_i(\agraphB_i)) \sgrc \agraph$,
$\sigma_i(\amap_i(\agraphB_i'))$ is \substitutable for 
$\sigma_i(\amap_i(\agraphB_i))$ in $\agraph$, and
$\agraph_i \iso \replace{\agraph}{\sigma(\amap_i(\agraphB_i))}{\sigma(\amap_i(\agraphB_i'))}$.

Let $i\in \set{1,2}$. Since $\Rw{S}$ is ground
we have $\sigma_i = \id$, and since $\IF{S}$ (hence $\Rw{S}$) is closed under isomorphisms, 
we may assume that $\amap_i = \id$, so that
$\agraph_i = \replace{\agraph}{\agraphB_i}{\agraphB_i'}$.
Furthermore, by definition of $\Rw{S}$, we have $\agraphB_i \ggeq \agraphB_i'$ and since $\ggeq$ is \rclosed, $\agraph \ggeq \agraph_i$.

We distinguish two cases.
\begin{itemize}
\item{Assume that $\nodes{\agraphB_1} \cap \nodes{\agraphB_2} = \emptyset$.
By Lemmas \ref{lem:replace_subs} and \ref{lemma:replace_commute}, for every $i,j\in \{1,2\}$ with $i \not = j$, 
we have $\agraphB_i \sgr \agraph_j$, 
 $\agraphB_i'$ is \substitutable for $\agraphB_i$ in $\agraph_j$ and:
 \begin{equation}
\replace{\agraph_2}{\agraphB_1}{\agraphB_1'} 
= 
\replace{\agraph_1}{\agraphB_2}{\agraphB_2'}.
\label{eq:commute_bis}
\end{equation}
By Definition \ref{def:sgrc} (\ref{sgrc:replace}), 
we have $\agraphB_i \sgrc \agraph_j$, thus
$\agraph_i \rightarrow_{\rules} \replace{\agraph_i}{\agraphB_j}{\agraphB_j'}$,  so that
$\agraph_i \rr{\eta_i} \replace{\agraph_i}{\agraphB_j}{\agraphB_j'}$, where: $\eta_i = 
\left\{\begin{array}{l}
1 \text{\ if $\replace{\agraph_i}{\agraphB_j}{\agraphB_j'} \gequiv \agraph$} \\
0 \text{\ otherwise} 
\end{array}\right.$

Note that by  Equation \ref{eq:commute_bis} we necessarily have $\eta_1 = \eta_2$.
Since $\agraphB_j \ggeq \agraphB_j'$ and $\ggeq$ is \rclosed,
we have 
$\agraph_i \ggeq \replace{\agraph_i}{\agraphB_j}{\agraphB_j'}$. 
If $\agraph \ggt \agraph_i$, for some $i = 1,2$, then 
$\agraph \ggt \replace{\agraph_i}{\agraphB_j}{\agraphB_j'}$, thus by definition $\eta_i = 0$.
By Equation \ref{eq:commute_bis}, we obtain the stated result.
 }

\item{Assume that $\nodes{\agraphB_1} \cap \nodes{\agraphB_2} \not = \emptyset$.
By definition of $\Rw{S}$, there exist
{\gliteral}s 
$\agraphC_i \iseq \agraphC_i'$ in $S$, ground
{\substitution}s $\theta_i$
and {\nsubstitution}s $\amapB_i$ (for $i = 1,2$)
such that
$\agraphB_i = \theta_i(\amapB_i(\agraphC_i))$ and
$\agraphB_i' = \theta_i(\amapB_i(\agraphC_i'))$.
We may assume that $\agraphC_1 \iseq \agraphC_1'$
and $\agraphC_2 \iseq \agraphC_2'$ are variable-disjoint, so that $\theta_1$ and $\theta_2$ have disjoint domains. 
By Definition \ref{def:sgrc} (\ref{sgrc:merge}), $\amapB_1(\agraphC_1)$ 
and $\amapB_2(\agraphC_2)$
admit a \merge $\agraphC$, with 
mgu $\theta$, such that $(\theta_1 \cup \theta_2) = \theta' \circ \theta$, for some \substitution $\theta'$,
and $\theta'(\agraphC) \sgrc \agraph$.
Since $\IF{S}$ is closed by node renaming, we may also assume, w.l.o.g., 
that $\amapB_i(\agraphC_i') \cap \nodes{\agraph} \subseteq \nodes{\amapB_i(\agraphC_i)}$, so that $\amapB_i(\agraphC_i')$ is \substitutable for $\amapB_i(\agraphC_i)$ in $\agraphC$
and also that 
$\replace{\agraphC}{\amapB_i(\agraphC_i)}{\amapB_i(\agraphC_i')}$ is \substitutable
for $\agraphC$ in $\agraph$.

Let $i\in \set{1,2}$ and assume that $\replace{\agraphC}{\theta(\amapB_i(\agraphC_i))}{\theta(\amapB_i(\agraphC_i'))} \ggt 
\agraphC$. Then, since the order $\ggt$ is \sclosed, 
we have
$\theta'(\replace{\agraphC}{\theta(\amapB_i(\agraphC_i))}{\theta(\amapB_i(\agraphC_i'))} \ggt 
\theta'(\agraphC)$, hence
$\replace{\theta'(\agraphC)}{\theta_i(\amapB_i(\agraphC_i))}{\theta_i(\amapB_i(\agraphC_i'))} \ggt 
\theta'(\agraphC)$. 
We get
\[\replace{\agraph}{\theta'(\agraphC)}{\replace{\theta'(\agraphC)}{\theta_i(\amapB_i(\agraphC_i))}{\theta_i(\amapB_i(\agraphC_i'))}} \ggt \agraph\] because $\ggeq$ is \rclosed,
thus, by Lemma \ref{lem:replace_subgraph},
\[\replace{\agraph}{\theta_i(\amapB_i(\agraphC_i))}{\theta_i(\amapB_i(\agraphC_i'))} \ggt \agraph.\]
Therefore, 
$\replace{\agraph}{\agraphB_i}{\agraphB_i'} \ggt \agraph$, which contradicts
the fact that $\agraphB_i \ggeq \agraphB_i'$.

Consequently, we have
$\replace{\agraphC}{\theta(\amapB_i(\agraphC_i))}{\theta(\amapB_i(\agraphC_i'))} \not \ggt 
\agraphC$ (for $i = 1,2$). This entails that the rule \PS\ is applicable 
on $\agraphC_1 \iseq \agraphC_1'$ and $\agraphC_2 \iseq \agraphC_2'$ (up to the node renamings $\amapB_1$, $\amapB_2$), yielding the \gliteral:
\[\replace{\agraphC}{\theta(\amapB_1(\agraphC_1))}{\theta(\amapB_1(\agraphC_1'))} \iseq 
\replace{\agraphC}{\theta(\amapB_2(\agraphC_2))}{\theta(\amapB_2(\agraphC_2'))}.\]
Since $S$ is \saturated and $\false \notin S$, there exist {\graph}s
$\agraphD_i$ (for $i =1,2$) 
such that either $\agraphD_1 = \agraphD_2$ or $\agraphD_1 \iseq \agraphD_2 \in \IF{S}$, and
$\theta'(\replace{\agraphC}{\theta(\amapB_i(\agraphC_i))}{\theta(\amapB_i(\agraphC_i'))})
\rightarrow_{\sRw{S}}^*
\agraphD_i$.
We assume that $\agraphD_1 \iseq \agraphD_2 \in \IF{S}$, the proof if $\agraphD_1 = \agraphD_2$ is similar and simpler.
Since $\ggeq$ is total on ground {\graph}s, we have either 
$\agraphD_1 \ggeq \agraphD_2$ or
$\agraphD_2 \ggeq \agraphD_1$. Assume by symmetry that
$\agraphD_1 \ggeq \agraphD_2$, so that $\Rw{S}$ contains a rule $\agraphD_1 \rightarrow \agraphD_2$.
Using Lemmas \ref{lem:replaceseq} and \ref{lem:rw_subgraph}, we get
$\replace{\agraph}{\theta'(\agraphC)}{\theta'(\replace{\agraphC}{\theta(\amapB_i(\agraphC_i))}{\theta(\amapB_i(\agraphC_i'))})} \rightarrow_{\sRw{S}}^* 
 \replace{\agraph}{\theta'(\agraphC)}{\agraphD_i}$, i.e., (by Lemma \ref{lem:replace_subgraph})
 $\agraph_i \rightarrow_{\sRw{S}}^* 
 \replace{\agraph}{\theta'(\agraphC)}{\agraphD_i}$.
 We deduce that
 $\agraph_2 \rr{0} \
 \replace{\agraph}{\theta'(\agraphC)}{\agraphD_2}$, and (due to the rule $\agraphD_1 \rightarrow \agraphD_2$), 
 $\agraph_1 \rr{\eta} \
 \replace{\agraph}{\theta'(\agraphC)}{\agraphD_2}$, where
 \[\eta = \left\{\begin{array}{l}
 0 \text{\, if $\replace{\agraph}{\theta'(\agraphC)}{\agraphD_2} \glt \agraph$ } \\
 1 \text{\, otherwise}
 \end{array}\right.\]
 Note that if $\agraph_i \glt \agraph$, for some $i = 1,2$ then necessarily
 $\replace{\agraph}{\theta'(\agraphC)}{\agraphD_2} \glt \agraph$ (as $\agraph_i \gleq \agraph$), hence $\eta = 0$. 
 Thus we get the result, with $\agraph' = \replace{\agraph}{\theta'(\agraphC)}{\agraphD_2}$, $\eta_1 = \eta$ and $\eta_2 = 0$.
}
\end{itemize}

\end{proof}

\newcommand{\joinsto}{\downarrow_{\Rw{S}}}

We write $\agraph_1 \joinsto \agraph_2$ if
there exists a \graph $\agraphB$ such that
$\agraph_i \rightarrow_{\Rw{S}}^* \agraphB$, for all $i = 1,2$.

\begin{lemma}
\label{lem:confluence}
If $S$ is \saturated  and $\false \not \in S$, then
the relation $\rightarrow_{\Rw{S}}$ is confluent.
\end{lemma}

\begin{proof}
Assume that there exist {\graph}s $\agraph, \agraph_1, \agraph_2$ such that
$\agraph \rightarrow^*_{\Rw{S}} \agraph_i$ for $i = 1,2$
and 
$\agraph_1 \not \joinsto \agraph_2$.
W.l.o.g., we assume that $\agraph$ is a minimal (w.r.t.\ $\glt$) \graph such that {\graph}s $\agraph_1,\agraph_2$ satisfying the above properties exist.
Thus, for every \graph $\agraph'$ such that $\agraph' \glt \agraph$, 
if $\agraph' \rightarrow^*_{\Rw{S}} \agraph_i'$ for $i = 1,2$
then
$\agraph_1' \joinsto \agraph_2'$.
By definition of the relation $\rr{n}$, there exist natural numbers $n_i$ (for $i = 1,2$) such that
$\agraph \rr{n_i} \agraph_i$.
We prove, by induction on 
the set $(n_1,n_2)$, 
that for every \graph $\agraphB \gequiv \agraph$, if 
 $\agraphB \rr{n_i} \agraph_i$ for $i = 1,2$
then there exists a \graph $\agraph'$ such that
for  $i = 1,2$, we have
$\agraph_i \rr{n_{3-i}} \agraph'$.
Note that this immediately entails the required result, by taking $\agraphB = \agraph$. 
First assume that $\agraphB \in \{ \agraph_1, \agraph_2 \}$, say, $\agraphB = \agraph_1$. Then by taking $\agraph'\isdef \agraph_2$ we have $\agraph_1 \rr{n_2} \agraph_2$, and by Proposition \ref{prop:rr_basic} (\ref{rr:mgn}), $\agraph_2 \rr{n_1} \agraph_2$, hence the result.
We now assume that
$\agraphB \rightarrow_{\Rw{S}} \agraphB_i \rightarrow_{\Rw{S}}^* \agraph_i$.
Note that  we have $\agraphB_i \rr{m_i} \agraph_i$ where by Proposition \ref{prop:rr_basic}, either 
$\agraphB_i \gequiv \agraphB \gequiv \agraph$, 
$n_i > 0$ because $\agraphB_i \not \glt \agraph$ and $m_i = n_i - 1$; or 
$\agraphB_i \glt \agraph$ and $n_i = m_i = 0$.
By Lemma \ref{lem:loc_confluence}, there exists 
$\agraphC$ such that:
$\agraphB_i \rr{\eta_i} \agraphC$, with $\eta_1,\eta_2 \in \{ 0,1 \}$ and 
if $\agraphB_i \glt \agraph$ for some $i = 1,2$ then 
$\eta_1 = \eta_2 = 0$.

Let $i \in \set{ 1,2}$, we have $\agraphB_i \rr{\eta_i} \agraphC$
and
$\agraphB_i \rr{m_i} \agraph_i$.
\begin{itemize}
\item{
If $\agraphB_i \glt \agraph$,
then by minimality of 
$\agraph$ there exists a \graph $\agraphC_i$ such that
$\agraphC \rightarrow_{\Rw{S}}^* \agraphC_i$
and
$\agraph_i \rightarrow_{\Rw{S}}^* \agraphC_i$.
By Proposition \ref{prop:srw-ggeq} we have $\agraphC, \agraph_i \gleq \agraphB_i \glt \agraph$, hence $\agraphC \rr{0} \agraphC_i$
and
$\agraph_i \rr{0} \agraphC_i$ by Proposition \ref{prop:rr_basic} (\ref{rr:0_lt}), and $\agraphC \rr{m_i} \agraphC_i$ and $\agraph_i \rr{\eta_i} \agraphC_i$ by Proposition \ref{prop:rr_basic} (\ref{rr:mgn}). 
}
\item{If $\agraphB_i \gequiv \agraph$,
then $m_i = n_i - 1 < n_i$. Moreover, we have either $n_{3-i} \geq 1 \geq \eta_i$ 
or $n_{3-i} = 0$, and in the latter case $\agraphB_{3-i} \glt \agraph$, so that 
$\eta_i = 0$ by Lemma \ref{lem:loc_confluence}.
In both cases, the pair $(m_i,\eta_i)$ is strictly smaller than the pair $(n_1,n_2)$. Consequently, by the induction hypothesis,
there exists a \graph $\agraphC_i$ such that
$\agraphC \rr{m_i} \agraphC_i$
and
$\agraph_i \rr{\eta_i} \agraphC_i$.
}
\end{itemize}
Thus in both cases we get that for $i=1,2$,
$\agraphC \rr{m_i} \agraphC_i$
and
$\agraph_i \rr{\eta_i} \agraphC_i$, for some graphs $\agraphC_i$.
We again distinguish two cases.
\begin{itemize}
\item{If $\agraphC \glt \agraph$, then by minimality of
$\agraph$ 
there exists a \graph $\agraph'$ such that for $i = 1,2$,
$\agraphC_i \rightarrow_{\Rw{S}}^* \agraph'$, hence
$\agraphC_i \rr{0} \agraph'$ by Proposition \ref{prop:rr_basic} (\ref{rr:0_lt}).
Since $\agraph_i \rr{\eta_i} \agraphC_i$, we deduce that
$\agraph_i \rr{\eta_i} \agraph'$ by Proposition \ref{prop:rr_basic} (\ref{rr:concat}).
If $\eta_1 = \eta_2 = 0$ then this entails that
$\agraph_i \rr{n_{3-i}} \agraph'$ and the proof is completed.
Otherwise, by definition of $\eta_i$, we have $\agraphB_1 \gequiv \agraphB_2 \gequiv \agraph$, hence $n_1, n_2 > 0$ and
$n_1, n_2 \geq \eta_1,\eta_2$, thus we also have $\agraph_i \rr{n_{3-i}} \agraph'$.
}
\item{If $\agraphC \gequiv \agraph$, then necessarily
$\agraphB_i \gequiv \agraph$ for all $i = 1,2$, thus
$m_i = n_i - 1$
and, by applying again the induction hypothesis,
there exists $\agraph'$ such that $\agraphC_i \rr{m_{3-i}} \agraph'$, hence by Proposition \ref{prop:rr_basic} (\ref{rr:concat})
$\agraph_i \rr{m_{3-i} + \eta_i} \agraph'$. Since  $\eta_i \leq 1$, we deduce that
$\agraph_i \rr{n_{3-i}} \agraph'$, and the proof is completed.
}
\end{itemize}

\end{proof}

\begin{definition}
We denote by $\acong_S$ the relation defined as follows:
$\agraph \acong_S \agraphB$ 
iff there exists a \graph $\agraphC$ such that
$\agraph \rightarrow_{\Rw{S}}^* \agraphC$
and
$\agraphB \rightarrow_{\Rw{S}}^* \agraphC$
\end{definition}

\begin{lemma}
\label{lem:acong}
If $S$ is \saturated and does not contain $\false$ then 
 $\acong_S$ is a \gcongruence.
\end{lemma}
\begin{proof}
It is clear that $\acong_S$ is reflexive and symmetric. Since $\iso = \rightarrow_{\Rw{S}}^0$ by definition, it is also \nclosed.
 
We now show that $\acong_S$ is transitive.
If $\agraph_1 \acong_S \agraph_2 \acong_S \agraph_3$
then 
there exist $\agraphB_1$ and $\agraphB_2$ such that
$\agraph_1 \rightarrow_{\Rw{S}}^* \agraphB_1$,
$\agraph_2 \rightarrow_{\Rw{S}}^* \agraphB_1$,
$\agraph_2 \rightarrow_{\Rw{S}}^* \agraphB_2$, and
$\agraph_3 \rightarrow_{\Rw{S}}^* \agraphB_2$.
Since $\rightarrow_{\Rw{S}}^*$ is confluent by Lemma \ref{lem:confluence},
we deduce that there exists a \graph 
$\agraphC$ such that
$\agraphB_1 \rightarrow_{\Rw{S}}^* \agraphC$,
$\agraphB_2 \rightarrow_{\Rw{S}}^* \agraphC$. By transitivity of 
$\rightarrow_{\Rw{S}}^*$ we get that
$\agraph_1 \rightarrow_{\Rw{S}}^* \agraphC$ and
$\agraph_3 \rightarrow_{\Rw{S}}^* \agraphC$, hence
$\agraph_1 \acong_S \agraph_3$. 

There remains to prove that $\acong_S$ is \rclosed.
Consider $\agraph, \agraphB_, \agraphB_1, \agraphB_2$ where 
$\agraphB \sgrc \agraph$, $\agraphB_1 \acong_S \agraphB_2$ and 
$\agraphB_1,\agraphB_2$ are \substitutable for $\agraphB$ in $\agraph$.
By definition of $\acong_S$ there exists a \graph $\agraphC$ such that
$\agraphB_i \rightarrow_{\Rw{S}}^* \agraphC$, for $i = 1,2$.
W.l.o.g.\ we assume that $\agraphC$ is disjoint from $\agraph$ and $\agraphB_i$, hence \substitutable for $\agraphB_i$ in $\replace{\agraph}{\agraphB}{\agraphB_i}$. 
Let $i = 1,2$, by Lemma \ref{lem:rw_subgraph} we have
$\replace{\agraph}{\agraphB}{\agraphB_i}
\rightarrow_{\Rw{S}}^* 
\replace{\replace{\agraph}{\agraphB}{\agraphB_i}}{\agraphB_i}{\agraphC}$
  and by Lemma \ref{lem:replaceseq}
we deduce that
$\replace{\agraph}{\agraphB}{\agraphB_i}
\rightarrow_{\Rw{S}}^*  \replace{\agraph}{\agraphB}{\agraphC}$. Therefore,
$\replace{\agraph}{\agraphB}{\agraphB_1}
\acong_S 
\replace{\agraph}{\agraphB}{\agraphB_2}$.
\end{proof}

\begin{lemma}
\label{lem:model}
If $S$ is \saturated and does not contain $\false$ then 
$\acong_S$ is a model of $S$.
\end{lemma}
\begin{proof}
Let $\agraph \iseq \agraph'$ be a \gequation in $S$ and let $\theta$ be a ground \substitution of the variables in 
$\agraph,\agraph'$.
By definition, we have $\theta(\agraph) \iseq \theta(\agraph') \in \IF{S}$. Since $\ggeq$ is total on ground {\graph}s, 
either $\theta(\agraph) \ggeq \theta(\agraph')$ or
$\theta(\agraph') \ggeq \theta(\agraph)$. We assume by symmetry that
 $\theta(\agraph) \ggeq \theta(\agraph')$, so that $\Rw{S}$ contains a rule 
 $\theta(\agraph) \rightarrow \theta(\agraph')$.
 Then by definition of $\acong_S$, we have 
 $\theta(\agraph) \acong_S \theta(\agraph')$, hence 
 $\acong_S$ is a model of $\agraph \iseq \agraph'$.
 
 Now consider a ground \gdisequation $\agraph \not \iseq \agraph'$ that is \redundant w.r.t.\ $S$; we prove that $\agraph \not \acong_S \agraph'$. This is sufficient to obtain the result, since all
 the ground instances of a \gdisequation occurring in $S$ are \redundant w.r.t.\ $S$.
 Assume for a contradiction that $\agraph \acong_S \agraph'$, so that there exists $\agraphB$ such that
 $\agraph \rightarrow_{\Rw{S}}^{n} \agraphB$
 and $\agraph' \rightarrow_{\Rw{S}}^{n'} \agraphB$, for some $n,n' \in {\Bbb N}$.
 The proof is by induction on the pair $(\{ \agraph,\agraph' \}, \{ n,n' \})$, ordered by the lexicographic and
 multiset extension of the ordering $\glt$ and of the usual order on natural numbers.
 \begin{itemize}
 \item{If there exists $\agraphC$ such that $\agraph \rightarrow_{\sRw{S}} \agraphC$ and 
 $\agraphC \not \iseq \agraph'$ is \redundant w.r.t.\ $S$ 
 then by Proposition \ref{prop:srw-ggeq}, $\agraphC \glt \agraph$, thus by the induction hypothesis
 $\agraphC \not \acong_S \agraph'$. By definition of $\acong_S$ we have 
 $\agraph \acong_S \agraphC$, hence
  $\agraph \not \acong_S \agraph'$.}
  \item{The proof is similar if  there exists $\agraphC$ such that $\agraph' \rightarrow_{\sRw{S}} \agraphC$ and 
 $\agraph \not \iseq \agraphC$ is \redundant w.r.t.\ $S$.}
 \item{Otherwise, $\agraph \not \iseq \agraph'$ must be an instance of a \gliteral in $S$, i.e., there exist a \gdisequation 
 $\agraphB \not \iseq \agraphB'$ occurring in $S$ (up to a renaming of {\node}s) and a \substitution $\sigma$ such that 
 $\agraph = \sigma(\agraphB)$ and  $\agraph' = \sigma(\agraphB')$. We distinguish two cases.
 \begin{itemize}
 \item{
If $n = n' = 0$ then $\sigma(\agraphB) \iso \sigma(\agraphB')$, thus 
$\agraphB$ and $\agraphB'$ are \unifiable (up to a renaming of {\node}s).
Consequently, the rule \REFL\ applies. Since $S$ is \saturated, this entails that $\false \in S$, which contradicts the hypotheses of the lemma.
  }
  \item{Otherwise, 
  we have either $n > 0$ or $n' > 0$.
  Assume that $\agraph \ggt \agraph'$, $n = 0$ and $n' > 0$. Then
  we have $\agraph \iso \agraphB$, thus 
 $\agraph' \rightarrow_{\Rw{S}}^{n'} \agraph$ and 
 $\agraph' \ggeq \agraph$ by Proposition \ref{prop:srw-ggeq}, a contradiction.
 The case where  $\agraph' \ggt \agraph$, with $n > 0$ and $n' = 0$ is symmetric. 
 Now assume that $n > 0$ and that $\agraph \ggeq \agraph'$ (the proof where $n' > 0$ and $\agraph' \ggeq \agraph$ is symmetric).
Then there exists 
a \gequation $\agraphD \iseq \agraphD'$ in $S$ (up to a renaming of {\node}s)
and a \substitution $\theta$ such that $\theta(\agraphD) \sgrc \agraph$, 
$\agraph \rightarrow_{\Rw{S}} \replace{\agraph}{\theta(\agraphD)}{\theta(\agraphD')} \rightarrow_{\Rw{S}}^{n-1}
\agraphC$ and $\theta(\agraphD) \ggeq \theta(\agraphD')$.
We assume, w.l.o.g., that $\agraphD \iseq \agraphD'$ and $\agraphB \not \iseq \agraphB'$ share no variable, so that 
$\sigma$ and $\theta$ have disjoint domains.  
Let $\gamma$ be a most general \substitution such that
$\gamma(\agraphD) \sgrc \gamma(\agraphB)$. By definition, there exists $\gamma'$ such that
$(\sigma \cup \theta) = \gamma' \circ \gamma$.
Since $\theta(\agraphD) \ggeq \theta(\agraphD')$, and $\ggeq$ is \sclosed, 
we have 
$\gamma(\agraphD') \not \ggt \gamma(\agraph')$ (as otherwise we would have
$\gamma'(\gamma(\agraphD')) \ggt \gamma'(\gamma(\agraph'))$, i.e.,
$\theta(\agraphD') \ggt \theta(\agraphD)$).
Similarly, since $\agraph \ggeq \agraph'$ and $\ggeq$ is \sclosed, 
we deduce that $\gamma(\agraphB') \not \ggt  \gamma(\agraphB)$.
This entails that  the rule \NS\ is applicable on 
$\agraphB \not \iseq \agraphB'$ and
$\agraphD \iseq \agraphD'$, yielding:
$\alit: (\replace{\gamma(\agraphB)}{\gamma(\agraphD)}{\gamma(\agraphD')} \not \iseq \gamma(\agraphB')$.
Then $\replace{\agraph}{\theta(\agraphD)}{\theta(\agraphD')} \iseq \agraph'$ is an instance of 
$\alit$, and since $S$ is \saturated, it must be \redundant w.r.t.\ $S$.
By the induction hypothesis, we deduce that 
$\replace{\agraph}{\theta(\agraphD)}{\theta(\agraphD')} \acong_S \agraph'$, thus
$\agraph \acong_S \agraph'$.
  }
  \end{itemize}}
  \end{itemize}
\end{proof}

\begin{theorem}[Completeness]
Any \saturated \formula not containing $\false$ is satisfiable.
\end{theorem}

\begin{proof}
This follows immediately from Lemma \ref{lem:model}.
\end{proof}

\section{An Application to Circuits}\label{sec:circuits}
\newcommand{\Circuits}{\textsc{Circuits}\xspace}
In this section we consider a class of graphs, denoted by {\Circuits}, which intuitively consists of graphs that represent gates. Informally, a gate with $n$ entries and $m$ outputs is represented in this class by a graph with $2(n+m) + 1$ nodes:
\begin{itemize}
	\item one node that contains the information about the number of entries and outputs;
	\item $n$ nodes that represent the entries and $m$ nodes that represent the outputs;
	\item $n+m$ \groot nodes that are used for rewriting operations.
\end{itemize}
We also define the notion of a \emph{subcircuit} and show that it is a \crel, thus guaranteeing the completeness of the proof procedure defined in Section \ref{sec:pf-proc}.


\newcommand{\card}[1]{|#1|}
\newcommand{\trg}[1]{\textsf{trg}(#1)}
\newcommand{\src}[1]{\textsf{src}(#1)}
\newcommand{\inputsin}[1]{\inputs{#1}^{\texttt{out}}}
\newcommand{\inputsout}[1]{\inputs{#1}^{\texttt{in}}}
\newcommand{\seqconf}[2]{\Upsilon^{#1}_{#2}}
\newcommand{\seqcon}[3]{\seqconf{#1}{#2}(#3)}
\newcommand{\parcomp}{\otimes}
\newcommand{\dN}{\mathbb{N}}

\begin{definition}
  \label{circuits}
We denote by \Circuits the class of graphs $\agraph$ satisfying the following properties:
\begin{enumerate}
	\item The set of labels is $\dN \setminus \set{0}$.\label{it:setlabs}
	\item{There are {\sort}s $\gatesort{n}{m}$, $\inportsort$ and $\outportsort$ 
		denoting gates with $n$ entries and $m$ outputs, input ports, and output ports 
		respectively. \label{it:sorts}}
	\item{For each node $\anode$ of \sort $\gatesort{n}{m}$ there exist 
		exactly $n$ nodes $\anode_1,\dots,\anode_n$ of sort $\inportsort$ labeled with $1,\dots,n$ respectively
		and $m$ nodes $\anode_1',\dots,\anode_m'$ of sort $\outportsort$ labeled with $1,\dots,m$ respectively,
		and edges $\mkedge{\anode_i}{\anode}$
		and  $\mkedge{\anode}{\anode_j'}$ for every $i \in \interv{1}{n}$ and $j \in \interv{1}{m}$. \label{it:gate-edges} }
	\item All nodes of {\sort} $\inportsort$ have exactly one incoming and at most one outgoing  edge; all nodes of {\sort} $\outportsort$ have exactly one outgoing and at most one incoming edge. 
	\label{it:in-out}
	\item Every node of sort $\inportsort$ which is labeled by $i\in \dN \setminus \set{0}$ admits  one outgoing edge, the target of which is a node of sort $\gatesort{n}{m}$, where $i \leq n$.\label{it:in-lab}
	\item Every node of sort $\outportsort$ which is labeled by $i\in \dN \setminus \set{0}$ admits  one incoming edge, the source of which is a node of sort $\gatesort{n}{m}$, where $i \leq m$.\label{it:out-lab}
	\item{All the other edges are of the form  $\mkedge{\anode}{\anodeB}$ where
		$\anode$ is of sort $\outportsort$ and $\anodeB$ is of sort $\inportsort$.\label{it:out-to-in}}	
	\item $\inputs{\agraph}$ is of the form $R_f\cdot R_i$, 
	where every node in the sequence $R_f$ is of \sort $\outportsort$ 
	and admits no incoming edge, and every node in the sequence $R_i$ is of \sort $\inportsort$ and admits no outgoing edge. The sequence $R_f$ is denoted by $\inputsout{\agraph}$, and $R_i$ by $\inputsin{\agraph}$. \label{it:circ-inputs}
      \end{enumerate}
\end{definition}
\newcommand{\emptygraph}{{\boldsymbol \emptyset}}
Intuitively, $\inputsout{\agraph}$ denotes the {\bf inputs}
of $\agraph$ and $\inputsin{\agraph}$ its {\bf outputs}.

\begin{example}
The empty \graph $\emptygraph \isdef (\emptyset,\tuple{},\emptyset,\emptyset)$, is in \Circuits, with $\inputsout{\emptygraph} = \inputsin{\emptygraph} = \tuple{}$.
\end{example}


\begin{definition}
	Given a node $\anode$ in $\inputsout{\agraph}$, we denote by $\trg{\anode}$ the node $\anodeB$ such that $\mkedge{\anode}{\anodeB} \in \edges{\agraph}$. Given a node $\anode$ in $\inputsin{\agraph}$, we denote by $\src{\anode}$ the node $\anodeB$ such that $\mkedge{\anodeB}{\anode} \in \edges{\agraph}$. 
\end{definition}

Since every node $\anode$ in $\inputsout{\agraph}$ is of sort $\outportsort$, the node $\anodeB$ such that $\mkedge{\anode}{\anodeB} \in \edges{\agraph}$ exists and is unique; similarly, if $\anode$ in $\inputsin{\agraph}$ then the node $\anodeB$ such that $\mkedge{\anodeB}{\anode} \in \edges{\agraph}$ exists and is unique.

\newcommand{\subcirc}{\sgr_{\mathrm{C}}}

\begin{definition}
	$\agraphB$ is a \emph{subcircuit of a \graph} $\agraph \in \Circuits$, denoted by $\agraphB \subcirc \agraph$, if and only if $\agraphB \sgr \agraph$ and $\agraphB \in \Circuits$.
\end{definition}

\begin{lemma}\label{lm:repl-subcirc}	
	Assume that $\agraphB \subcirc \agraph$, let $\agraphB'$ be a \graph in \Circuits that is \substitutable for $\agraphB$ in $\agraph$ and let $\agraph' \isdef \replace{\agraph}{\agraphB}{\agraphB'}$. Then $\agraph' \in \Circuits$.
\end{lemma}

\begin{proof}
	Let $\amap \isdef \amapsub{\agraphB}{\agraphB'}{\agraph}$ (see Definition \ref{def:replace}), we prove that $\agraph' \in \Circuits$. Items (\ref{it:setlabs}) and (\ref{it:sorts}) of the Definition~\ref{circuits} are immediate to verify.
	\begin{description}
		\item[Item \ref{it:gate-edges}] Let $\anode \in \nodes{\agraph'}$ be a node of sort $\gatesort{n}{m}$. Note that  $\anode$ cannot be in $\inputs{\agraph} \cup \inputs{\agraphB'}$. If $\anode \in \nodes{\agraphB'}$, then because $\agraphB' \in \Circuits$, it is guaranteed that there exist exactly $n$ nodes $\anode_1,\dots,\anode_n$ of sort $\inportsort$ labeled with $1,\dots,n$ respectively and $m$ nodes $\anode_1',\dots,\anode_m'$ of sort $\outportsort$ labeled with $1,\dots,m$ respectively,
		and edges $\mkedge{\anode_i}{\anode}$ and  $\mkedge{\anode}{\anode_j'}$ in $\edges{\agraphB'}$ for all $i \in \interv{1}{n}$ and $j \in \interv{1}{m}$. Since $\edges{\agraphB'} \subseteq \edges{\agraph'}$, we have the result.
		
		Otherwise, $\anode \in \nodes{\agraph}\setminus \nodes{\agraphB}$. Let $i \in \interv{1}{n}$, we prove that there exists a node $\anode_i'$ of sort $\inportsort$ labeled with $i$ such that $\mkedge{\anode_i'}{\anode} \in \edges{\agraph'}$. The proof for nodes of sort $\outportsort$ is similar. Since $\agraph \in \Circuits$, there exists a node $\anode_i \in \nodes{\agraph}$ of sort $\inportsort$ labeled with $i$ such that $\mkedge{\anode_i}{\anode} \in \edges{\agraph}$. Necessarily $\mkedge{\anode_i}{\anode} \in \edges{\agraph} \setminus \edges{\agraphB}$, and we have the result by taking $\anode_i' \isdef \amap(\anode_i)$. 
		\item[Item \ref{it:in-out}] Consider a node $\anode \in \nodes{\agraph'}$ of sort $\inportsort$, and first assume that $\anode \in \nodes{\agraph}\setminus \nodes{\agraphB}$, so that $\amap(\anode) = \anode$ and $\anode \notin \nodes{\agraphB'}$ by Proposition \ref{prop:substitutable}. By definition there exists a unique edge $\mkedge{\anode'}{\anode} \in \edges{\agraph} \setminus \edges{\agraphB}$, and since there can be no edge of the form $\mkedge{\anode''}{\anode}$ in $\edges{\agraphB'}$, $\anode$ admits $\amap(\mkedge{\anode'}{\anode}) = \mkedge{\amap(\anode')}{\anode}$ as a unique incoming edge. 			
		If $\anode \in \nodes{\agraphB'}$ then by definition $\amap(\anode) = \anode$ and by Proposition \ref{prop:substitutable}, $\anode \notin \nodes{\agraph}\setminus \nodes{\agraphB}$ because $\anode \notin \nodes{\agraphB}$. Since $\agraphB \in \Circuits$, the proof in this case is straightforward.
		
		The proof that $\anode$ admits at most one outgoing edge is carried out in a similar way, and so are the proofs when $\anode$ is of sort $\outportsort$.
		\item[Item \ref{it:in-lab}] Consider a node $\anode \in \nodes{\agraph'}$ of sort $\inportsort$ and labeled by $i$. A case analysis depending on whether $\anode \in \nodes{\agraph}\setminus \nodes{\agraphB}$ or $\anode \in \nodes{\agraphB'}$ as in the previous point shows that the target of $\anode$ is of sort $\gatesort{n}{m}$ for some $n\geq i$.
		\item[Item \ref{it:out-lab}] The proof is the same as in the previous case.
		\item[Item \ref{it:out-to-in}] Because $\amap$ preserves sorts, it is straightforward to prove that all other edges are of the form $\mkedge{\anode}{\anodeB}$ where $\anode$ is of sort $\outportsort$ and $\anodeB$ is of sort $\inportsort$.
		\item[Item \ref{it:circ-inputs}] By definition we have $\inputs{\agraph'} = \amap(\inputs{\agraph})$, and since $\agraph$ and $\agraphB'$ are both in \Circuits, $\inputs{\agraph'}$ is of the required form.
	\end{description}
\end{proof}

\begin{lemma}\label{lm:circuit-merge}
	Assume that $\agraphB \subcirc \agraph$, $\agraphB' \subcirc \agraph$ and that $\agraphB$ and $\agraphB'$ are \compatible. Then the $E$-\merge of $\agraphB$ and $\agraphB'$ with \[E\ \isdef\ \setof{\anyedge{\anode}{\anodeB}}{\anode \in \inputs{\agraphB} \setminus \nodes{\agraphB'},\, \anodeB\in \inputs{\agraphB'} \setminus \nodes{\agraphB},\, \anyedge{\anode}{\anodeB} \in \edges{\agraph}}\]  is in \Circuits.
\end{lemma}

\begin{proof}
	Let $\agraph'$ denote the  $E$-\merge of $\agraphB$ and $\agraphB'$; by Proposition \ref{prop:merge}, $\agraph'\sgr \agraph$. It is straightforward to verify that $\agraph' \in \Circuits$: the existence of nodes and edges satisfying the requirements of the definition of \Circuits are deduced from the fact that $\agraphB$ and $\agraphB'$ are both in $\Circuits$, and the uniqueness properties are consequences of the fact that $\agraphB$ and $\agraphB'$ are both subgraphs of $\agraph$ which is in \Circuits.
\end{proof}

\begin{lemma}
	The relation $\subcirc$ is a \crel, with $\aclass = \Circuits$.
\end{lemma}

\begin{proof}
	First note that since $\agraphB \in \Circuits$, if $\anode \in \nodes{\agraph}\setminus \nodes{\agraphB}$ is of sort $\inportsort$ then there can be no edge $\mkedge{\anode_i}{\anode} \in \edges{\agraphB}$. Indeed, necessarily, $\anode_i \in \inputsout{\agraphB}$ and by definition, $\anode_i$ cannot admit any outgoing edge. Similarly, if $\anode \in \nodes{\agraph}\setminus \nodes{\agraphB}$ is of sort $\outportsort$ then there can be no edge $\mkedge{\anode}{\anode_i} \in \edges{\agraphB}$.
	\begin{enumerate}
		\item By definition, if $\agraphB \subcirc \agraph$ then $\agraph, \agraphB \in \Circuits$.
		\item By Lemma \ref{lm:sub-subg} we have $\agraphB' \sgr \replace{\agraph}{\agraphB}{\agraphB'}$ and by Lemma \ref{lm:repl-subcirc} $\replace{\agraph}{\agraphB}{\agraphB'}\in \Circuits$.
		\item It is clear that if $\agraph \in \Circuits$ then $\amap(\agraph) \in \Circuits$ for all {\nsubstitution}s $\amap$. By Proposition \ref{prop:mapsgr}, we deduce that if $\agraphB \subcirc \agraph$ then $\amap(\agraphB) \subcirc \amap(\agraph)$.
		\item 
		By Lemma \ref{lem:replace_subs}  we have $\agraphB \sgr \replace{\agraph}{\agraphC}{\agraphC'}$ and by Lemma \ref{lm:repl-subcirc}, $\replace{\agraph}{\agraphC}{\agraphC'}\in \Circuits$.
		\item If $\agraphB \subcirc \agraph$ and $\agraphB' \subcirc \agraph$ then the \merge $\agraph'$ of $\agraphB$ and $\agraphB'$ considered in Lemma \ref{lm:circuit-merge} is such that $\agraph' \subcirc \agraph$.
	\end{enumerate}
\end{proof}

\subsection{Parallel and Sequential Composition}

It is natural to construct a circuit by composing simpler components. We introduce two standard composition operations: parallel and sequential composition. The empty graph is the neutral element for parallel composition, and we introduce the class of so-called \emph{identity graphs}, which can be viewed as neutral elements for sequential composition.

\begin{definition}
	Consider two graphs $\agraph_1, \agraph_2$ in \Circuits, such that $\nodes{\agraph_1} \cap \nodes{\agraph_2} = \emptyset$.
	We define the \emph{parallel composition of $\agraph_1$ and $\agraph_2$}, denoted by $\agraph_1\parcomp \agraph_2$, as the graph $\agraph_3$ constructed as follows. 
	\begin{enumerate}
		\item $\nodes{\agraph_3} \isdef {\nodes{\agraph_1}} \cup {\nodes{\agraph_2}}$;
		\item $\inputs{\agraph_3} \isdef \inputsout{\agraph_1}\cdot \inputsout{\agraph_2}\cdot \inputsin{\agraph_1} \cdot \inputsin{\agraph_2}$;
		\item $\edges{\agraph_3} \isdef {\edges{\agraph_1}} \cup {\edges{\agraph_2}}$;
		\item $\labs{\agraph_3} \isdef \labs{\agraph_1} \cup \labs{\agraph_2}$.
	\end{enumerate}
\end{definition}
It is straightforward to check that $\agraph_2\parcomp \agraph_1$ is a \graph when $\agraph_1$ and $\agraph_2$ are disjoint.

\begin{proposition}
	Parallel composition is an associative operation 
	on mutually disjoint {\graph}s, and the empty \graph $\emptygraph$ is its neutral element.
\end{proposition}

\newcommand{\seqcomp}[2]{#2 \circ #1}

\begin{definition}
\label{def:parcomp}
	Consider two graphs $\agraph_1, \agraph_2$ in \Circuits, such that $\nodes{\agraph_1} \cap \nodes{\agraph_2} = \emptyset$, where $\inputsin{\agraph_1} = \tuple{\anode_1,\ldots, \anode_n}$ and $\inputsout{\agraph_2} = \tuple{\anodeB_1,\ldots, \anodeB_n}$. We let $\seqconf{\agraph_2}{\agraph_1}$ denote the mapping on $\nodes{\agraph_1} \cup \nodes{\agraph_2}$ defined as follows:
	\begin{itemize}
		\item For all $ i = 1, \ldots, n$, $\seqcon{\agraph_2}{\agraph_1}{\anode_i} \isdef \trg{\anodeB_i}$ and $\seqcon{\agraph_2}{\agraph_1}{\anodeB_i} \isdef \src{\anode_i}$;
		\item For all other nodes $\anode \in \nodes{\agraph_1} \cup \nodes{\agraph_2}$, $\seqcon{\agraph_2}{\agraph_1}{\anode} \isdef \anode$.
	\end{itemize}
	We define the \emph{sequential composition of $\agraph_1$ and $\agraph_2$}, denoted by $\seqcomp{\agraph_1}{\agraph_2}$, as the graph $\agraph_3$ constructed as follows. 
	\begin{enumerate}
		\item $\nodes{\agraph_3} \isdef \seqcon{\agraph_2}{\agraph_1}{\nodes{\agraph_1}} \cup \seqcon{\agraph_2}{\agraph_1}{\nodes{\agraph_2}}$;
		\item $\inputs{\agraph_3} \isdef \inputsout{\agraph_1}\cdot \inputsin{\agraph_2}$;
		\item $\edges{\agraph_3} \isdef \seqcon{\agraph_2}{\agraph_1}{\edges{\agraph_1}} \cup \seqcon{\agraph_2}{\agraph_1}{\edges{\agraph_2}}$;
		\item $\labs{\agraph_3} \isdef \labs{\agraph_1} \cup \labs{\agraph_2}$.
	\end{enumerate}
\end{definition}
It is straightforward to check that $\seqconf{\agraph_2}{\agraph_1}$ is a \graph when $\agraph_1$ and $\agraph_2$ satisfy the conditions of Definition~\ref{def:parcomp}.

\begin{example}\label{ex:seq-comp}
	Consider the graph $\agraph_3$ defined as follows:
	\begin{eqnarray*}
		\nodes{\agraph_3} & \isdef & \set{\anode_1, \ldots, \anode_7}\\
		\inputs{\agraph_3} & \isdef &\tuple{\anode_1, \anode_5, \anode_7}\\
		\edges{\agraph_3} & \isdef & \set{\mkedge{\anode_1}{\anode_2},\, \mkedge{\anode_2}{\anode_3},\, \mkedge{\anode_3}{\anode_4},\, \mkedge{\anode_3}{\anode_6},\, \mkedge{\anode_4}{\anode_5},\, \mkedge{\anode_6}{\anode_7}}\\
		\labs{\agraph_3} & \isdef & \set{(\anode_2, 1),\, (\anode_3, f),\, (\anode_4, 1),\, (\anode_6, 2)}
	\end{eqnarray*}
	This graph is an element of \Circuits, $\anode_3$ is of sort $\gatesort{1}{2}$, nodes $\anode_2, \anode_5$ and $\anode_7$ are of sort $\inportsort$ and nodes $\anode_1$, $\anode_4$ and $\anode_6$ are of sort $\outportsort$.
	
	We also consider the graphs $\agraph_2$ and $\agraph_1$, defined as follows:
	\begin{eqnarray*}
		\nodes{\agraph_2} & \isdef & \set{\anodeB_1, \ldots, \anodeB_5}\\
		\inputs{\agraph_2} & \isdef &\tuple{\anodeB_1, \anodeB_5}\\
		\edges{\agraph_2} & \isdef & \set{\mkedge{\anodeB_1}{\anodeB_2},\, \mkedge{\anodeB_2}{\anodeB_3},\, \mkedge{\anodeB_3}{\anodeB_4},\, \mkedge{\anodeB_4}{\anodeB_5}}\\
		\labs{\agraph_2} & \isdef & \set{(\anodeB_2, 1),\, (\anodeB_3, g),\, (\anodeB_4, 1)}\\
		\nodes{\agraph_1} & \isdef & \set{v_1, v_2, v_3}\\
		\inputs{\agraph_1} & \isdef &\tuple{v_1}\\
		\edges{\agraph_1} & \isdef & \set{\mkedge{v_1}{v_2},\, \mkedge{v_2}{v_3}}\\
		\labs{\agraph_1} & \isdef & \set{(v_2, 1),\, (v_3, h)}
	\end{eqnarray*}
	Node $\anodeB_3$ is of sort $\gatesort{1}{1}$, nodes $\anodeB_1$ and $\anodeB_4$ are of sort $\outportsort$; and nodes $\anodeB_2$ and $\anodeB_5$ is of sort $\inportsort$. Node $v_3$ is of sort $\gatesort{1}{0}$, node $v_1$ is of sort $\outportsort$ and node $v_2$ is of sort $\inportsort$. 	
	All graphs are depicted in Figure \ref{fig:seq-comp}, 
 and the graph $\seqcomp{\agraph_3}{(\agraph_2 \parcomp \agraph_1)}$  is depicted in Figure \ref{fig:seq-comp2}.
\end{example}

\begin{figure}
	\begin{center}
	\begin{minipage}{0.4\textwidth}
	\begin{center}
		{\tiny
		\begin{tikzpicture}
			\node[inputNode, thick] (i1) at (6, 0) {$\anode_1$};
			\node[inputNode, thick] (i2) at (7, 0) {$\anode_2$};
			\node[inputNode, thick] (i3) at (8, 0) {$\anode_3$};			
			\node[inputNode, thick] (i4) at (9, 0.75) {$\anode_4$};
			\node[inputNode, thick] (i5) at (10, 0.75) {$\anode_5$};
			\node[inputNode, thick] (i6) at (9, -0.75) {$\anode_6$};
			\node[inputNode, thick] (i7) at (10, -0.75) {$\anode_7$};
			\node[draw=none] (h1) at (11, 0) {};
			\draw[stateTransition] (i1) -- (i2);
			\draw[stateTransition] (i2) -- (i3);
			\draw[stateTransition] (i3) -- (i4);
			\draw[stateTransition] (i3) -- (i6);
			\draw[stateTransition] (i4) -- (i5);
			\draw[stateTransition] (i6) -- (i7);
			
		\end{tikzpicture}
	}
	\end{center}
	\end{minipage}
	\begin{minipage}{0.4\textwidth}
		\begin{center}
			{\tiny
				\begin{tikzpicture}
					\node[inputNode, thick] (i1) at (6, 0.75) {$\anodeB_1$};
					\node[inputNode, thick] (i2) at (7, 0.75) {$\anodeB_2$};
					\node[inputNode, thick] (i3) at (8, 0.75) {$\anodeB_3$};			
					\node[inputNode, thick] (i4) at (9, 0.75) {$\anodeB_4$};
					\node[inputNode, thick] (i5) at (10, 0.75) {$\anodeB_5$};
					\node[inputNode, thick] (j1) at (6, -0.75) {$v_1$};
					\node[inputNode, thick] (j2) at (7, -0.75) {$v_2$};
					\node[inputNode, thick] (j3) at (8, -0.75) {$v_3$};	
					\node[draw=none] (h2) at (5, 0) {};
					\draw[stateTransition] (i1) -- (i2);
					\draw[stateTransition] (i2) -- (i3);
					\draw[stateTransition] (i3) -- (i4);
					\draw[stateTransition] (i4) -- (i5);
					\draw[stateTransition] (j1) -- (j2);
					\draw[stateTransition] (j2) -- (j3);					
				\end{tikzpicture}
			}
		\end{center}
	\end{minipage}
\end{center}
\caption{Graphs $\agraph_3$, $\agraph_1$ and $\agraph_2$   in Example \ref{ex:seq-comp}}\label{fig:seq-comp}
\end{figure}

\begin{figure}
	\begin{center}
		{\tiny
			\begin{tikzpicture}
				\node[inputNode, thick] (i1) at (6, 0) {$\anode_1$};
				\node[inputNode, thick] (i2) at (7, 0) {$\anode_2$};
				\node[inputNode, thick] (i3) at (8, 0) {$\anode_3$};			
				\node[inputNode, thick] (i4) at (9, 0.75) {$\anode_4$};
				\node[inputNode, thick] (i6) at (9, -0.75) {$\anode_6$};
				\node[inputNode, thick] (j2) at (12, 0.75) {$\anodeB_2$};
				\node[inputNode, thick] (j4) at (12, -0.75) {$v_2$};
				\node[inputNode, thick] (j8) at (13, -0.75) {$v_3$};
				\node[inputNode, thick] (j5) at (13, 0.75) {$\anodeB_3$};
				\node[inputNode, thick] (j6) at (14, 0.75) {$\anodeB_4$};
				\node[inputNode, thick] (j7) at (15, 0.75) {$\anodeB_5$};
				\draw[stateTransition] (i1) -- (i2);
				\draw[stateTransition] (i2) -- (i3);
				\draw[stateTransition] (i3) -- (i4);
				\draw[stateTransition] (i3) -- (i6);
				\draw[stateTransition] (i4) -- (j2);
				\draw[stateTransition] (i6) -- (j4);
				\draw[stateTransition] (j2) -- (j5);
				\draw[stateTransition] (j4) -- (j8);
				\draw[stateTransition] (j5) -- (j6);
				\draw[stateTransition] (j6) -- (j7);				
			\end{tikzpicture}
		}
	\end{center}
	\caption{Sequential composition of $\agraph_2 \parcomp \agraph_1$ and $\agraph_3$ (Example \ref{ex:seq-comp})}\label{fig:seq-comp2}
\end{figure}

\begin{proposition}
	If the graphs $\agraph_1$, $\agraph_2$ and $\agraph_3$ in \Circuits are 
	mutually disjoint and such that $\seqcomp{\agraph_2}{\agraph_3}$ and $\seqcomp{\agraph_1}{\agraph_2}$ are well-defined, then the sequential compositions $\seqcomp{\agraph_1}{(\seqcomp{\agraph_2}{\agraph_3})}$ and $\seqcomp{(\seqcomp{\agraph_1}{\agraph_2})}{\agraph_3}$ are well-defined and equal.
\end{proposition}

We have the following commutation rule that shows that parallel and sequential composition can be switched:

\begin{lemma}
	Consider the pairwise disjoint graphs $\agraph_1, \agraph_2, \agraph_3, \agraph_4$ in \Circuits, and assume $\card{\inputsout{\agraph_1}} = \card{\inputsin{\agraph_2}}$ and $\card{\inputsout{\agraph_3}} = \card{\inputsin{\agraph_4}}$.
	Then we have:
	\[\seqcomp{(\agraph_1\parcomp \agraph_3)}{(\agraph_2\parcomp \agraph_4)}\ =\ (\seqcomp{\agraph_1}{\agraph_2})\parcomp (\seqcomp{\agraph_3}{\agraph_4}).\]
\end{lemma}

\begin{proof}
	We define the following graphs:
	\[\begin{array}{rclcrclcrcl}
		\agraphB_1 & \isdef & \seqcomp{\agraph_1}{\agraph_2} & \quad & \agraphB_2 & \isdef & \seqcomp{\agraph_3}{\agraph_4} & \quad & \agraphB_3 & \isdef & \agraphB_1\parcomp \agraphB_2\\
		\agraphC_1 & \isdef & \agraph_1 \parcomp \agraph_3 & &  \agraphC_2 & \isdef & \agraph_2 \parcomp \agraph_4 & & \agraphC_3 & \isdef& \seqcomp{\agraphC_1}{\agraphC_2}
	\end{array}\]
	and prove that $\agraphB_3 = \agraphC_3$.
	
	First note that by construction, we have $\inputsout{\agraphC_1} = \inputsout{\agraph_1}\cdot \inputsout{\agraph_3}$ and $\inputsin{\agraphC_2} = \inputsin{\agraph_2}\cdot\inputsin{\agraph_4}$. Since $\card{\inputsout{\agraph_1}} = \card{\inputsin{\agraph_2}}$ and $\card{\inputsout{\agraph_3}} = \card{\inputsin{\agraph_4}}$ by hypothesis, $\agraphC_3$ is well-defined. Furthermore, because all graphs are pairwise disjoint, it is straightforward to verify that $\seqconf{\agraphC_2}{\agraphC_1} = \seqconf{\agraph_2}{\agraph_1} \cup \seqconf{\agraph_4}{\agraph_3}$.
	\begin{enumerate}
		\item \begin{eqnarray*}
			\nodes{\agraphB_3} &= &\nodes{\agraphB_1} \cup \nodes{\agraphB_2} \\
			& = & \seqcon{\agraph_2}{\agraph_1}{\nodes{\agraph_1}} \cup \seqcon{\agraph_2}{\agraph_1}{\nodes{\agraph_2}} \cup \seqcon{\agraph_4}{\agraph_3}{\nodes{\agraph_3}} \cup \seqcon{\agraph_4}{\agraph_3}{\nodes{\agraph_4}}\\
			& = & \seqcon{\agraphC_2}{\agraphC_1}{\nodes{\agraph_1}} \cup \seqcon{\agraphC_2}{\agraphC_1}{\nodes{\agraph_2}}\cup \seqcon{\agraphC_2}{\agraphC_1}{\nodes{\agraph_3}}\cup \seqcon{\agraphC_2}{\agraphC_1}{\nodes{\agraph_4}}\\
			& = & \seqcon{\agraphC_2}{\agraphC_1}{\nodes{\agraph_1}\cup \nodes{\agraph_3}} \cup \seqcon{\agraphC_2}{\agraphC_1}{\nodes{\agraph_2}\cup \nodes{\agraph_4}}\\
			& =& \seqcon{\agraphC_2}{\agraphC_1}{\nodes{\agraphC_1}} \cup \seqcon{\agraphC_2}{\agraphC_1}{\nodes{\agraphC_2}}\\
			& = & \nodes{\agraphC_3}.
		\end{eqnarray*}
		\item \begin{eqnarray*}
			\inputs{\agraphB_3} & = & \inputsout{\agraphB_1}\cdot \inputsout{\agraphB_2}\cdot \inputsin{\agraphB_1}\cdot \inputsin{\agraphB_2}\\
			& = & \inputsout{\agraph_1}\cdot \inputsout{\agraph_3}\cdot \inputsin{\agraph_2}\cdot \inputsin{\agraph_4}\\
			& = & \inputsout{\agraphC_1}\cdot \inputsin{\agraphC_2}\\
			& = & \inputs{\agraphC_3}.
		\end{eqnarray*}
		\item \begin{eqnarray*}
			\edges{\agraphB_3} & = & \edges{\agraphB_1} \cup \edges{\agraphB_2}\\
			& = & \seqcon{\agraph_2}{\agraph_1}{\edges{\agraph_1}} \cup \seqcon{\agraph_2}{\agraph_1}{\edges{\agraph_2}} \cup \seqcon{\agraph_4}{\agraph_3}{\edges{\agraph_3}} \cup \seqcon{\agraph_4}{\agraph_3}{\edges{\agraph_4}}\\
			& = & \seqcon{\agraphC_2}{\agraphC_1}{\edges{\agraph_1}} \cup \seqcon{\agraphC_2}{\agraphC_1}{\edges{\agraph_2}} \cup \seqcon{\agraphC_2}{\agraphC_1}{\edges{\agraph_3}} \cup \seqcon{\agraphC_2}{\agraphC_1}{\edges{\agraph_4}}\\
			& = & \seqcon{\agraphC_2}{\agraphC_1}{\edges{\agraph_1} \cup \edges{\agraph_3}} \cup \seqcon{\agraphC_2}{\agraphC_1}{\edges{\agraph_2} \cup \edges{\agraph_4}}\\
			& = & \seqcon{\agraphC_2}{\agraphC_1}{\edges{\agraphC_1}} \cup \seqcon{\agraphC_2}{\agraphC_1}{\edges{\agraphC_2}}\\
			& = & \edges{\agraphC_3}.		
		\end{eqnarray*}
	\item \begin{eqnarray*}
		\labs{\agraphB_3} & = & \labs{\agraphB_1} \cup \labs{\agraphB_2}\\
		& = & \labs{\agraph_1} \cup \labs{\agraph_2} \cup \labs{\agraph_3} \cup \labs{\agraph_4}\\
		& = & \labs{\agraphC_1} \cup \labs{\agraphC_2}\\
		& = & \labs{\agraphC_3}.
	\end{eqnarray*}
	\end{enumerate}
\end{proof}

\newcommand{\idty}{\mathrm{id}}
\newcommand{\idn}[1]{\idty_{#1}}
\newcommand{\idtygraph}{$\idty$-graph\xspace}
\newcommand{\idngraph}[1]{$\idn{#1}$-graph\xspace}
\newcommand{\idngr}[1]{I_{#1}}
We define the class of \emph{identity graphs}. Intuitively, these graphs can be viewed as parallel compositions of wires. Their sequential composition with a graph leaves the latter unchanged, up to a renaming of nodes.


\begin{definition}
	A graph $\agraph$ is an \emph{\idtygraph} if it is of the following form:
	\begin{itemize}
		\item $\nodes{\agraph} = \set{\anode_1, \anode_2}$,
		\item $\inputs{\agraph} = \tuple{\anode_1, \anode_2}$,
		\item $\edges{\agraph} = \set{\mkedge{\anode_1}{\anode_2}}$,
		\item $\labs{\agraph} = \emptyset$ 
		\item $\anode_1$ is of \sort $\outportsort$ and $\anode_2$ is of sort $\inportsort$.
	\end{itemize}
	The empty graph is an \emph{\idngraph{0}} and for $k > 0$, a graph $\agraph$ is an \emph{\idngraph{k}} if it is of the form $\agraphB \parcomp \agraphB'$, where $\agraphB$ is an \idtygraph and $\agraphB'$ is an \idngraph{k-1}.
\end{definition}
It is easy to check that $\agraph$ is an \idngraph{1} iff it is an \idtygraph and  that 
every \idngraph{k} is in \Circuits.
In what follows, we will denote by $\idngr{k}$ any \idngraph{k}.



\begin{proposition}
	Consider a graph $\agraph$ in \Circuits, and let $n\isdef \inputsin{\agraph}$ and $m\isdef \inputsout{\agraph}$. Then $\agraph \iso (\seqcomp{\idngr{n}}{\agraph})$ and $\agraph \iso (\seqcomp{\agraph}{\idngr{m}})$.
\end{proposition}

\begin{proof}
	Let $\agraphB \isdef \seqcomp{\idngr{n}}{\agraph}$, we prove that $\agraph \iso \agraphB$, the proof for the other statement is similar. Let $\inputsout{\agraph} = \tuple{\anode_1, \ldots, \anode_n}$, $\inputsout{\idngr{n}} = \tuple{\anodeB_1,\ldots, \anodeB_n}$ and $\inputsin{\idngr{n}} = \tuple{\anodeB_1',\ldots, \anodeB_n'}$. Let $\Upsilon \isdef \seqconf{\agraph}{\idngr{n}}$. Note that by construction, for all $i = 1, \ldots, n$, we have $\Upsilon{(\anodeB_i')} = \trg{\anode_i}$ and $\Upsilon{(\anode_i)} = \src{\anodeB_i'} = \anodeB_i$.
	
	Consider an \nsubstitution $\amap$ such that for all $i = 1, \ldots, n$, $\amap(\anodeB_i) = \anode_i$ and for all $\anode \in \innernodes{\agraph}$, $\amap(\anode) = \anode$. We show that $\amap(\agraphB) = \agraph$. 
	\begin{itemize}
		\item \begin{eqnarray*}
			\nodes{\amap(\agraphB)} & = & \amap(\Upsilon(\nodes{\agraph}) \cup \Upsilon(\nodes{\idngr{n}}))\\
			& = & \amap\left((\nodes{\agraph} \setminus \setof{\anode_i}{i = 1, \ldots, n}) \cup \setof{\anodeB_i}{i = 1,\ldots, n}\right)\\
			& = & (\nodes{\agraph} \setminus \setof{\anode_i}{i = 1, \ldots, n}) \cup \setof{\anode_i}{i = 1, \ldots, n}\\
			&  = & \nodes{\agraph}.
		\end{eqnarray*}
		\item \begin{eqnarray}
			\inputs{\amap(\agraphB)} & = & \amap(\inputsout{\agraph}\cdot \inputsin{\idngr{n}})\\
			& = & \amap(\inputsout{\idngr{n}})\cdot \amap(\inputsin{\agraph})\\
			& = & \inputsout{\agraph} \cdot \inputsin{\agraph}\\
			& = & \inputs{\agraph}.
		\end{eqnarray}
		\item Let $E \isdef \setof{\mkedge{\anode_i}{\trg{\anode_i}}}{i = 1, \ldots, n}$, so that no node $\anode_i$ occurs in $\edges{\agraph} \setminus E$. We have:
		\begin{eqnarray*}
			\edges{\amap(\agraphB)} & = & \amap(\Upsilon(\edges{\agraph}) \cup \Upsilon(\edges{\idngr{n}}))\\
			& = & \amap(\Upsilon(\edges{\agraph} \setminus E)) \cup \amap(\Upsilon(E)) \cup \amap(\Upsilon(\edges{\idngr{n}}))\\
			& = & \amap(\edges{\agraph} \setminus E) \cup \amap(\setof{\mkedge{\anodeB_i}{\trg{\anode_i}}}{i = 1,\ldots, n})\\
			& = & (\edges{\agraph} \setminus E) \cup E\\
			& = & \edges{\agraph}.
		\end{eqnarray*}
		\item $\labs{\agraphB} = \amap(\labs{\agraph} \cup \labs{\idngr{n}}) = \amap(\labs{\agraph}) = \labs{\agraph}$.
	\end{itemize}
\end{proof}

\newcommand{\swapg}{swap-graph\xspace}

Another class of graphs that are frequently used in the construction of circuits are \emph{{\swapg}s}, which intuitively permit to change the order of wires.

\begin{definition}
	A graph $\agraph$ is a \emph{\swapg} if it is of the following form:
	\begin{itemize}
		\item $\nodes{\agraph} = \set{\anode_1, \anode_2, \anode_3, \anode_4}$,
		\item $\inputs{\agraph} = \tuple{\anode_1, \anode_2, \anode_3, \anode_4}$,
		\item $\edges{\agraph} = \set{\mkedge{\anode_1}{\anode_4},\, \mkedge{\anode_2}{\anode_3}}$,
		\item $\labs{\agraph} = \emptyset$, 
		\item  Nodes $\anode_1$ and $\anode_2$ are of \sort $\outportsort$ and nodes $\anode_3$ and $\anode_4$ are of sort $\inportsort$.
	\end{itemize}
\end{definition}


\newcommand{\elemgate}{elementary gate\xspace}
\newcommand{\elemcirc}{elementary circuit\xspace}

We define the subclass of \emph{{\elemcirc}s}, which are compositions of so-called \emph{{\elemgate}s}, identity graphs and {\swapg}s.

\begin{definition}
	A graph $\agraph$ is an \emph{\elemgate} if it is in \Circuits and the following conditions hold:
	\begin{itemize}
		\item $\nodes{\agraph}$ contains a single node of sort $\gatesort{n}{m}$, for some $n$ and $m$.
		\item If $\mkedge{\anode_i}{\anodeB_i}$ and $\mkedge{\anode_j}{\anodeB_j}$ are edges such that $\anodeB_i$ occurs before $\anodeB_j$ in  $\inputsout{\agraph}$, 
		then $\labs{\agraph}(\anode_i) < \labs{\agraph}(\anode_j)$. 
		\item If $\mkedge{\anode_i}{\anodeB_i}$ and $\mkedge{\anode_j}{\anodeB_j}$ are edges such that $\anode_i$ occurs before $\anode_j$ in $\inputsin{\agraph}$, then 
		 $\labs{\agraph}(\anodeB_i) < \labs{\agraph}(\anodeB_j)$.
	\end{itemize}
	A graph $\agraph$ is an \emph{\elemcirc} if it is obtained by sequential and parallel compositions from {\elemgate}s, {\idtygraph}s and {\swapg}s.
\end{definition}

\begin{example}
	Graphs $\agraph_1$, $\agraph_2$ and $\agraph_3$ in Example \ref{ex:seq-comp} are all {\elemgate}s, hence $\seqcomp{\agraph_3}{(\agraph_2 \parcomp \agraph_1)}$ is an \elemcirc.
\end{example}

The class of {\elemcirc}s is general enough to encode several classes of graphs that are used in diagrammatic reasoning, including the orthogonal diagrams of \cite{lafont-ortho}
or the ZX calculus \cite{Coecke_2011},
 thus making it to use the proof procedure of Section \ref{sec:pf-proc} to automatically verify formulas involving such graphs.

\section{Discussion}

We have defined an extension of the Superposition calculus to a class of graphs that is general enough to encode the graph-based languages that are used to perform diagrammatic reasoning on quantum algorithms and protocols. This calculus is complete, meaning that if a graph formula is unsatisfiable, then it is guaranteed that the calculus will generate $\bot$.  We now intend to adapt an implementation of the Saturation calculus to design a tool on which we will be able to evaluate the efficiency of the calculus and investigate how it can be improved for specific subclasses of graphs -- for example, by defining suitable orderings for such subclasses in order to reduce the search space as much as possible. 

There are several lines of future work that would be interesting to
explore. We plan to investigate how variables representing graphs can
be introduced into the calculus. During the construction of a
refutation using a calculus that can handle such variables, these
variables would be instantiated by graphs with specific properties
that guarantee correctness, and would for example permit to
automatically synthesize circuits. 
Another promising topic is to extend graph literals with constraints. This would permit to represent algorithms in a more concise or natural way, while still being able to automatically verify these algorithms.

\paragraph{Acknowledgments.} This work has been partially funded by the \emph{``Investissements d'avenir''} (ANR-15-IDEX-02) program of
the French National Research Agency.

\bibliographystyle{abbrv}
\bibliography{gr}

\end{document}